\documentclass[superscriptaddress,longbibliography,nofootinbib,notitlepage,aps,pra,10pt]{revtex4-1}
\usepackage{iftex}
\ifPDFTeX\pdfoutput=1\fi

\usepackage{silence}
\usepackage{graphicx}
\usepackage{amsthm}
\usepackage{thmtools}
\usepackage{mathtools}
\usepackage{thm-restate}
\usepackage{amsmath}
\usepackage{amssymb}
\usepackage{bm}
\usepackage{xcolor}

\usepackage{placeins}
\usepackage{algorithm}
\usepackage{algpseudocode}
\usepackage[normalem]{ulem}
\usepackage[colorlinks]{hyperref}
\hypersetup{
	pdfstartview={FitH},
	pdfnewwindow=true,
	colorlinks=true,
	linkcolor=blue,
	citecolor=blue,
	filecolor=blue,
	urlcolor=blue}
\usepackage{tikz}
\usetikzlibrary{calc,backgrounds}
\usepackage{physics}
\usepackage{cancel}
\usepackage{multirow}
\usepackage{pgffor}
\usepackage{url}
\usepackage{hypcap}
\usepackage{dsfont}
\usepackage{soul}
\usepackage{inconsolata}

\usepackage{hyperref}
\hypersetup{pdfpagemode=UseNone}

\allowdisplaybreaks

\newcommand{\thm}[1]{\hyperref[thm:#1]{Theorem~\ref*{thm:#1}}}
\newcommand{\defn}[1]{\hyperref[defn:#1]{Definition~\ref*{defn:#1}}}
\newcommand{\lem}[1]{\hyperref[lem:#1]{Lemma~\ref*{lem:#1}}}
\newcommand{\prop}[1]{\hyperref[prop:#1]{Proposition~\ref*{prop:#1}}}
\newcommand{\fig}[1]{\hyperref[fig:#1]{Figure~\ref*{fig:#1}}}
\newcommand{\tab}[1]{\hyperref[tab:#1]{Table~\ref*{tab:#1}}}
\renewcommand{\sec}[1]{\hyperref[sec:#1]{Section~\ref*{sec:#1}}}
\newcommand{\app}[1]{\hyperref[app:#1]{Appendix~\ref*{app:#1}}}
\newcommand{\cor}[1]{\hyperref[cor:#1]{Corollary~\ref*{cor:#1}}}
\newcommand{\obs}[1]{\hyperref[obs:#1]{Observation~\ref*{obs:#1}}}
\newcommand{\nn}{\nonumber \\}
\newcommand{\append}[1]{\hyperref[append:#1]{Appendix~\ref*{append:#1}}}

\newtheorem{remark}{Remark}

\renewcommand{\ket}[1]{|#1\rangle}
\renewcommand{\bra}[1]{\langle #1|}

\usepackage[capitalise]{cleveref} 

\newtheorem{theorem}{Theorem}
\newtheorem{definition}[theorem]{Definition}

\newtheorem{lemma}[theorem]{Lemma}

\newtheorem{corollary}[theorem]{Corollary}

\newcommand{\MQ}{\affiliation{%
School of Mathematical and Physical Sciences,
Macquarie University, Sydney, NSW 2109, AU} }

\newcommand{\MQS}{\affiliation{%
Molecular Quantum Solutions ApS, Blegdamsvej 17, 2100 Copenhagen, Denmark} }

\newcommand{\SQA}{\affiliation{Sydney Quantum Academy, Sydney, NSW, Australia} }

\newcommand{\UMD}{\affiliation{%
Joint Center for Quantum Information and Computer Science, University of Maryland, College Park, MD 20742, USA}}

\newcommand{\UTS}{\affiliation{%
Centre for Quantum Software and Information,
University of Technology Sydney, Sydney, NSW 2007, AU}}

\newcommand{\Diraq}{\affiliation{%
Diraq, Sydney, New South Wales, Australia}}

\newcommand{\CQ}{\affiliation{%
ContinoQuantum, Sydney, NSW 2093, AU}}

\newcommand{\PKU}{\affiliation{
Beijing International Center for Mathematical Research,
Peking University, Beijing, China}}

\makeatletter
\renewcommand{\l@subsection}[2]{}
\renewcommand{\l@subsubsection}[2]{}
\makeatother

\begin{document}
\title{Adiabatic tracking and geometric schedules for classical and quantum Gibbs-state preparation}

\author{Pedro C. S. Costa}\email[Corresponding author: ]{pcs.costa@protonmail.com}\UTS\MQ\CQ
\author{Mauro E.~S. Morales}\UMD\Diraq
\author{Dong An}\PKU
\author{Kaur Kristjuhan}\MQS\MQ\SQA

\begin{abstract}
We develop a discrete adiabatic framework for Gibbs-state preparation in classical and quantum optimisation algorithms. For finite time-inhomogeneous reversible Markov chains, we derive a tracking theorem that bounds the distance between the evolving distribution and the instantaneous stationary state. The bound combines contraction of accumulated error with the new error generated as the stationary distribution changes. For Gibbs paths, it yields explicit step-complexity bounds and an effective tracking metric determined by the spectral gap and equilibrium energy fluctuations. The associated constant-speed schedule minimizes the metric-energy functional and, in the small-gap regime, takes the form $\dot\beta\propto\delta(\beta)/V(\beta)$, where $V(\beta)=\sqrt{\operatorname{Var}_{\pi_\beta}(E)}$. We then apply the discrete quantum adiabatic theorem to temperature-dependent Szegedy quantum walks, whose eigenphase gap satisfies $\Delta_w=\Theta(\sqrt{\delta_+})$, with $\delta_{+}$ the underlying Markov kernel’s right spectral gap. Under uniform bounds on the discrete operator variation and suitable spectral and path-integral assumptions, gap-adapted schedules give a sufficient fixed-accuracy cost of $O(\delta_{+,\min}^{-1/2})$. This gives a quadratic improvement in gap dependence relative to the classical inverse-gap relaxation scale under the stated assumptions. Numerical tests on a solvable two-state model and Sherrington–Kirkpatrick instances compare tracking and final preparation errors and investigate the transfer of schedules from smaller to larger systems.
\end{abstract}

\maketitle
\setcounter{tocdepth}{1}
\tableofcontents

\section{Introduction}

Gibbs-state preparation is a central task in statistical
physics, sampling, and optimisation
~\cite{kirkpatrick1983optimization,SommaBoixoBarnumKnill2008}. During annealing, the dynamics must relax towards a
stationary distribution that itself changes with temperature. An effective schedule must therefore
balance the local relaxation rate with the variation of the instantaneous Gibbs state. We investigate this
interplay through a discrete tracking theorem for classical Markov chains and its connection to coherent
quantum-walk evolution, with particular emphasis on the role of spectral information in annealing schedules.

The connection to adiabatic evolution provides a common
framework for these two settings. In quantum mechanics,
adiabatic theorems describe how a system initially prepared
in an eigenstate of a slowly varying Hamiltonian remains
close to the corresponding instantaneous eigenspace
~\cite{BornFock1928,Kato1950,JansenRuskaiSeiler2007}.
The required evolution time depends on the spectral gap
and the variation of the Hamiltonian along the path.
Consequently, nonuniform schedules can slow the evolution
near regions of small spectral separation. Discrete
adiabatic theorems extend this framework to families of
unitary operators~\cite{Costa2022QLSP,AnCostaBerry2025},
providing a natural description of parameter-dependent
quantum circuits and quantum walks.

A closely related problem arises for a family of finite
Markov kernels $\{\mathcal W_s\}_{s\in[0,1]}$ with
stationary distributions $\{\pi_s\}_{s\in[0,1]}$,
satisfying $\mathcal W_s\pi_s=\pi_s$. The stationary
distribution corresponds to the eigenvalue one, while
deviations from stationarity contract according to the
relaxation spectrum. One may therefore ask whether
an initial distribution $\pi_{s_i}$ can be transported
along a sequence of changing kernels while remaining
close to the instantaneous stationary distribution.

For a fixed Markov kernel, convergence towards
equilibrium is conventionally characterised by the
mixing time, which estimates the number of iterations
required to approach the stationary distribution within
a prescribed accuracy
~\cite{LevinPeresWilmer2017,AldousFill2002,
MontenegroTetali2006}. For reversible chains, the
corresponding bounds depend on the relaxation spectrum
and the initial distance from equilibrium. However,
fixed-kernel mixing estimates do not prescribe how the
parameter should change along an annealing path.
Classical adiabatic results address time-dependent
Markov chains and provide guarantees for terminal
preparation and, in stable adiabatic formulations,
tracking along the evolution
~\cite{Kovchegov2010,BradfordKovchegov2011,
bradford2012stable,bradford2025stable}.
Here, our objective is to retain the local contraction
and stationary-distribution increments explicitly in
a finite-step estimate suitable for schedule design.

Simulated annealing provides a natural application.
In this setting, Markov-chain Monte Carlo dynamics
approximately prepare Gibbs distributions while
the inverse temperature is gradually increased.
The annealing schedule, and often the number of
updates performed at each temperature, are chosen
heuristically. For Gibbs distributions, the motion
of the instantaneous stationary state is determined
by equilibrium energy fluctuations, while relaxation
depends on the spectral properties of the underlying
Markov kernel. Their interplay suggests schedules
adapted simultaneously to thermodynamic variation
and local relaxation.

We derive a discrete adiabatic tracking theorem for
finite time-inhomogeneous Markov chains that separates
the contraction of previously accumulated errors from
the error introduced by each stationary-distribution
increment. For reversible chains, the bound retains
the instantaneous spectral gaps and changes of
weighted norm. Its Gibbs specialisation relates
stationary-state variation to the energy variance
and leads to a geometric description of annealing
schedules. In particular, a local approximation of
the tracking bound defines an effective metric and
a corresponding constant-speed schedule. The resulting
variational optimum concerns this effective functional
rather than the exact finite-time preparation error.

The quantum counterpart offers the possibility of
improving the classical gap dependence through
spectral amplification. Szegedy's quantisation of
reversible Markov chains yields a quantum walk with
an eigenphase gap
$\Delta_w=\Theta(\sqrt{\delta})$, where $\delta$
is the classical spectral gap~\cite{Szegedy2004}.
Somma et al.~\cite{SommaBoixoBarnumKnill2008}
exploited this mechanism to establish a quadratic
improvement in inverse-gap scaling for quantum
simulated annealing. Their construction, however,
requires additional operations based on quantum
phase estimation and quantum-Zeno-type procedures
to approximately project onto successive coherent
Gibbs states. Subsequent works developed explicit
quantum-walk implementations and investigated their
heuristic performance
~\cite{Lemieux2020,sanders2020compilation}.
Despite these theoretical developments, the extent
to which the predicted quadratic gap advantage
translates into practical improvements for
combinatorial optimisation remains an open question.

This difficulty was highlighted in the earlier
numerical investigation of classical and quantum
optimisation heuristics
~\cite{costa2025assessing}, where the anticipated
quadratic quantum advantage was not observed for
the Sherrington--Kirkpatrick (SK) model. These
observations motivate a closer examination of
whether spectral amplification alone is sufficient
to improve the complexity of quantum annealing.
In particular, while
$\Delta_w=\Theta(\sqrt{\delta})$, the discrete
adiabatic error also depends on the variation
of the quantum-walk operator and the distribution
of evolution steps along the annealing path.
Thus, a generic schedule need not inherit the
quadratic gap improvement suggested by the
eigenphase spectrum.

We investigate this question using direct sequences
of temperature-dependent Szegedy walks, without
explicit intermediate projections onto coherent
Gibbs states. Applying the discrete quantum
adiabatic theorem
~\cite{Costa2022QLSP,AnCostaBerry2025},
we analyse the finite differences of the walk
operators and their dependence on the local
eigenphase gap. We identify sufficient regularity
and gap-adaptation conditions under which the
adiabatic runtime recovers inverse-linear
eigenphase-gap dependence, corresponding to
inverse-square-root dependence on the minimum
classical spectral gap. This establishes a
schedule-dependent mechanism for retaining the
spectral amplification in the adiabatic complexity
bound, without claiming a general practical
quantum advantage.

We complement the theoretical analysis with numerical experiments on the Sherrington--Kirkpatrick (SK) model. For classical Markov dynamics, we compare annealing schedules in terms of pathwise tracking errors and terminal Gibbs-state preparation. For quantum walks, we investigate gap-adapted schedules and their performance in reaching prescribed Gibbs-preparation and optimisation thresholds, examining the corresponding quantum-walk step requirements. These finite-size experiments illustrate the role of spectral information in schedule construction but do not establish the quantum regularity conditions uniformly across an asymptotic family of SK instances.

The paper is organized as follows.
\Cref{sec:main_find,sec:related_work} summarise
the main contributions and related work.
The classical framework, tracking theorem,
and Gibbs-state schedules are developed in
\cref{sec:prel,sec:disct_track,sec:applic_Gibbs},
with numerical tests in \cref{sec:Num_valid_cl}.
The discrete quantum adiabatic framework,
Szegedy walks, and gap-adapted schedules are
presented in \cref{sec:int_disc,sec:QW_S,sec:gap_adap}.
The corresponding quantum numerical results
are discussed in
\cref{sec:numerical_quantum_scaling}.
Finally, \cref{sec:conc} presents our conclusions.

\section{Main findings}
\label{sec:main_find}

Our main contribution is an explicit discrete tracking
bound that retains local relaxation and stationary-state
variation, together with a geometric schedule obtained
from its smooth Gibbs specialisation. The classical
bound applies directly to the discrete evolution,
whereas the metric describes a leading local tracking
envelope. The quantum result establishes sufficient
gap-scaling conditions under uniform regularity and
path-integral assumptions. Numerical tests examine
the resulting schedules without identifying their
finite-size performance with an asymptotic complexity
guarantee.

For a sequence of Markov kernels
$\{\mathcal W_j\}_{j=0}^T$ with stationary
distributions $\{\pi_j\}_{j=0}^T$, we consider
$p_{j+1}=\mathcal W_{j+1}p_j$, with $p_0=\pi_0$.
Defining $e_j=p_j-\pi_j$, we obtain the exact
recursion
\begin{equation}
e_{j+1}
=
\mathcal W_{j+1}e_j
+
\mathcal W_{j+1}(\pi_j-\pi_{j+1}).
\end{equation}
The first term describes the relaxation of previously
accumulated error, while the second represents the
error introduced by the variation of the instantaneous
stationary distribution.

For reversible chains, contraction is naturally
expressed in the instantaneous weighted norm
\begin{equation}
\|q\|_{\pi_k^{-1}}^2
=
\sum_{x\in\Omega}
\frac{q(x)^2}{\pi_k(x)}.
\end{equation}
Since this norm changes along the path, the tracking
bound depends on both the local spectral gaps
$\delta_k$ and norm-comparison factors $C_k$.
For lazy reversible kernels, introducing the
effective damping rate
\begin{equation}
e^{-\Gamma_k^{\mathrm{exact}}}
=
(1-\delta_k)C_{k-1},
\end{equation}
the condition $\Gamma_k^{\mathrm{exact}}>0$
ensures contraction of the previously accumulated
error after accounting for the change of norm.
The total tracking error additionally depends
on the stationary-state increments.

For Gibbs distributions,
\begin{equation}
\pi_\beta(x)
=
\frac{e^{-\beta E(x)}}{Z(\beta)},
\end{equation}
we have
\begin{equation}
\|\partial_\beta\pi_\beta\|_{\pi_\beta^{-1}}
=
\sqrt{\operatorname{Var}_{\pi_\beta}(E)}.
\end{equation}
This relation leads to a local tracking metric
whose density is proportional to
$(1-\delta)V/\delta$, where
$V(\beta)=\sqrt{\operatorname{Var}_{\pi_\beta}(E)}$.
In the small-gap regime, neglecting norm-comparison
corrections, the corresponding constant-speed
schedule takes the form
\begin{equation}
\dot{\beta}(t)
\propto
\frac{\delta(\beta)}
{\sqrt{\operatorname{Var}_{\pi_\beta}(E)}}.
\end{equation}
Constant metric speed minimises the associated
metric-energy functional for fixed endpoints
and duration. It does not necessarily minimise
the exact maximum tracking error or the final
total-variation distance.

We next consider the quantum analogue using the
Szegedy walk construction employed in
Ref.~\cite{Lemieux2020},
\begin{equation}
U_w(\beta)
=
R W^\dagger(\beta)S W(\beta).
\end{equation}
If $\delta(\beta)=1-\lambda_1(\beta)$ denotes
the spectral gap of the underlying reversible
Markov chain, the corresponding eigenphase gap
satisfies
\begin{equation}
\Delta_w(\beta)
=
\arccos[1-\delta(\beta)]
=
\Theta\!\left(\sqrt{\delta(\beta)}\right).
\end{equation}

This square-root spectral amplification does not,
by itself, guarantee a quadratic improvement
in discrete adiabatic complexity. The discrete
quantum adiabatic theorem
~\cite{Costa2022QLSP,AnCostaBerry2025}
also depends on the first- and second-order
finite differences of the walk operator,
$D^{(1)}U_w$ and $D^{(2)}U_w$, controlled
by coefficients $c_1(s)$ and $c_2(s)$.

For the present construction, the first-order
variation satisfies
\begin{equation}
c_1(s)
\lesssim
2|\beta'(s)|
\|\partial_\beta W(\beta)\|,
\end{equation}
whereas the second-order bound additionally
involves $\beta''(s)$, $(\beta'(s))^2$,
and derivatives of $W(\beta)$.

We consider gap-adapted schedules designed
to satisfy
\begin{equation}
c_1(s)
\lesssim
C\Delta_w(s)^p,
\qquad 1<p<2,
\end{equation}
for a constant $C$ independent of the minimum
spectral gap. A corresponding schedule is
\begin{equation}
\dot{\beta}(s)
\propto
\frac{\Delta_w(\beta)^p}
{\|\partial_\beta W(\beta)\|}
=
\Theta\!\left(
\frac{\delta(\beta)^{p/2}}
{\|\partial_\beta W(\beta)\|}
\right),
\end{equation}
provided the denominator is nonzero and the
resulting schedule is well defined.

Under the uniform discrete,
spectral-comparison, quadrature, and path-integral
hypotheses of \cref{thm:quantum_gap_scaling},
including control of the second-order variation,
the sufficient fixed-accuracy number of quantum-walk
applications satisfies
\begin{equation}
T_{\rm Q}
=
O\left(\Delta_{w,\min}^{-1}\right)
=
O\left(\delta_{\min}^{-1/2}\right).
\end{equation}

This establishes a conditional quadratic
improvement in the inverse-gap dependence
relative to the classical inverse-gap
relaxation scale, rather than a general
reduction in total computational cost.
The schedule conditions are sufficient;
no necessity result for all possible
annealing schedules is claimed.

\section{Related Work}\label{sec:related_work}

\subsubsection{Classical adiabatic and tracking theorems}

Kovchegov~\cite{Kovchegov2010} and Bradford and Kovchegov~\cite{BradfordKovchegov2011} study classical adiabatic preparation for time-inhomogeneous Markov chains. They consider transition kernels gradually deformed from an initial kernel $P_{\mathrm{initial}}$ to a final kernel $P_{\mathrm{final}}$. They define the adiabatic time $T_\epsilon$ as the minimum number of steps required for the distribution produced by this inhomogeneous evolution to be within total-variation distance $\epsilon$ of the stationary distribution $\pi_f$ of $P_{\mathrm{final}}$, uniformly over the initial distribution. For a linear interpolation between the initial and final kernels, they show that choosing the total runtime $T$ such that
\begin{equation}
T \geq 
C \frac{t_{\mathrm{mix}}(\epsilon/2)^2}{\epsilon}
,    
\end{equation}
is sufficient to guarantee final total-variation error at most $\epsilon$ uniformly over the initial distribution. Here $C>0$ is an appropriate constant  and $t_{\mathrm{mix}}$ is the mixing time of the final kernel. This provides an upper bound on the minimal adiabatic time $T_{\epsilon}$. They further extend the analysis to smoother, entry-dependent interpolations and show how the scaling of the adiabatic time depends on the behavior of the interpolation near its endpoint. As an application, they study adiabatically varying Glauber dynamics for the Ising model, for which the stationary measures are Gibbs distributions associated with the instantaneous Hamiltonians. Their results provide a final-time guarantee expressed in terms of the mixing properties of the final Markov kernel. For reversible chains, inserting standard mixing-time estimates yields a sufficient runtime bound with quadratic inverse-gap dependence, up to logarithmic factors and dependence on the smallest stationary probability. The analysis prescribes a chosen interpolation between the endpoint kernels, but it does not derive a schedule adapted to the instantaneous gaps or to the local variation of the stationary distributions along the interpolation.

Stable adiabatic-time results go beyond terminal preparation. Bradford, Kovchegov, and Nguyen control tracking throughout linear interpolations~\cite{bradford2012stable}, and Bradford treats general continuous paths~\cite{bradford2025stable}. Their guarantees use global mixing-time and path-regularity quantities. Our \cref{thm:tracking_changing_norm} instead retains individual stationary increments, contraction factors, and norm changes in an instantaneous weighted-norm estimate, which also implies a total-variation bound. It applies to a prescribed discrete sequence without requiring a smooth interpolation. Reversibility is an additional assumption enabling the contraction argument; the result is not a generalization of all stable adiabatic-time theorems.

For smooth Gibbs paths, \cref{lemma:effective_metric} extracts the leading local tracking coefficient and uses its square to define a metric for schedule construction. This distinction concerns the information retained in the bound and its subsequent geometric use. A comparison of inverse-gap powers alone does not establish a uniform improvement over earlier runtime bounds: stationary-state variation, kernel-path regularity, norm comparisons, and state-space dependence must be accounted for on the same problem family.

\subsubsection{Thermodynamic geometry and minimum-dissipation protocols}

A related approach to the design of finite-time protocols is provided by thermodynamic geometry, which studies the dissipation generated when a system is driven by external controls. The analysis of such finite-time nonequilibrium processes naturally induces a Riemannian metric that quantifies the thermodynamic cost of moving between nearby equilibrium states. Early formulations introduced thermodynamic length through the Hessians of thermodynamic potentials, while later microscopic formulations connected this geometry to Fisher information and equilibrium fluctuations \cite{Crooks2007}. Subsequently in \cite{sivak2012thermodynamic} the authors derived a dynamical thermodynamic metric directly from linear-response theory.

Consider a system with energy $E(x,\boldsymbol{\lambda})$, where $x$ denotes a microstate and $\boldsymbol{\lambda}$ is a collection of externally controlled parameters. At fixed bath inverse temperature $\beta$, the equilibrium distribution is
\begin{equation}
\pi(x\mid\boldsymbol{\lambda}) =
\frac{e^{-\beta E(x,\boldsymbol{\lambda})}}{
Z(\boldsymbol{\lambda})},   
\end{equation}
The generalized force conjugate to the control parameter $\lambda_i$ is defined by
\begin{equation}
X_i = -\frac{\partial E}{\partial\lambda_i}.   
\end{equation}
When the control parameters vary in time, the system generally does not remain exactly in the instantaneous equilibrium state. The actual average generalized force therefore differs from the equilibrium average evaluated at the same instantaneous parameter value. This force lag produces an excess power beyond the power required in a quasistatic transformation.

Within linear-response theory, the force lag is determined by equilibrium time correlations of the conjugate forces. Define
\begin{equation}
\Sigma_{ij}^{(\boldsymbol{\lambda})}(t)=
\left\langle
\delta X_j(0)\delta X_i(t)
\right\rangle_{\boldsymbol{\lambda}}.    
\end{equation}
The generalized friction tensor is then
\begin{equation}
\zeta_{ij}(\boldsymbol{\lambda})
= \beta
\int_0^\infty
\Sigma_{ij}^{(\boldsymbol{\lambda})}(t)\,dt. \end{equation}
This expression is often referred to as the Kirkwood formulation of the friction tensor. It combines the magnitude of equilibrium fluctuations with the timescale over which those fluctuations decay. The friction is therefore large when the conjugate forces fluctuate strongly, when the system relaxes slowly, or both.

For protocols that vary slowly relative to the relaxation timescale, the leading excess power takes the local quadratic form
\begin{equation}
P_{\mathrm{ex}}(t)
= \dot{\boldsymbol{\lambda}}(t)^{\mathsf T}
\boldsymbol{\zeta}(\boldsymbol{\lambda}(t))
\dot{\boldsymbol{\lambda}}(t).  
\end{equation}
The corresponding excess work is
\begin{equation}
W_{\mathrm{ex}} = \int_0^{\mathcal T}
P_{\mathrm{ex}}(t)\,dt
= \int_0^{\mathcal T}
\dot{\boldsymbol{\lambda}}^{\mathsf T}
\boldsymbol{\zeta}(\boldsymbol{\lambda})
\dot{\boldsymbol{\lambda}}\,dt.    
\end{equation}

The excess work is the difference between the work performed during the finite-time protocol and the work that would be required during an infinitely slow, quasistatic transformation. If the system is allowed to equilibrate completely after the protocol, the excess work is ultimately dissipated as heat.

The friction tensor is symmetric and positive semidefinite under the equilibrium reversibility assumptions. Wherever it is positive definite, it defines a Riemannian metric on the control-parameter space; otherwise, it defines a degenerate quadratic form.
\begin{equation}
\mathcal L_{\zeta} =
\int_0^{\mathcal T} \sqrt{
\dot{\boldsymbol{\lambda}}^{\mathsf T}
\boldsymbol{\zeta}(\boldsymbol{\lambda})
\dot{\boldsymbol{\lambda}}
}\,dt.    
\end{equation}
For a fixed path and duration $\mathcal T$, the excess work obeys
\begin{equation}
W_{\mathrm{ex}} \geq
\frac{\mathcal L_{\zeta}^{\,2}}{\mathcal T}.   \end{equation}
This lower bound follows from the Cauchy–Schwarz inequality. It is saturated when the protocol is traversed at constant thermodynamic speed, or equivalently when the excess power is constant. For a single control parameter $\lambda$, the excess power reduces to
\begin{equation}
P_{\mathrm{ex}}(t) =
\zeta(\lambda(t))\dot{\lambda}(t)^2.
\end{equation}
The constant-excess-power condition then gives
\begin{equation}
 \dot{\lambda}_{\mathrm{opt}}(t) \propto
\frac{1}{\sqrt{\zeta(\lambda(t))}}. 
\end{equation}
The minimum-dissipation protocol therefore proceeds slowly in regions of large thermodynamic friction and more rapidly in regions of small friction. When the friction tensor is positive definite and the path can vary in a multidimensional control space, smooth minimum-dissipation protocols follow affinely parametrized geodesics of the friction metric.

The friction tensor can also be decomposed into static and dynamical contributions. In \cite{sivak2012thermodynamic} the authors express it in terms of the Fisher-information matrix and an integral relaxation-time matrix. For a Gibbs family, the Fisher-information matrix is
\begin{equation}
I_{ij}(\boldsymbol{\lambda}) =
\left\langle
\frac{\partial\log\pi}{\partial\lambda_i}
\frac{\partial\log\pi}{\partial\lambda_j}
\right\rangle_{\boldsymbol{\lambda}}.  
\end{equation}
When the energy is linear in the control parameters, or under other conditions discussed in \cite{sivak2012thermodynamic}, this becomes
\begin{equation}
I_{ij}(\boldsymbol{\lambda}) =
\beta^2
\left\langle
\delta X_i\delta X_j
\right\rangle_{\boldsymbol{\lambda}}.  
\end{equation}
The friction tensor may consequently be written schematically as $\boldsymbol{\zeta}\sim \boldsymbol{I}\circ\boldsymbol{\tau}
$, where $\circ$ denotes the entrywise product and $\boldsymbol{\tau}$ contains the integrated relaxation times. Fisher information measures the local statistical distinguishability of neighboring equilibrium states, whereas the relaxation-time factor measures how fast the dynamics responds to a change in the controls.

For the inverse-temperature Gibbs family $\pi_\beta(x)$, the score function is $\partial_\beta\log\pi_\beta(x) =
-\left( E(x)-\langle E\rangle_\beta
\right)$. The Fisher information in the inverse-temperature direction is therefore
\begin{equation}
I_{\beta\beta} = \sum_x \pi_\beta(x) \left( \partial_\beta\log\pi_\beta(x)
\right)^2 = \operatorname{Var}_{\pi_\beta}(E).
\end{equation}
Thus, the energy variance determines the static thermodynamic geometry of the Gibbs curve parametrized by inverse temperature. Including the dynamics introduces an integrated energy-correlation time,
\begin{equation}
\tau_E(\beta) = \frac{ \displaystyle \int_0^\infty \left\langle
\delta E(0)\delta E(t) \right\rangle_{\pi_\beta} \,dt}{
\operatorname{Var}_{\pi_\beta}(E)
}.    
\end{equation}
Up to prefactors that depend on the precise convention used to treat temperature as a control coordinate, the friction in the inverse-temperature direction has the form
\begin{equation}
\zeta_{\beta\beta}(\beta) \sim
\operatorname{Var}_{\pi_\beta}(E)\tau_E(\beta).   
\end{equation}
The corresponding constant-excess-power schedule is
\begin{equation}
\dot{\beta}_{\mathrm{opt}}(t) \propto
\frac{1}{\sqrt{ \operatorname{Var}_{\pi_{\beta(t)}}(E)
\tau_E(\beta(t))
}}. 
\end{equation}
This schedule allocates more time to regions where the Gibbs distribution changes fast or where energy fluctuations relax slowly. The construction provides a minimum-dissipation principle within the linear-response regime; its guarantee concerns excess work rather than a direct bound on the distance between the evolving state and the instantaneous equilibrium distribution.

\subsubsection{Quantum adiabatic theorems and quantized Markov chains}

The quantum analogue is based on replacing a reversible Markov kernel by a unitary quantum walk whose spectrum encodes that of the classical chain. For a reversible Markov kernel $\mathcal W$ with eigenvalues
\begin{equation}
1=\lambda_0>\lambda_1\geq\lambda_2\geq\cdots,
\end{equation}
the Szegedy construction associates a quantum walk whose nontrivial eigenvalues occur in conjugate pairs $e^{\pm i\theta_k}$, with $\cos\theta_k=\lambda_k$, up to convention-dependent factors in the definition of the walk~\cite{Szegedy2004}. If $\delta=1-\lambda_1$ denotes the spectral gap of the Markov chain, the corresponding eigenphase gap satisfies
\begin{equation}
\Delta_w = \arccos(1-\delta)= \Theta(\sqrt{\delta}),
\end{equation}
for $\delta\ll1$. This square-root relation is the basic spectral amplification associated with quantized Markov chains.

This mechanism underlies the quantum simulated-annealing algorithm in~\cite{SommaBoixoBarnumKnill2008}, which was introduced for combinatorial optimization. Their construction combines quantized Markov chains with phase-estimation and quantum-Zeno-type procedures to approximately project onto a sequence of coherent Gibbs states. Its dependence on the minimum gap of the classical stochastic matrices is proportional to
\begin{equation}
\frac{1}{\sqrt{\delta_{\min}}},
\end{equation}
compared with the $1/\delta_{\min}$ dependence associated with classical relaxation. The quadratic improvement therefore originates from the square-root eigenphase gap, while phase estimation and repeated approximate projections are used to follow the Gibbs state along the annealing path. Related work ~\cite{SommaBoixo2013} established a general spectral-gap amplification theorem requiring frustration-free Hamiltonians. In the present setting, this property is naturally satisfied by the parent Hamiltonian associated with a reversible Markov chain, even when the original combinatorial optimization Hamiltonian is frustrated.

Several subsequent works investigated the practical realization of quantum-walk-based optimization algorithms. In particular, Ref.~\cite{Lemieux2020} developed explicit circuits for the Szegedy quantization of Metropolis--Hastings dynamics and numerically studied quantum-walk heuristics for discrete optimization, finding evidence of polynomial improvements in selected heuristic settings. However, these studies did not directly demonstrate the predicted quadratic improvement in inverse-gap scaling for concrete families of combinatorial optimization problems. To our knowledge, numerical evidence establishing this quadratic scaling advantage for such problems remains lacking.

A separate line of work concerns discrete quantum adiabatic theorems for slowly varying families of unitary operators. Instead of a continuous Hamiltonian evolution, one considers a product of unitaries $U(s_j)$ sampled along a path and studies whether an initial eigenstate remains close to the corresponding instantaneous eigenspace. The resulting bounds depend both on the instantaneous eigenphase separation and on the finite differences of the unitary family~\cite{Costa2022QLSP,AnCostaBerry2025}. Such results provide a framework for understanding direct products of slowly varying quantum walks without introducing an explicit projection onto the instantaneous state after each step.

For temperature-dependent Szegedy walks, the relation
\begin{equation}
\Delta_w(\beta)
=
\Theta\!\left(\sqrt{\delta(\beta)}\right),
\end{equation}
suggests the possibility of retaining the quadratic spectral advantage during an annealing process. Unlike phase-estimation- or quantum-Zeno-based constructions, however, a direct sequence of varying unitary walks has no mixing mechanism that automatically drives the state towards the instantaneous Gibbs state. Its behaviour must instead be understood through adiabatic following, where both the eigenphase gap and the variation of the walk along the path become relevant.

\section*{Part A: Classical adiabatic tracking}

\section{Preliminaries and Technical Setup}\label{sec:prel}

Let $\Omega$ be a finite state space with $|\Omega|=N$. We represent
probability distributions as column vectors and use column-stochastic
transition matrices. Thus,
\begin{equation}
\mathcal W_{yx} = \Pr(X_{t+1}=y\mid X_t=x), \qquad
\sum_{y\in\Omega}\mathcal W_{yx}=1,
\end{equation}
and probability distributions evolve according to $p_{t+1} = \mathcal Wp_t$. To compare a probability distribution with a strictly positive reference
distribution $\pi$, it is useful to introduce an inner product weighted by
the inverse of $\pi$. Define
\begin{equation}
D_\pi\coloneqq \operatorname{diag}
\bigl(\pi(x)\bigr)_{x\in\Omega}.
\end{equation}
For $u,v\in\mathbb R^{|\Omega|}$, define
\begin{equation}
\langle u,v\rangle_{\pi^{-1}}
\coloneqq
 u^\top D_\pi^{-1}v = \sum_{x\in\Omega}
\frac{u(x)v(x)}{\pi(x)}.
\end{equation}
The associated norm is
\begin{equation}
\|u\|_{\pi^{-1}} \coloneqq
\sqrt{ \langle u,u\rangle_{\pi^{-1}}} = \left( \sum_{x\in\Omega} \frac{u(x)^2}{\pi(x)}
\right)^{1/2}.
\end{equation}

On the other hand, sampling guarantees are often stated in terms of the total
variation distance,
\begin{equation}
\label{eq:TV_er}
\|p-\pi\|_{\mathrm{TV}} = \frac{1}{2}
    \sum_{x\in\Omega}|p(x)-\pi(x)|.
\end{equation}
However, this can be linked with the weighted norm as follows,
\begin{align}
\sum_{x\in\Omega}|p(x)-\pi(x)|
&= \sum_{x\in\Omega} \sqrt{\pi(x)}
\frac{|p(x)-\pi(x)|}{\sqrt{\pi(x)}} \nonumber\\
&\leq \left( \sum_{x\in\Omega}\pi(x) \right)^{1/2}
\left(\sum_{x\in\Omega} \frac{(p(x)-\pi(x))^2}{\pi(x)}
    \right)^{1/2},
\end{align}
where we have used the Cauchy-Schwarz inequality. Since $\sum_{x\in\Omega}\pi(x)=1$, it follows that
\begin{equation}
\|p-\pi\|_{\mathrm{TV}} \leq \frac{1}{2}
\|p-\pi\|_{\pi^{-1}}.
\end{equation}
Thus, any bound in the weighted norm $\|\cdot\|_{\pi^{-1}}$ immediately implies a standard total-variation sampling guarantee.

We now introduce the main definitions used throughout the classical analysis.
       
\begin{definition}[Stationary distribution]
A probability vector $\pi$ is stationary for $\mathcal{W}$ if
\begin{equation}
 \mathcal{W} \pi=\pi .    
\end{equation}
\end{definition}

\begin{definition}[Zero-mass subspace]
Let $\Omega$ be a finite state space with $|\Omega|=N$, and let $\mathbf 1\in\mathbb R^N$ denote the all-ones vector. We define the zero-mass subspace by
\begin{equation}
\mathcal H_0 \coloneqq \left\{ q\in\mathbb R^N:
\mathbf 1^\top q=0\right\}.
\end{equation}
Equivalently, $q\in\mathcal H_0$ if and only if
\begin{equation}
    \sum_{x\in\Omega}q(x)=0.
\end{equation}
\end{definition}

The subspace $\mathcal H_0$ contains the difference between any two probability distributions on $\Omega$. Indeed, if $p$ and $q$ are probability distributions, then
\begin{align}
\mathbf 1^\top(p-q) &= \mathbf 1^\top p-\mathbf 1^\top q \nonumber\\ 
&= 1-1 \nonumber\\
&= 0.
\end{align}
Hence, $p - q\in\mathcal H_0$.

\begin{lemma}[Contraction gap for reversible chains]
\label{lemma:contractiongap}
Let $\mathcal W$ be a reversible Markov kernel with stationary distribution
$\pi$. Assume that $\mathcal W$ is irreducible and aperiodic, so that
$\pi(x)>0$ for all $x\in\Omega$ and the eigenvalue $1$ is simple. Let $ 1=\lambda_0,\lambda_1,\ldots,\lambda_{N-1},$ be the eigenvalues of $\mathcal W$, and define $\lambda_*\coloneqq \max_{j\geq 1}|\lambda_j|$. Then, for every $r\in\mathcal H_0$,
\begin{equation}
\|\mathcal W r\|_{\pi^{-1}} \leq \lambda_*\,\|r\|_{\pi^{-1}}.
\end{equation}
In particular, if $p$ is any probability distribution and $r=p-\pi$, then
$r\in\mathcal H_0$ and
\begin{equation}
\|\mathcal W(p-\pi)\|_{\pi^{-1}} \leq \lambda_*\,\|p-\pi\|_{\pi^{-1}}.
\end{equation}
Hence the contraction gap in the norm $\|\cdot\|_{\pi^{-1}}$ is
\begin{equation}
\delta = 1-\lambda_* =1-\max_{1\leq j\leq N-1}|\lambda_j|.
\end{equation}
\end{lemma}
Proof of the Lemma above is presented in \cref{proof:contractiongap}.

\section{A Basic Discrete Adiabatic Tracking Theorem}\label{sec:disct_track}

We now prove the simplest discrete adiabatic tracking theorem.

\begin{theorem}[Discrete adiabatic tracking with changing norms]\label{thm:tracking_changing_norm}
Let $\{\mathcal W_j\}_{j=0}^T$ be a sequence of irreducible,
aperiodic, reversible column-stochastic Markov kernels on a finite state space $\Omega$, with stationary distributions $\pi_j$ satisfying $\mathcal W_j\pi_j=\pi_j$. For each $j$, define
\begin{equation}
\|q\|_j = \left( \sum_{x\in\Omega}
\frac{q(x)^2}{\pi_j(x)}
\right)^{1/2},
\end{equation}
and let
\begin{equation}
\delta_j = 1-\max_{k\geq1}
|\lambda_k(\mathcal W_j)|
\end{equation}
denote the absolute spectral gap. Then, for every $q\in\mathcal H_0$,
\begin{equation}
\|\mathcal W_jq\|_j
\leq
(1-\delta_j)\|q\|_j.
\end{equation}
Define the norm-comparison constants
\begin{equation}
\label{eq:norm_comp}
C_j \coloneqq \sup_{q\in\mathcal H_0,\;q\neq 0}
\frac{\|q\|_{j+1}}{\|q\|_j},
    \qquad
    j=0,\ldots,T-1.    
\end{equation}

Then, when $p_0=\pi_0$ and the update probability distribution is given by $p_j=\mathcal{W}_jp_{j-1}$, for $j=1,\cdots, T$, we have
\begin{equation}
 \|p_T-\pi_T\|_T \leq \sum_{j=1}^T \left[
\left( \prod_{r=j}^T(1-\delta_r) \right)
\left( \prod_{r=j}^{T-1}C_r \right)\right]
\|\pi_j-\pi_{j-1}\|_j .    
\end{equation}
which is further bounded by
\begin{equation}
\label{eq:first_b}
\|p_T-\pi_T\|_T \leq \sum_{j=1}^T
\exp\!\left(-\sum_{r=j}^T\delta_r +
\sum_{r=j}^{T-1}\log C_r
\right) \|\pi_j-\pi_{j-1}\|_j .    
\end{equation}
\end{theorem}

Notice that the theorem above, where the proof is presented in \cref{proof:track}, has the same qualitative structure as a discrete adiabatic theorem: local changes in the instantaneous invariant state are damped by the contraction accumulated from the local gaps along the remaining path.

If all local gaps are uniformly bounded below and the norm-comparasion constants are uniformly bounded above, the preceding theorem gives a global estimate.

\begin{corollary}[Uniform-gap tracking bound with changing norms]
\label{cor:uniform_gap_changing_norm}
Assume the hypotheses of \cref{thm:tracking_changing_norm}. Suppose, in
addition, that $\delta_j\geq\delta_{\min}>0$ and $C_j\leq C_{\max}$ for all relevant indices. Then
\begin{equation}
\label{eq:uniform_gap_bound}
\|p_T-\pi_T\|_T
\leq
\sum_{j=1}^T
e^{-(T-j+1)\delta_{\min}}
C_{\max}^{T-j}
\|\pi_j-\pi_{j-1}\|_j.
\end{equation}
\end{corollary}

The uniform-gap estimate can be further simplified by bounding the variation of the stationary distribution at each step by its maximal one-step value, leading to the following corollary.

\begin{corollary}[Bound in terms of the maximal one-step variation]
\label{cor:maximal_one_step_variation}
Under the hypotheses of
\cref{cor:uniform_gap_changing_norm}, define $\omega_\pi
\coloneqq
\max_{1\leq j\leq T}
\|\pi_j-\pi_{j-1}\|_j$, then
\begin{equation}
\label{eq:uniform_geometric_sum}
\|p_T-\pi_T\|_T
\leq
\omega_\pi e^{-\delta_{\min}}
\sum_{m=0}^{T-1}
\left(
C_{\max}e^{-\delta_{\min}}
\right)^m.
\end{equation}
If $C_{\max}e^{-\delta_{\min}}\neq1$, then
\begin{equation}
\label{eq:uniform_finite_bound}
\|p_T-\pi_T\|_T
\leq \omega_\pi e^{-\delta_{\min}}
\frac{ 1-\left(C_{\max}e^{-\delta_{\min}}\right)^T}{ 1-C_{\max}e^{-\delta_{\min}}
}.
\end{equation}
If $C_{\max}e^{-\delta_{\min}}=1$, then
\begin{equation}
\label{eq:uniform_critical_bound}
\|p_T-\pi_T\|_T
\leq
T\omega_\pi e^{-\delta_{\min}}.
\end{equation}
\end{corollary}

The proofs of the previous Corollaries can be found in \cref{app:proof_coroll5_6_7}. The geometric-series bound becomes particularly transparent when the combined effect of contraction and norm variation is strictly damping, i.e., the error generated at an earlier step is progressively suppressed as the evolution continues, which motivates the following effective-damping formulation.

\begin{corollary}[Stable effective-damping regime]
\label{cor:stable_effective_damping}
Under the hypotheses of
\cref{cor:uniform_gap_changing_norm}, suppose that $C_{\max}e^{-\delta_{\min}}<1$, or equivalently $\log C_{\max}<\delta_{\min}$. 
Define the uniform effective damping rate by
\begin{equation}
\label{eq:damp_rate}
\Gamma \coloneqq \delta_{\min}-\log C_{\max}>0.
\end{equation}
Then
\begin{equation}
\label{eq:effective_damping_finite_bound}
\|p_T-\pi_T\|_T \leq \omega_\pi e^{-\delta_{\min}} \frac{1-e^{-T\Gamma}}{1-e^{-\Gamma}}.
\end{equation}
In particular,
\begin{equation}
\label{eq:effective_damping_uniform_bound}
\|p_T-\pi_T\|_T \leq \omega_\pi \frac{e^{-\delta_{\min}}}{1-e^{-\Gamma}}.
\end{equation}
For small $\Gamma>0$, this gives
\begin{equation}
\label{eq:effective_damping_scaling}
\|p_T-\pi_T\|_T=O\left(\frac{\omega_\pi}{\Gamma}\right)=O\left(\frac{\omega_\pi}{\delta_{\min}-\log C_{\max}}\right).
\end{equation}
\end{corollary}
\begin{remark}[Exact and conservative damping rates]
The rate $\Gamma=\delta_{\min}-\log C_{\max}$ is a conservative uniform estimate obtained using $1-\delta_k\le e^{-\delta_k}$. The exact damping rate of the scalar recurrence is
\begin{equation}
\Gamma_k^{\mathrm{exact}}
\coloneqq-\log\!\bigl[(1-\delta_k)C_{k-1}\bigr],    
\end{equation}
with $\Gamma_k^{\mathrm{exact}}=+\infty$ when $\delta_k=1$. It satisfies $\Gamma_k^{\mathrm{exact}}
\ge\delta_k-\log C_{k-1}
\ge\Gamma.$ Thus, $\Gamma>0$ guarantees uniform contraction of the propagated-error term. This condition is sufficient but not necessary.
\end{remark}

To convert the preceding bounds into an explicit dependence on the number of steps $T$, we need to control the maximal one-step variation $\omega_\pi$ along a smooth stationary path, which can be done under the conditions of our next lemma.

\begin{lemma}[One-step variation of a smooth stationary path] \label{lemma:One_step}
Let $s\mapsto\pi_s$ be a differentiable path of probability distributions on $\Omega$. Let $s_j=j/T$,  $\pi_j=\pi_{s_j}$ and $j=0,\ldots,T$. For each $j$, let $\|\cdot\|_j$ denote the instantaneous weighted norm associated with $\pi_j$. Assume there exists a constant $L_{\pi}>0$, independent of $T$, such that
\begin{equation}
\label{eq:bound_L}
\sup_{1\leq j\leq T}
\sup_{s\in[s_{j-1},s_j]}
\|\dot\pi_s\|_j
\leq
L_{\pi}.
\end{equation}
Then the maximal one-step variation satisfies
\begin{equation}
\label{eq:one_max}
\omega_{\pi}
= \max_{1\leq j\leq T}
\|\pi_j-\pi_{j-1}\|_j
\leq
\frac{L_{\pi}}{T}.
\end{equation}
\end{lemma}
Proof of the Lemma above is presented in \cref{proof:one_step}. 

We now combine \cref{lemma:One_step} with \cref{cor:uniform_gap_changing_norm}. Previously, we found,
\begin{equation}
\|p_T-\pi_T\|_T \leq \omega_\pi e^{-\delta_{\min}}
\frac{1-e^{-T\Gamma}}{1-e^{-\Gamma}}.
\end{equation}
Now, with \cref{lemma:One_step}, $\delta_{\pi}\leq L_{\pi}/T$, we get
\begin{align}
\|p_T-\pi_T\|_T &\leq \frac{L_{\pi}}{T} e^{-\delta_{\min}}
\frac{1-e^{-T\Gamma}}{1-e^{-\Gamma}}\nonumber\\
& \leq \frac{L_{\pi}}{T} 
\frac{e^{-\delta_{\min}}}{1-e^{-\Gamma}},
\end{align}
where we also used the fact that $1-e^{-T\Gamma} \leq 1$. Thus, to guarantee $\|p_T-\pi_T\|_T\leq\varepsilon$, it is sufficient,
in the small-$\Gamma$ regime, to choose
\begin{equation}
T \geq c\, \frac{L_{\pi}}
{\varepsilon\left(\delta_{\min}-\log C_{\max}\right)},
\end{equation}
for some constant $c>0$, provided that $\delta_{\min}>\log C_{\max}$.

\section{Application to Gibbs-State Preparation}
\label{sec:applic_Gibbs}

We now specialise the theorem to Gibbs distributions. Let
$E:\Omega\to\mathbb R$ be an energy function and define
\begin{equation}
\label{eq:gibbs_dis}
\pi_\beta(x) = \frac{e^{-\beta E(x)}}{Z(\beta)},
\qquad Z(\beta) = \sum_{y\in\Omega}e^{-\beta E(y)}.
\end{equation}
A straightforward result from \cref{thm:tracking_changing_norm} is given in the corollary below. 
\begin{corollary}[Adiabatic Gibbs-state preparation in weighted norms]
\label{cor:gibbs_weighted}

Consider $\pi_j=\pi_{\beta_j}$ and
$\mathcal W_j=\mathcal W_{\beta_j}$, where each
$\mathcal W_{\beta_j}$ is an irreducible, aperiodic Markov kernel
reversible with respect to the Gibbs distribution $\pi_{\beta_j}$.
Let $\delta_j$ denote its absolute spectral gap, so that
\begin{equation}
    \|\mathcal W_{\beta_j}q\|_j
    \leq
    (1-\delta_j)\|q\|_j,
    \qquad
    q\in\mathcal H_0.
\end{equation}

If $p_0=\pi_{\beta_0}$ and
$p_{k+1}=\mathcal W_{\beta_{k+1}}p_k$, then
\begin{equation}
\label{eq:gibbs_weighted_tracking}
\|p_T-\pi_{\beta_f}\|_T \leq  \sum_{j=1}^T \exp\!\left(
    -\sum_{r=j}^T\delta_r
 + \sum_{r=j}^{T-1}\log C_r  \right)
    \|\pi_{\beta_j}-\pi_{\beta_{j-1}}\|_j .
\end{equation}
Moreover,
\begin{equation}
\|p_T-\pi_{\beta_f}\|_{\mathrm{TV}} \leq \frac{1}{2} \|p_T-\pi_{\beta_f}\|_T,
\end{equation}
and therefore
\begin{equation}
\label{eq:gibbs_weighted_tv_tracking}
\|p_T-\pi_{\beta_f}\|_{\mathrm{TV}}
    \leq
    \frac{1}{2}
    \sum_{j=1}^T
    \exp\!\left(
    -\sum_{r=j}^T\delta_r
    +
    \sum_{r=j}^{T-1}\log C_r
    \right)
    \|\pi_{\beta_j}-\pi_{\beta_{j-1}}\|_j .
\end{equation}
\end{corollary}

We now consider \cref{cor:uniform_gap_changing_norm} for the Gibbs state preparation when $\beta\in[\beta_0,\beta_f]$ is treated as a continuous inverse-temperature parameter, so then the result from \cref{lemma:One_step} can be applied. Ultimately, we are looking for $L_{\pi}$ when we have the Gibbs state preparation, which we will analyse next. 

In this scenario, we now that the discrete evolution is obtained by sampling this path at the points
\begin{equation}
    \beta_j=\beta_0+j\eta,
    \qquad
    \eta=\frac{\beta_f-\beta_0}{T},
    \qquad
    j=0,\ldots,T.
\end{equation}

The variation of the Gibbs distribution with respect to $\beta$ can be computed
explicitly. We have
\begin{align}
\partial_\beta \pi_\beta(x) &= \partial_\beta \left( \frac{e^{-\beta E(x)}}{Z(\beta)} \right) \nonumber\\
&= -E(x)\frac{e^{-\beta E(x)}}{Z(\beta)} - \frac{e^{-\beta E(x)}}{Z(\beta)^2}
\partial_\beta Z(\beta).
\end{align}
Since
\begin{equation}
\partial_\beta Z(\beta) = -\sum_{y\in\Omega}E(y)e^{-\beta E(y)}
    = -Z(\beta)\langle E\rangle_\beta,
\end{equation}
where
\begin{equation}
\langle E\rangle_\beta \coloneqq 
    \sum_{y\in\Omega}E(y)\pi_\beta(y),
\end{equation}
it follows that
\begin{align}
\partial_\beta \pi_\beta(x) &= -E(x)\pi_\beta(x) + \langle E\rangle_\beta \pi_\beta(x)
\nonumber\\
&= -\big(E(x)-\langle E\rangle_\beta\big)\pi_\beta(x).
\end{align}
Therefore, in the instantaneous weighted norm, 
\begin{align}
\|\partial_\beta\pi_\beta\|^2_{\pi^{-1}_{\beta}}
&= \sum_{x\in\Omega} \frac{(\partial_\beta\pi_\beta(x))^2}{\pi_{\beta}(x)} 
\nonumber\\
&= \sum_{x\in\Omega} (E(x)-\langle E\rangle_\beta)^2\pi_\beta(x) \nonumber\\
&= \operatorname{Var}_{\pi_\beta}(E).
\end{align}
Hence, 
\begin{equation}
\|\partial_\beta\pi_\beta\|_{\pi^{-1}_{\beta}} =
\sqrt{\operatorname{Var}_{\pi_\beta}(E)}.
\end{equation}

Now, using the Taylor approximation on $\pi_{\beta}$
\begin{equation}
\pi_{\beta_j} - \pi_{\beta_{j-1}} = \eta\partial_\beta\pi_{\beta} + O(\eta^2),   
\end{equation}
so
\begin{equation}
\|\pi_{\beta_j} - \pi_{\beta_{j-1}}\|_j \leq \eta\sqrt{\operatorname{Var}_{\pi_{\beta_j}}(E)} + O(\eta^2).  
\end{equation}

We can then plug  these results in \cref{cor:gibbs_weighted}, which immediately give us
\begin{equation}
\|p_T-\pi_{\beta_f}\|_{\mathrm{TV}}
\leq \frac{1}{2} \sum_{j=1}^T \exp\!\left( -\sum_{r=j}^T\delta_r + \sum_{r=j}^{T-1}\log C_r \right)
\left(\eta\sqrt{\operatorname{Var}_{\pi_{\beta_j}}(E)} + O(\eta^2)\right) .    
\end{equation}

Thus, the energy variance determines how finely the inverse-temperature path must be discretized: regions with large energy fluctuations require smaller steps in $\beta$. 

A more useful consequence is obtained by combining the Gibbs-state variation
bound with \cref{cor:uniform_gap_changing_norm}. This yields an explicit
step-complexity estimate for Gibbs-state preparation.

\begin{corollary}[Uniform-gap step complexity for Gibbs-state preparation]
\label{cor:uniform_gap_complexity}
Let $\beta:[0,1]\rightarrow [\beta_0,\beta_f]$ be a fixed $C^2$ schedule, and set $\pi(s)$, $\pi_j=\pi(j/T)$. Assume the hypotheses of \cref{cor:gibbs_weighted}, with $\delta_j \geq \delta_{\min} > 0$ and $C_j\leq C_{\max}$, for all relevant indices. Define
$j=1,\dots, T$. Define
\begin{equation}
\label{eq:A_M}
A
\coloneqq \sup_{s\in[0,1]} |\dot{\beta}(s)| \sqrt{\operatorname{Var}_{\pi_{\beta(s)}}(E)},
\quad M_2\coloneqq \sup_{s,u \in [0,1]}\|\ddot{\pi}(s)\|_{\pi(u)^{-1}}.
\end{equation}
Then
\begin{equation}
\omega_{\pi} \leq \frac{A}{T} + \frac{M_2}{2T^2}.    
\end{equation}
 If $\Gamma = \delta_{\min} - \log C_{\max}>0$, then the total-variation error satisfies
\begin{equation}
\label{eq:Bound_A_M}
\|p_T-\pi_{\beta_f}\|_{\mathrm{TV}}
\leq \frac{e^{-\delta_{\min}}}{2(1-e^{-\Gamma})}\left(\frac{A}{T} + \frac{M_2}{2T^2}\right).
\end{equation}
Consequently, writing
\begin{equation}
H\coloneqq \frac{e^{-\delta_{\min}}}{2(1-e^{-\Gamma})},  
\end{equation}
the target error $\|p_T-\pi_{\beta_f}\|_{\mathrm{TV}}\leq \epsilon$ is guaranteed whenever
\begin{equation}
T \geq \frac{HA + \sqrt{H^2A^2 + 2\epsilon H M_2}}{2\epsilon}.   
\end{equation}
For fixed path constants and fixed $\Gamma > 0$, the leading error bound is $HA/T$; in the weak-damping regime, $H=O(\Gamma^{-1})$.
\end{corollary}

The proof of the corollary is presented in
\cref{proof:uniform_gap_complexity}. In the special case in which the weighted norms do not change, $C_{\max}=1$, and hence $\Gamma=\delta_{\min}$. For a linear schedule,
$\dot{\beta}(s)=\beta_f-\beta_0$, the condition 
\begin{equation}
T \geq \max \left\{
\frac{\beta_f-\beta_0}
{\epsilon\delta_{\min}}
V_{\max},\sqrt{\frac{M_2}{2\epsilon\delta_{\min}}}\right\},
\end{equation}
guarantees  $\|p_T-\pi_{\beta_f}\|_{\mathrm{TV}}\leq \epsilon$,  where
\begin{equation}
V_{\max} = \max_{\beta\in[\beta_0,\beta_f]}
\sqrt{\operatorname{Var}_{\pi_\beta}(E)}.   
\end{equation}
The quantity $|\dot{\beta}(s)|
\sqrt{\operatorname{Var}_{\pi_{\beta(s)}}(E)}$ has a natural geometric interpretation: it is the instantaneous speed of
the Gibbs path with respect to the weighted statistical or Fisher-information metric .

\subsection{One-Dimensional Tracking Geometry and the Gibbs Schedule}

Our goal is now to extract the Gibbs-preparation schedule that follows from the discrete adiabatic tracking theorem in instantaneous weighted norms. We do this from a geometric point of view. We first introduce the one-dimensional Gibbs manifold and then derive the corresponding geodesic equation using the parallel-transport formulation. The schedule is obtained by analysing how the tracking error accumulates under a sequence of Markov operators for the Gibbs path. This error accumulation naturally induces a metric on the Gibbs manifold. Putting these ingredients together, we obtain the affine, constant-metric-speed parametrization of the Gibbs path. This parametrization minimizes the metric-energy functional induced by the local approximation to the tracking bound among monotone schedules with fixed endpoints and duration.

We can view the problem geometrically as involving the Gibbs manifold, $\mathcal{M}_\beta=[\beta_0,\beta_f],$ which is the one-dimensional parameter manifold. We next introduce a local, slowly varying approximation to the discrete tracking bound, where its proof is presented in \cref{proof:metric}.  The resulting metric is therefore an effective metric induced by the leading local form of the bound.

\begin{lemma}[Local tracking metric for Gibbs-state preparation]
\label{lemma:effective_metric}
Let $\beta[0,1]\to[\beta_0,\beta_f]$ be a fixed smooth monotone schedule, with $t_j=j/T$ $\beta_j=\beta(t_j)$, and $V(\beta)=\sqrt{\operatorname{Var}_{\pi_\beta}(E)}$. Assume the tracking inequality, 
\begin{equation}
e_j \leq (1-\delta_j)C_{j-1}e_{j-1}
+ (1-\delta_j)|h_j|V_j
+ O(h_j^2),
\end{equation}
$p_0=\pi_0$, and that $\delta,V,\dot\beta$ are continuously differentiable, with
\begin{equation}
\delta(\beta)\geq \delta_{*} > 0, \quad  C_j=1+O\left(T^{-1}\right).
\label{eq:C_local_metric}
\end{equation}
Define the scalar tracking envelope
\begin{equation}
\label{eq:B_env}
B_0=0,\qquad
B_j=(1-\delta_j)C_{j-1}B_{j-1}
 +(1-\delta_j)\|\pi_j-\pi_{j-1}\|_j. 
 \end{equation}
Then $e_j\le B_j$, and, uniformly on every fixed interval $t_j\in[t_0,1]$, $t_0>0$,
\begin{equation}
B_j=\frac1T
\frac{[1-\delta(\beta_j)]|\dot\beta(t_j)|V(\beta_j)}
{\delta(\beta_j)} + O(T^{-2}).    
\end{equation}
The remainder constants may depend on the fixed path and $\delta_*$. This expansion describes the upper bound, not necessarily the actual error. Squaring the leading local coefficient defines, wherever $V>0$ and $0<\delta<1$, the effective tracking metric
\begin{equation}
\label{eq:main_metricOd}
g^{\mathrm{track}}
=\frac{(1-\delta(\beta))^2\operatorname{Var}_{\pi_\beta}(E)}
{\delta(\beta)^2}\,d\beta\otimes d\beta.
\end{equation}
\end{lemma}

The metric in \cref{eq:main_metricOd} is obtained by expressing the square of the leading local tracking cost as a quadratic form in the schedule velocity. For a smooth and strictly positive stationary path, the norm-comparison condition in \cref{eq:C_local_metric} is naturally satisfied, since consecutive weighted norms differ by $O(T^{-1})$. The norm-comparison factor therefore introduces an $O(T^{-1})$ correction to the local damping rate, while the tracking dynamics are governed by the competition between error injection, proportional to $|\dot{\beta}|V$, and relaxation at rate $\delta$. Finally, the geodesic equation $\nabla_{\dot{\gamma}}\dot{\gamma}=0$ is equivalent to the Euler--Lagrange equations of the metric-energy functional $\mathcal{L}(\beta,\dot{\beta})$. Applying this variational principle to the effective tracking metric yields the constant-speed Gibbs schedule stated in the following corollary.

\begin{corollary}[Constant-speed schedule for the local tracking metric]
\label{lemma:ched}
Let $\gamma(t)=\pi_{\beta(t)}$ be a monotone curve on the
one-dimensional Gibbs manifold equipped with the local tracking metric
\begin{equation}
g_{\beta\beta}^{\mathrm{track}}(\beta)
=
\frac{
(1-\delta(\beta))^2 V(\beta)^2
}{
\delta(\beta)^2
}.
\end{equation}
Then $\gamma(t)$ is an affinely parametrized geodesic if and only if it has
constant speed with respect to the tracking metric,
\begin{equation}
\|\dot{\gamma}(t)\|_{\mathrm{track}}
=
\sqrt{
g_{\beta\beta}^{\mathrm{track}}(\beta(t))
}
\,|\dot{\beta}(t)|
=
\alpha,
\end{equation}
where $\alpha>0$ is fixed by the boundary conditions. Consequently, for a
monotonically increasing schedule,
\begin{equation}
\label{eq:classical_sched_Trac}
\dot{\beta}(t)
=
\alpha
\frac{
\delta(\beta(t))
}{
(1-\delta(\beta(t)))V(\beta(t))
}
=
\alpha
\frac{
\delta(\beta(t))
}{
(1-\delta(\beta(t)))
\sqrt{\operatorname{Var}_{\pi_{\beta(t)}}(E)}
}.
\end{equation}
\end{corollary}
The proof of the Corollary above is presented in \cref{proof:ched}.

The preceding result has a direct geometric interpretation. Since the Gibbs manifold is one-dimensional, the path connecting $\pi_{\beta_0}$ and
$\pi_{\beta_f}$ is fixed, and the freedom in choosing the schedule lies only in its time parametrization. The effective tracking metric stretches regions
where the Gibbs distribution changes fast or where the Markov dynamics relaxes slowly. The affine geodesic compensates for this variation by
traversing the Gibbs path at constant metric speed. Consequently, the schedule proceeds more slowly in regions of large energy fluctuations or weak effective damping, and faster in regions where the Gibbs state
changes slowly and tracking errors contract efficiently.

The metric and the resulting schedule are obtained from a local slow-variation approximation to the discrete tracking bound. The affine geodesic should therefore be understood as the distinguished constant-speed parametrization selected by the induced tracking geometry, rather than as an
exact optimizer of the complete nonlocal tracking bound.

The classical tracking theorem provides a natural route towards constructing adaptive simulated-annealing algorithms for combinatorial optimization. Rather than prescribing a fixed cooling schedule, the inverse temperature $\beta$ can be updated according to the balance between the contraction of the Markov dynamics and the variation of its instantaneous Gibbs distribution. This suggests an iterative procedure in which the Markov chain is evolved at the current temperature, the local spectral gap and energy fluctuations are estimated from sampled trajectories, and the next temperature increment is determined by the tracking condition. Its practical implementation and numerical assessment are left for future work.

\section{Numerical validation of the Adiabatic Tracking and schedule design}
\label{sec:Num_valid_cl}

\subsection{Analytically solvable two-state model}\label{sec:two-state_m}

We start by looking at a solvable two-state model. There are two main reasons to consider this simple case. First, it provides a simple setting where all quantities entering in the discrete adiabatic theorem can be evaluated exactly, allowing us to test the tightness of the tracking bound independently of sampling or numerical estimation errors. Second, it provides a controlled benchmark for comparing different inverse temperature schedules.

Consider the configuration space $\Omega =\{0,1\}$, with energies $E(0)=0$ and $E(1)=1$. In this case the gibbs distribution \cref{eq:gibbs_dis} has the following partition function
\begin{equation}
Z(\beta) = 1+ e^{-\beta},    
\end{equation}
so 
\begin{equation}
\pi_{\beta}(0) =\frac{e^{\beta}}{1+e^\beta},\quad     \pi_{\beta}(1) =\frac{1}{1+e^\beta}.
\end{equation}
Thus, for this two-state problem, the energy variance is
\begin{align}
\label{eq:explic_var}
 \text{Var}_{\pi_\beta}(E) &= (0-\langle E \rangle_{\beta})^2\pi_{\beta}(0) +  (1-\langle E \rangle_{\beta})^2\pi_{\beta}(1)\nn
 &=\pi_{\beta}(1)\pi_{\beta}(0)\nn
 &=\pi_{\beta}(1)(1-\pi_{\beta}(1))\nn
 &=\frac{e^{\beta}}{(1+e^{\beta})^2}\nn
 &=\frac{1}{4\cosh^2(\beta/2)},
\end{align}
where we used that $\langle E \rangle_{\beta}=\pi_{\beta}(1)$ and $\pi_{\beta}(0)=1-\pi_{\beta}(1)$. Notice that the quantity of interest is $V(\beta)=\sqrt{\text{Var}_{\pi_\beta}(E)}$.

Now, to isolate the effect of the annealing schedule, we consider a reversible Markov chain, whose the stationary distribution is exactly $\pi_\beta$, we define
\begin{equation}
\label{eq:Markov_chain_analy}
    \mathcal{W}_\beta =(1-\delta(\beta))I + \delta(\beta)\Pi_{\beta},
\end{equation}
where $\Pi_\beta$ denotes the rank-one stochastic operator that maps every normalized distribution to $\pi_\beta$, $\Pi_\beta p=\pi_\beta$, so $\mathcal{W}_\beta\pi_\beta=\pi_\beta.$ From the constructed Markov operator \cref{eq:Markov_chain_analy}, notice that because $\Pi_\beta$ has eigenvalues 0 and 1, the two eigenvalues of the Markovian operator are 
\begin{equation}
\lambda_1 = 1,\quad \lambda_2=1-\delta(\beta),    
\end{equation}
so $\delta(\beta)$ is exactly the gap of the $\mathcal{W}_\beta$, and for this problem we consider,
\begin{equation}
\delta(\beta) = e^{-\beta}.    
\end{equation}
For the discrete inverse-temperature schedule, $\beta_0,\dots \beta_T$, we then have the following recursively dynamics,
\begin{equation}
p_j =\mathcal{W}_{\beta_j}p_{j-1} = (1-\delta_j)p_{j-1} + \delta_j\pi_j, \quad p_0=\pi_{\beta_0}.   
\end{equation}
with $\delta_j =\delta(\beta_j)$ and $\pi_j=\pi_{\beta_j}$. 

We can now compute explicitly the error $e_j=p_j-\pi_j$ for each step $j$ using \cref{eq:fund_theor}. Following the steps for the proof of \cref{eq:fund_theor}, we have $ \mathcal{W}_{\beta_j}\left(e_j + \pi_{j-1} + \pi_j\right)$, however, from \cref{eq:Markov_chain_analy}, we see that
\begin{equation}
\label{eq:err_two_m}
 e_j = (1-\delta_j)(e_{j-1} + \pi_{j-1} - \pi_{j}).  
\end{equation}
For a monotonically increasing inverse-temperature schedule, $\pi_j(1)$ decreases with $j$. Starting from $e_0=0$, \cref{eq:err_two_m} implies that $e_{j-1}(1)\geq 0$, while $\pi_{j-1}-\pi_j \geq 0$. Since both vectors have zero total mass, they are aligned. The triangle inequality is therefore saturated, and tacking the instantaneous weighted norm gives
\begin{equation}
 \|e_j\|_{j} = (1-\delta_j)\|e_{j-1}\|_j + (1-\delta_j)\|\pi_{j-1} - \pi_{j}\|_j,  
\end{equation}
i.e., there is no inequality at this point. However, we still have to include the norm-comparison factor. For the two-dimensional configuration space we now that every zero-mass vector $q\in H_0$ can be written as $(a,-a)$, so the instantaneous weighted norm satisfies  
\begin{align}
\|q\|^2_j &= a^2 \left(\frac{1}{\pi_j(0)} + \frac{1}{\pi_j(1)}\right) \nn
&= \frac{a^2}{\pi_j(0)\pi_j(1)} \nn
&= \frac{a^2}{V(\beta_j)^2},  
\end{align}
Thus,
\begin{equation}
C_{j-1} = \frac{\|e_{j-1}\|_j}{\|e_{j-1}\|_{j-1}} =  \frac{V(\beta_{j-1})}{V(\beta_j)},   
\end{equation}
where $V(\beta_{j})$ is given in \cref{eq:explic_var}.
We then have an equality expression for the error, 
\begin{equation}
 \|e_j\|_{j} = (1-\delta_j)C_{j-1}\|e_{j-1}\|_{j-1} + (1-\delta_j)\|\pi_{j-1} - \pi_{j}\|_j,  
\end{equation}
where all the quantities have a closed form.

We now investigate the tracking error under different schedules $\beta(t)$, with
$t\in[0,1]$, $\beta(0)=\beta_0$, and $\beta(1)=\beta_f$. Throughout this section we take $\beta_0=0$ and $\beta_f=4$. In addition to standard interpolation schedules, we consider schedules constructed from the Gibbs-state variance, from the spectral gap, from the thermodynamic friction tensor \cite{sivak2012thermodynamic}, and from the local tracking metric derived above, all the schedules are given in \cref{app:sched_construc}.

\begin{figure*}[t]
\centering

% ---------- first row ----------
\begin{minipage}[t]{0.40\textwidth}
    \centering
    \includegraphics[width=\linewidth]{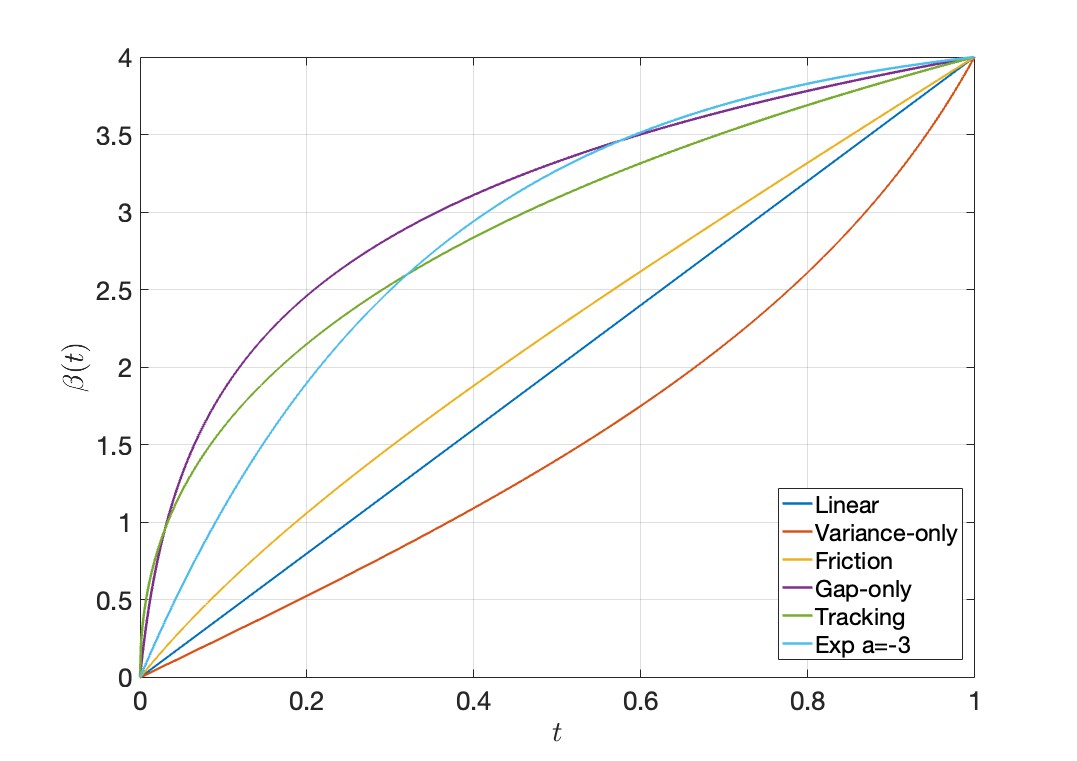}\\[-1mm]
    \textbf{(a)}
\end{minipage}
\hfill
\begin{minipage}[t]{0.40\textwidth}
    \centering
    \includegraphics[width=\linewidth]{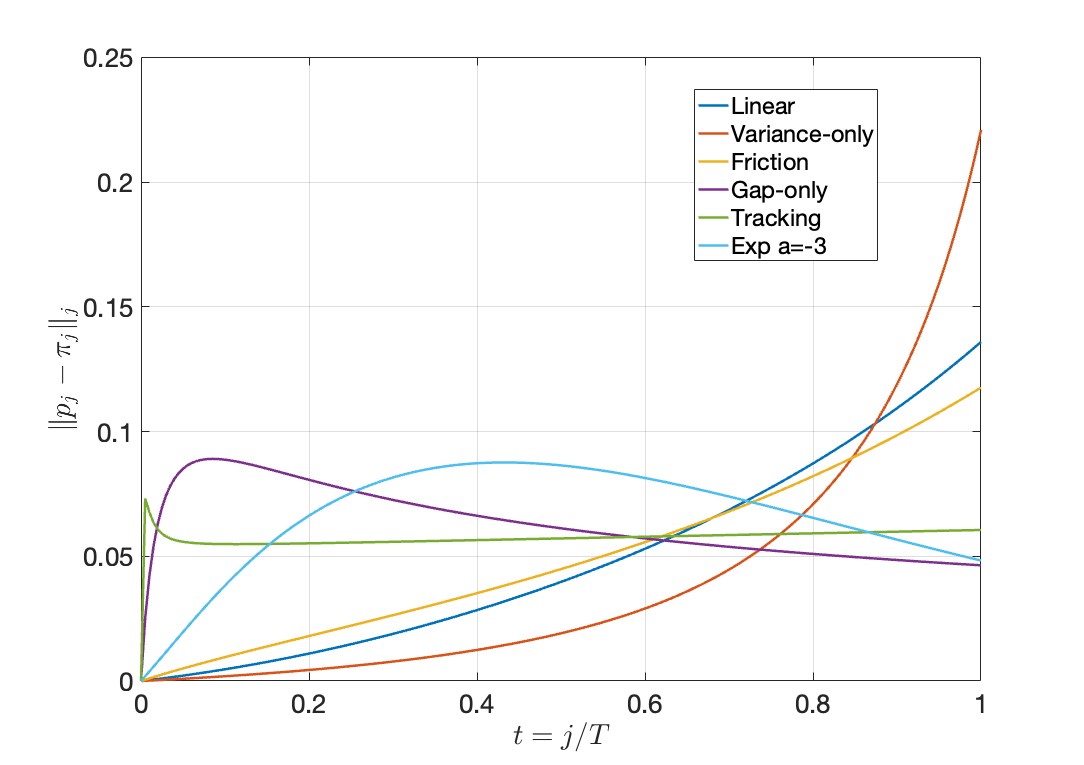}\\[-1mm]
    \textbf{(b)}
\end{minipage}

\vspace{3mm}

% ---------- second row ----------
\begin{minipage}[t]{0.45\textwidth}
    \centering
    \includegraphics[width=\linewidth]{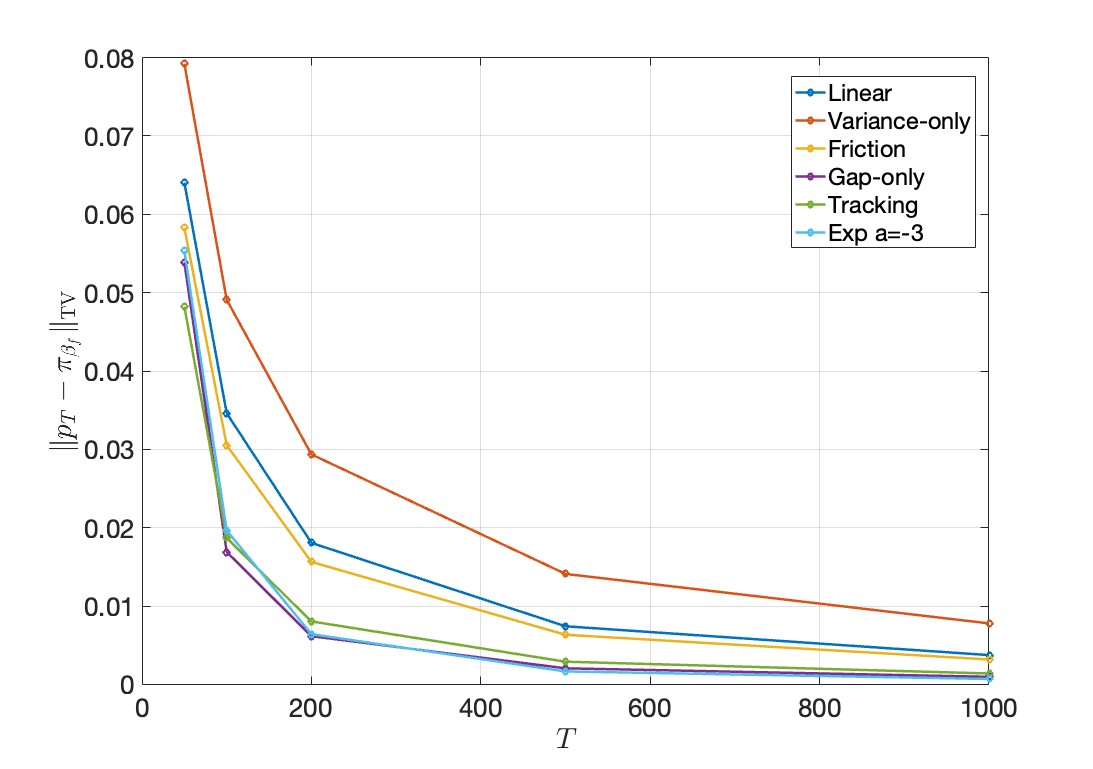}\\[-1mm]
    \textbf{(c)}
\end{minipage}

\caption{
Comparison of schedule design and tracking performance for the analytically solvable two-state model with $\beta_0=0$ and $\beta_f=4$. (a) Inverse-temperature schedules $\beta(t)$ for the main theoretically
motivated protocols. The variance-only schedule depends only on the Gibbs-state variation, the friction schedule follows the constant-excess-power prescription, the gap-only schedule adapts to the relaxation gap, and the local tracking schedule combines the gap and
Gibbs-state variation according to the tracking metric.
(b) Instantaneous weighted tracking error
$e_j=\|p_j-\pi_j\|_{j}$ for $T=200$ Markov steps.
Although $e_j$ contains contributions generated at previous steps, it need not increase monotonically because subsequent Markov steps contract
the accumulated deviation. The local tracking schedule gives the smallest maximum instantaneous error $e_{\max}$ among the theoretically motivated schedules considered.
(c) Final total-variation error $\|p_T-\pi_{\beta_f}\|_{\mathrm{TV}}$ as a function of the number of Markov steps $T$. The gap-only schedule gives the smallest terminal error,
 showing that minimizing the maximum instantaneous tracking error and minimizing the final preparation error are distinct optimization objectives.
}
\label{fig:two_state_summary}

\end{figure*}

The results in \cref{fig:two_state_summary} compare the different schedules
through both the terminal preparation error and the maximum instantaneous weighted tracking error,
\begin{equation}
\label{eq:e_max}
e_{\max} = \max_{0\leq j\leq T}\|e_j\|_j .
\end{equation}
A clear distinction appears between pathwise tracking and terminal state
preparation. The local Tracking schedule gives the smallest maximum instantaneous error, $e_{\max}=7.31\times10^{-2}$,
showing the best control of the deviation from the instantaneous Gibbs state throughout the evolution. However, errors generated at earlier steps can be contracted by subsequent applications of the Markov chain through the factor
$1-\delta_j$, so minimizing $e_{\max}$ does not necessarily minimize the final
error.

Indeed, the Gap-only schedule gives the smallest final total-variation error,
$6.17\times10^{-3}$, followed closely by the front-loaded exponential schedule
with $a=-3$, for which the final error is $6.42\times10^{-3}$. The Tracking
schedule gives a slightly larger terminal error, $8.05\times10^{-3}$, despite
having the smallest $e_{\max}$. The Friction schedule shows an intermediate
behaviour, improving over both Linear and Variance-only but not outperforming
Tracking in pathwise control. This illustrates that constant-excess-power
optimization~\cite{sivak2012thermodynamic}, local adiabatic tracking, and
terminal state preparation correspond to distinct finite-time objectives. Detailed numerical values for all schedules at $T=200$ are reported in
\cref{tab:two_state_T200} in \cref{sec:supp_mate}.

\subsection{Dynamical tracking versus endpoint performance in the
Sherrington--Kirkpatrick model}

We next investigate the tracking geometry using the Sherrington--Kirkpatrick (SK) spin-glass model~\cite{kirkpatrick1983optimization}. This provides a nontrivial many-body setting in which, for the system sizes considered here, the complete Gibbs distribution and the full time-inhomogeneous Markov evolution can be evaluated exactly, without Monte Carlo sampling of the evolving probability distribution. The model also allows us to separate two distinct finite-time objectives: maintaining the evolving distribution close to the instantaneous Gibbs state throughout the annealing path, and efficiently preparing a high-quality approximation to the Gibbs state at the final inverse temperature.

For a disorder realization $r$, the configuration space is
$\Omega=\{-1,+1\}^{n}$, with $|\Omega|=2^n$. A soin configuration is denoted by $\sigma=\{\sigma_1,\cdots,\sigma_n\}$, where $\sigma_{i}\in\{-1,+1\}$. The couplings $J_{ij}^{(r)}$, $i<j$ are independently sampled from $\{-1,+1\}$, with equal probability, with $J^{(r)}_{ji}=J^{(r)}_{ij}$ and $J_{ii}^{(r)}=0$. In the classical schedule comparison, the energy is
\begin{equation}
 E_r(\sigma) =  -2\sum_{1\leq i<j\leq n} J_{ij}^{(r)}\sigma_i\sigma_j ,
\end{equation}
Here the factor of two is part of the energy convention adopted for this classical schedule benchmark. The sum is written over unordered pairs $i<j$; the factor does not arise from double counting within that sum. 

We consider system sizes $n=\{6,8,10,12\}$ and anneal from
$\beta_0=0$ to $\beta_f=0.5$. The exact Gibbs distribution at each inverse temperature is computed directly from the energy spectrum. The Markov evolution is evaluated for
\begin{equation}
\label{eq:samp_va}
T\in \{2^k\times 10^2\}^7_{k=0}.
\end{equation}
For each system size we generate independent disorder realizations, with the same realization used for every schedule being compared so that all comparisons are paired at the instance level. The schedules considered are Linear, Variance-only, Friction, Gap-only, Tracking, and exponential schedules with $a=-3$ and $a=3$. Their explicit construction is given in \cref{app:sched_construc}.

For each realization and schedule, the actual distribution $p_j$ generated after $j$ Markov steps is compared directly with the instantaneous Gibbs distribution $\pi_{\beta_j}$. 
For the SK instances considered, we evaluate the final total-variation error $e_{\mathrm{TV}}^{\mathrm{final}}$, given in \cref{eq:TV_er}
at the final evolution step $j=T$, so that
$e_{\mathrm{TV}}^{\mathrm{final}}=e_{\mathrm{TV}}^T$.

To characterize optimization-relevant endpoint performance, we additionally use the final ground-state probability
\begin{equation}
P_{\mathrm{gs}}(T) = \sum_{x:E(x)=E_{\mathrm{gs}}} p_T(x),
\end{equation}
normalized by its exact Gibbs value,
\begin{equation}
q_{\mathrm{succ}}(T) = \frac{ P_{\mathrm{gs}}(T) }{ P_{\mathrm{gs}}^{\mathrm{Gibbs}}
}.
\label{eq:sk_qsucc}
\end{equation}
Thus, $q_{\mathrm{succ}}=1$ corresponds to reproducing the Gibbs ground-state probability at $\beta_f$. We also use the approximation-energy ratio
\begin{equation}
q_{\mathrm{AEAR}}(T)
= \frac{ \langle E\rangle_T-\langle E\rangle_{\beta_0}
}{ \langle E\rangle_{\beta_f}^{\mathrm{Gibbs}}
- \langle E\rangle_{\beta_0} },
\label{eq:sk_qaear}
\end{equation}
where $\langle E\rangle_T = \sum_x p_T(x)E(x)$. This quantity measures the fraction of the total Gibbs energy improvement, relative to the initial high-temperature state, that has been achieved after $T$ steps, as introduced in \cite{costa2025assessing}. Again, $q_{\mathrm{AEAR}}=1$ corresponds to the exact Gibbs expectation value at $\beta_f$.

We then tested the different schedules by determining the first sampled number of Markov steps required to reach the target value $q_\star=0.99$. In particular,
\begin{equation}
\label{eq:target_acc}
T_{0.99}^{P_{\mathrm{gs}}}
= \min\left\{T \in \mathcal{T} :q_{\mathrm{succ}}(T)\geq0.99\right\},
\qquad
T_{0.99}^{\mathrm{AEAR}}
= \min\left\{T \in \mathcal{T}:q_{\mathrm{AEAR}}(T)\geq0.99\right\}.
\end{equation}
The minimum is taken over the sampled values given in \cref{eq:samp_va}. Instances that do not reach the target within this range are treated as right-censored. Mean hitting times are computed only over instances that reach the target and may therefore involve different subsets of realizations for different schedules. We also report the fraction of  realizations reaching each target. 

For the SK model, we consider the lazy Metropolis transition
matrix $\mathcal W(\beta)$, whose elements are given by
\begin{equation}
\label{eq:SK_lazy_transition}
\mathcal W_{yx}(\beta)
=
\begin{cases}
(1-\alpha)T_{yx}(\beta),
& y\in\mathcal N(x),\\[4pt]
\displaystyle
1-(1-\alpha)\sum_{z\neq x}T_{zx}(\beta),
& y=x,\\[8pt]
0,
& \text{otherwise},
\end{cases}
\end{equation}
where
\begin{equation}
\label{eq:SK_acceptance}
T_{yx}(\beta)
=
\frac{1}{M}
\min\left\{
1,e^{\beta[E(x)-E(y)]}
\right\},
\qquad x\neq y.
\end{equation}
Here, $\mathcal N(x)$ denotes the set of configurations
reachable from $x$ by a single-spin flip, $M=|\mathcal N(x)|=n$,
and $\alpha\in[0,1)$ is the laziness parameter. In our numerical simulations, we set $\alpha=1/2$. Further details of this matrix construction and implementation  are provided in \cref{app:schedu_sk}.

In~\cref{fig:sk_TV_final}, we show the mean final total-variation error $e_{\mathrm{TV}}^{\mathrm{final}}$. We observe that the gap-weighted schedules, particularly Tracking and Gap-only, achieve smaller terminal errors, while Friction exhibits intermediate performance and Variance-only does not consistently outperform Linear.

We then compare the schedules in \cref{fig:sk_success_summary} using the two target criteria defined in \cref{eq:target_acc}. We find that schedules incorporating spectral information reach the targets $q_{\mathrm{succ}}\geq0.99$ and $q_{\mathrm{AEAR}}\geq0.99$ with fewer Markov steps than Linear or Variance-only annealing over the system sizes considered. Both $\overline{T}_{0.99}^{P_{\mathrm{gs}}}$ and $\overline{T}_{0.99}^{\mathrm{AEAR}}$ increase with $n$, reflecting the increasing difficulty of preparing the low-temperature Gibbs state, while Tracking and Gap-only remain systematically favorable across the investigated finite-size range. The log-linear fits
\begin{equation}
\log \overline{T}_{0.99}=a+\alpha n
\end{equation}
serve only as empirical guides to the finite-size data and are not interpreted as asymptotic scaling laws. Detailed results, including those for the remaining schedules, are provided in \cref{sec:supp_mate}.

\begin{figure}[t]
    \centering
\includegraphics[width=0.45\linewidth]{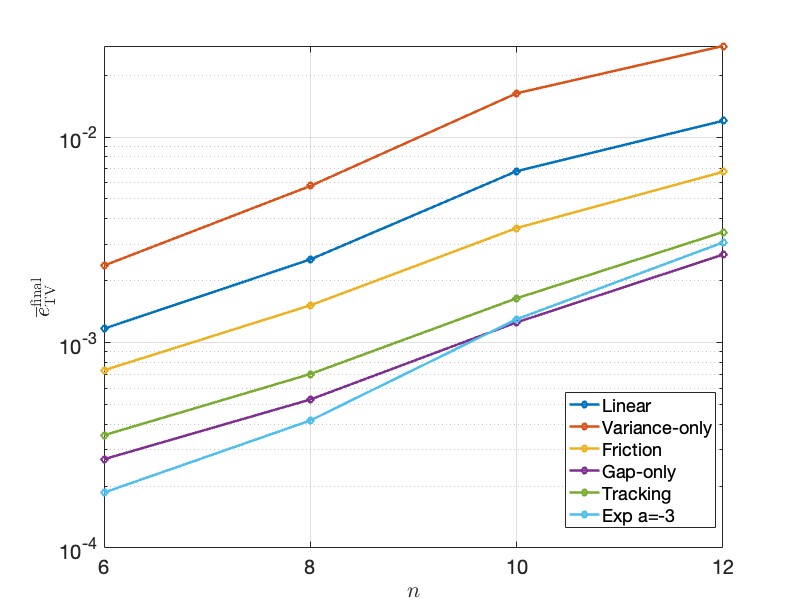}
    \caption{Mean final total-variation error
$e_{\mathrm{TV}}^{\mathrm{final}}$ for the lazy SK chain
at $T=6.4\times 10^3$, using the same disorder realizations
for every schedule at each system size $n$.}
\label{fig:sk_TV_final}
\end{figure}

\begin{figure*}[t]
\centering
\begin{minipage}[t]{0.46\textwidth}
    \centering

    \includegraphics[width=\linewidth]{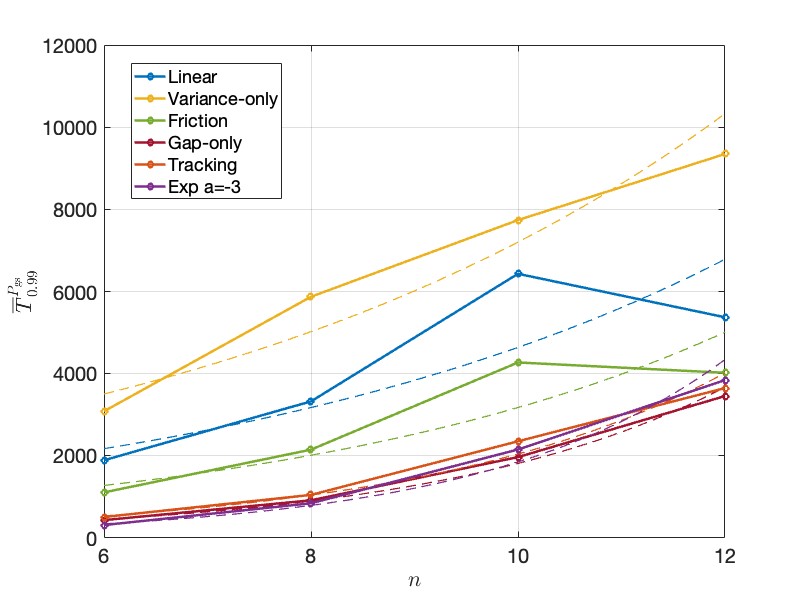}\\[-1mm]
    \textbf{(a)}
\end{minipage}
\hfill
\begin{minipage}[t]{0.46\textwidth}
    \centering
    \includegraphics[width=\linewidth]{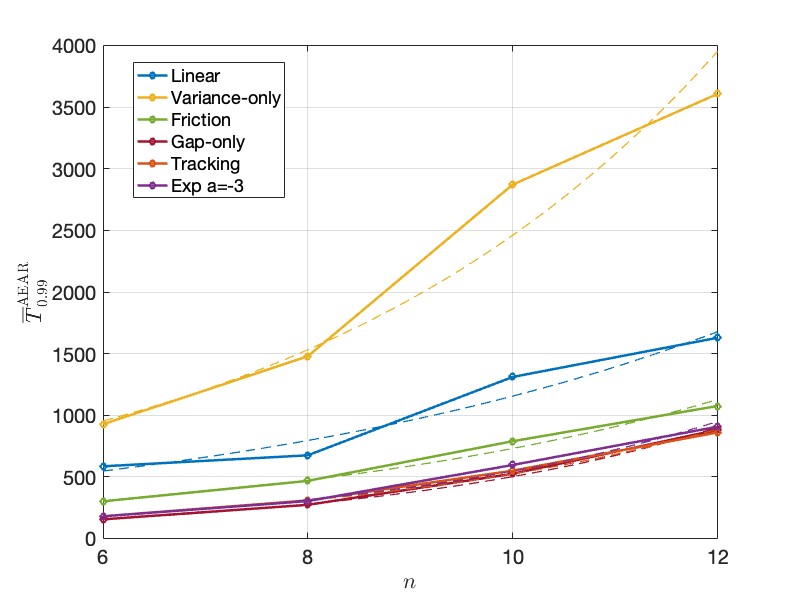}\\[-1mm]
    \textbf{(b)}
\end{minipage}

\caption{
Endpoint race performance for the lazy SK chain using the same disorder realizations for every schedule at each system size $n$. We have in (a) the mean first sampled number of Markov steps $\overline{T}_{0.99}^{P_{\mathrm{gs}}}$ required to reach $q_{\mathrm{succ}}(T) \geq0.99$. The mean is evaluated over realizations that reach the target within the number of steps given in \cref{eq:samp_va}. Different schedules may contribute different subsets of realizations to these conditional means. In (b) we have the race statistic for the normalized energy target $q_{\mathrm{AEAR}}(T)\geq0.99$. The vertical axes are logarithmic and the dashed curves denote finite-size log-linear fits. These fits are shown only as guides to the observed finite-size behaviour and are not interpreted as asymptotic scaling laws.}
\label{fig:sk_success_summary}
\end{figure*}

We also investigated whether the schedule construction can be transferred to system sizes beyond those for which the instantaneous Markov spectrum is computed directly. In this case, rather than diagonalizing the full Markov operator along the annealing path, we extrapolate the spectral-gap behavior from exact results obtained for $n=6,8,10,12$. We then analyze the extrapolation and test the predicted schedule for $n=18$, estimating the energy fluctuations directly from samples of the corresponding Metropolis dynamics. Details of the extrapolation procedure are provided in \cref{app:SK_extap}.

We first show the predicted schedules in \cref{fig:sk_n18_schedule}. Both gap-informed schedules clearly differ from linear interpolation, evolving faster at low inverse temperatures and allocating more Markov steps to the high inverse-temperature region. The Gap-only schedule exhibits the strongest adjustment, while Tracking also incorporates the sampled energy fluctuations, resulting in a smoother variation.

\begin{figure}[t]
    \centering
    \includegraphics[width=0.50\linewidth]{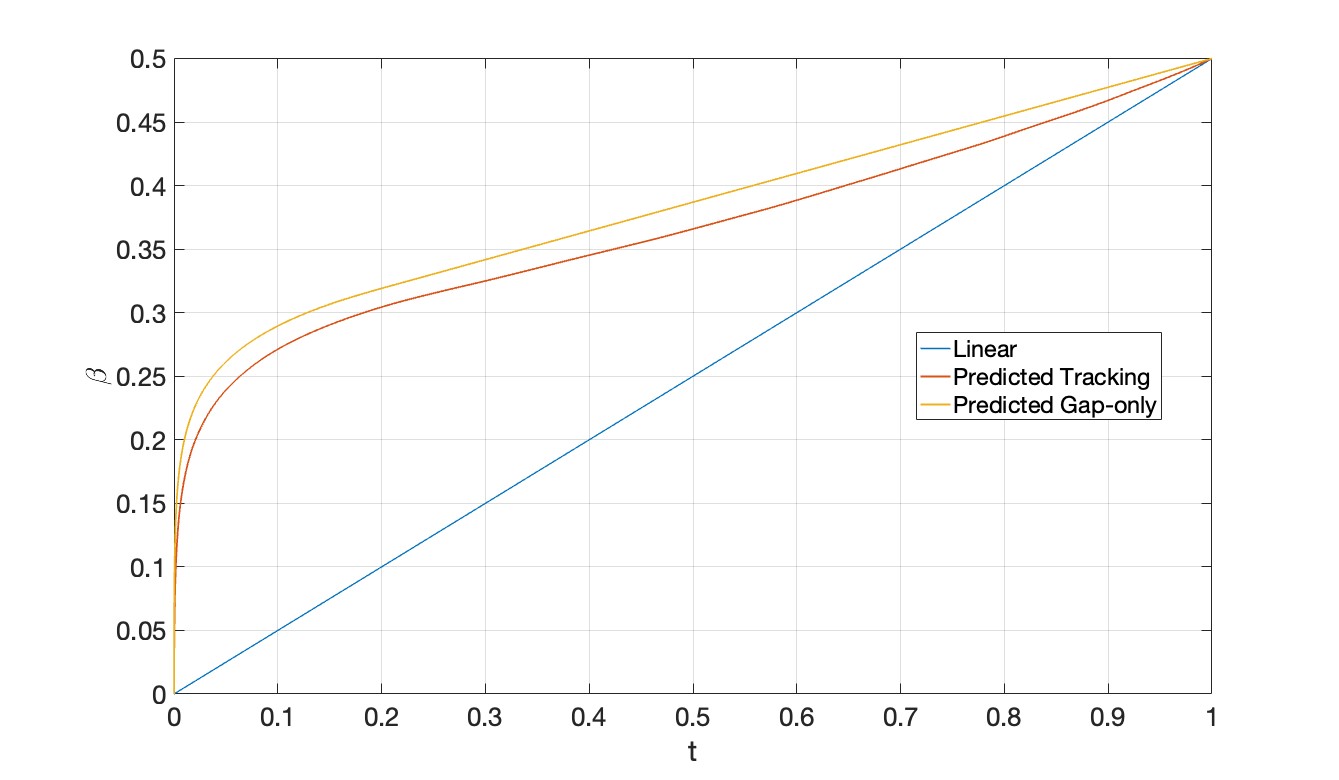}
    \caption{Inverse-temperature schedules for the $n=18$ SK transfer experiment. The predicted Gap-only schedule is constructed from the extrapolated spectral-gap  $\widehat{\delta}_{18}(\beta)$ obtained from the
    $n=6,8,10,12$ results. The predicted Tracking schedule combines the same extrapolated gap with the energy standard deviation $\widehat V(\beta)$ sampled directly from the $n=18$ Metropolis dynamics.
    Both predicted schedules allocate a larger fraction of the evolution to the low-temperature region than the linear interpolation.
    }
    \label{fig:sk_n18_schedule}
\end{figure}

The corresponding terminal-performance results are shown in \cref{fig:sk_n18_performance}. The left panel reports the mean normalized approximation-energy ratio, $\overline q_{\rm AEAR}(T)$, while the right panel shows the mean normalized ground-state probability, $\overline q_{\rm succ}(T)$. The predicted schedules are constructed using spectral gaps extrapolated from averages over 56 disorder realizations at each smaller system size.

For the approximation-energy objective, both predicted schedules outperform Linear over the entire range of Markov steps considered, with Tracking achieving the highest mean $\overline q_{\rm AEAR}$. The differences decrease as $T$ increases and the evolution approaches the quasistatic regime. The improvement is more pronounced for the ground-state probability, where both gap-adapted schedules achieve higher $\overline q_{\rm succ}$ than Linear. This suggests that allocating more steps to the low-temperature region, where relaxation is slower, improves the preparation of the target Gibbs distribution.

\begin{figure*}[t]
\centering
\begin{minipage}[t]{0.48\textwidth}
    \centering
    \includegraphics[width=\linewidth]{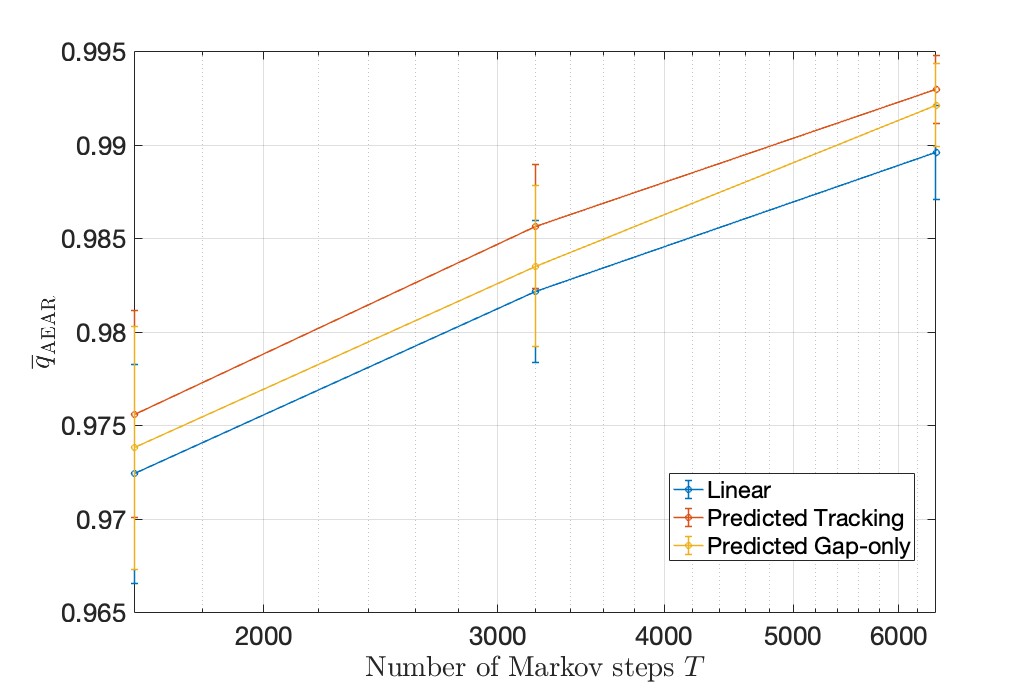}\\[-1mm]
    \textbf{(a)}
\end{minipage}
\hfill
\begin{minipage}[t]{0.48\textwidth}
    \centering
    \includegraphics[width=\linewidth]{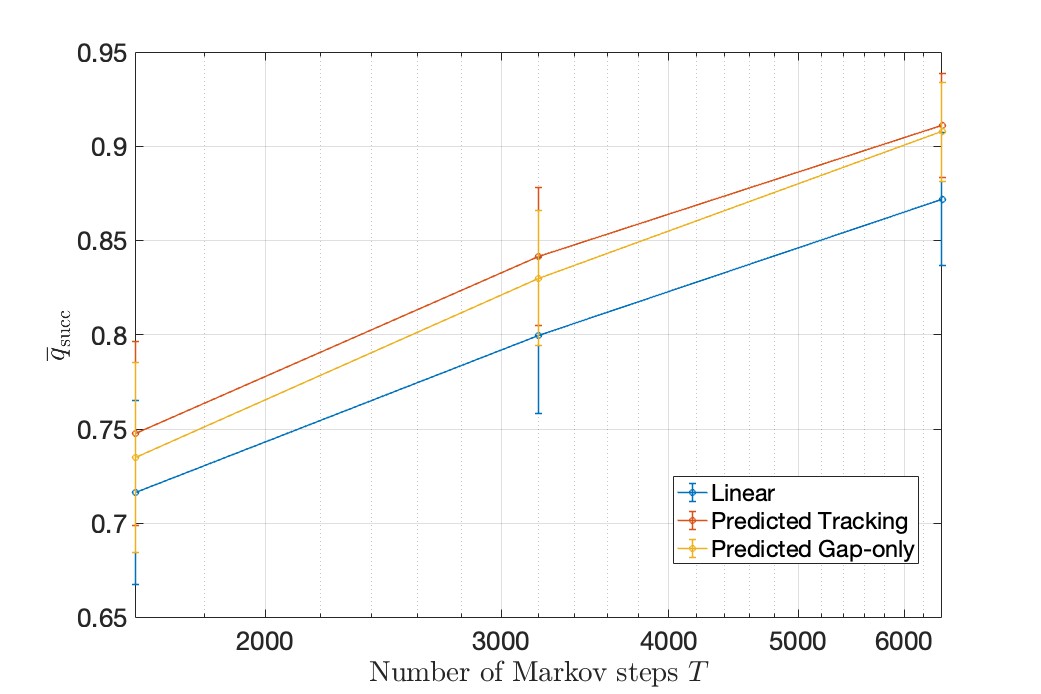}\\[-1mm]
    \textbf{(b)}
\end{minipage}

\caption{Terminal performance for the $n=18$ SK extrapolation study using the linear, predicted Tracking, and predicted Gap-only schedules. (a) Mean normalized approximation-energy ratio $\overline q_{\rm AEAR}$. (b) Mean normalized ground-state probability $\overline q_{\rm succ}$.
Error bars show the variation across disorder realizations. Both gap-informed schedules outperform the linear schedule at finite $T$, more significantly for the ground-state probability. The $n=18$ gap behavior is extrapolated from the
exactly analyzed systems $n=6,8,10,12$ and should therefore be interpreted as a prediction rather than a rigorous gap estimate. }
\label{fig:sk_n18_performance}
\end{figure*}

These results provide a practical test of the proposed schedule construction beyond the sizes that can be treated exactly. The full $n=18$ Markov spectrum is not required to design the annealing path: the slow-relaxation behaviour is inferred from the finite-size dependence of smaller exactly tractable systems, while the thermodynamic contribution $\widehat V(\beta)$ is estimated directly from samples of the target system. The improvement over the Linear schedule therefore suggests that spectral information from smaller systems can be combined with directly sampled thermodynamic quantities to construct effective schedules at larger sizes. Since the extrapolated gap is not a certified lower bound, however, these results should be interpreted as numerical evidence rather than as a rigorous application of the tracking theorem.

\section*{Part B: Quantum adiabatic tracking}

\section{Discrete Quantum adiabatic theorems}
\label{sec:int_disc}

We apply the discrete quantum adiabatic theorem of
Refs.~\cite{Costa2022QLSP,AnCostaBerry2025} to a temperature-dependent family
of Szegedy quantum walks. Let $U_w(s)$, with $ s\in[0,1]$,  be a family of unitary operators sampled at $s_j= j/T_d$, $j=0,\ldots,T_d-1$. The physical evolution after $n$ steps is
\begin{equation}
\mathcal{U}\!\left(\frac{n}{T_d}\right)
= \prod_{j=0}^{n-1} U_w\!\left(\frac{j}{T_d}\right),
\end{equation}
while $\mathcal{U}_A(s)$ denotes the corresponding ideal adiabatic evolution, which
transports the instantaneous eigenspace of interest without leakage. The theorem quantifies how close $\mathcal{U}(s)$ remains to $\mathcal{U}_A(s)$ in terms of the
variation of the walk and the local eigenphase gap. The first and second
finite differences are bounded as
\begin{equation}
\label{eq:dif_op}
\left\| D^{(k)}U_w(s) \right\| \leq \frac{c_k(s)}{T_d^k},
\qquad
k=1,2,
\end{equation}
where $c_1(s)$ and $c_2(s)$ measure the local first- and second-order
variation of the unitary family. We denote by $\hat c_k(s)$ the corresponding
local maxima over neighboring steps.

Let $\check{\Delta}_2(s)$ denote the minimum eigenphase separation, over the
relevant neighboring walk operators, between the instantaneous eigenspace
being followed and the remainder of the spectrum. Provided that the spectral-separation assumptions hold and
\begin{equation}
\label{eq:run_cond}
T_d\ge\max_s
\frac{4\widehat c_1(s)}{\check\Delta_2(s)},    
\end{equation}
a simplified form of the discrete adiabatic theorem gives
\begin{align}
\left\| \mathcal{U}(s)-\mathcal{U}_A(s) \right\|
\leq {}& \frac{12\hat c_1(0)}{T_d\check{\Delta}_2(0)^2} +
\frac{ 12\hat c_1(s)}{T_d\check{\Delta}_2(s)^2}
+  \frac{ 6\hat c_1(s)}{T_d\check{\Delta}_2(s)}
    \nonumber\\
&+ 305\sum_{j=1}^{n-1} \frac{\hat c_1(j/T_d)^2 }{T_d^2 \check{\Delta}_2(j/T_d)^3}
+ 44 \sum_{j=0}^{n-1} \frac{ \hat c_1(j/T_d)^2}{T_d^2 \check{\Delta}_2(j/T_d)^2} + 32 \sum_{j=1}^{n-1} \frac{\hat c_2(j/T_d)}{T_d^2 \check{\Delta}_2(j/T_d)^2}.
\label{eq:discrete_adiabatic_bound}
\end{align}

The theorem bounds adiabatic error through both spectral separation and the first and second discrete variations of the walk. Below we give sufficient conditions under which these contributions retain inverse-linear eigenphase-gap scaling. The eigenphase amplification alone does not supply such a runtime guarantee.

\section{Szegedy quantum walks}
\label{sec:QW_S}

The quantum walk operator $U_w$ used for Gibbs-state preparation is the one introduced in \cite{Lemieux2020},
\begin{equation}
\label{eq:S_walk}
U_w(\beta) = R W(\beta)^{\dagger} S W(\beta),
\end{equation}
where $W(\beta)$ is the state-preparation unitary
\begin{equation}
\label{eq:prep_W}
W(\beta)\ket{x}\ket{0} = \sum_y \sqrt{\mathcal W_{yx}(\beta)} \ket{y}\ket{x}.
\end{equation}
Here, $\mathcal W(\beta)$ is the reversible Metropolis transition matrix at inverse temperature $\beta$, $S$ is the swap operator between the two registers, and
\begin{equation}
R=2\Pi_0-I,
\end{equation}
with $\Pi_0$ the projector onto the subspace spanned by $\{\ket{x}\ket{0}\}$.

There is a simple way to relate the spectrum of the quantum walk to that of
the underlying reversible Markov chain. We introduce the operator
\begin{equation}
X(\beta) = \Pi_0 W(\beta)^\dagger S W(\beta)\Pi_0,
\label{eq:discriminant_operator}
\end{equation}
usually referred to as the discriminant operator. Using the definition of
the state-preparation unitary in \cref{eq:prep_W}, we obtain
\begin{align}
X(\beta) &= \sum_{x,y}
\sqrt{ \mathcal W_{xy}(\beta)\mathcal W_{yx}(\beta)}\ket{y}\bra{x}\otimes
\ket{0}\bra{0}.
\end{align}
For a reversible Markov chain, detailed balance gives $\mathcal W_{yx}(\beta)\pi_\beta(x)  = \mathcal W_{xy} (\beta)\pi_\beta(y)$
Therefore,
\begin{equation}
\label{eq:X_mat}
X(\beta) = \sum_{x,y} \sqrt{ \frac{\pi_\beta(x)} {\pi_\beta(y)}}  \mathcal W_{yx}(\beta)\ket{y}\bra{x}\otimes
    \ket{0}\bra{0}.
\end{equation}
Restricting to the subspace $\mathcal E_0$, this can be written as $ X(\beta) = D_{\pi_\beta}^{-1/2} \mathcal W(\beta)  D_{\pi_\beta}^{1/2}$, where $D_{\pi_\beta} =
\operatorname{diag}(\pi_\beta)$. Thus, $X(\beta)$ is related to $\mathcal W(\beta)$ by a similarity
transformation and therefore has exactly the same eigenvalues.

Let
\begin{equation}
X(\beta)\ket{\phi_k(\beta)} = \lambda_k(\beta)
    \ket{\phi_k(\beta)}.
\end{equation}
For a reversible and ergodic chain, $1=\lambda_0(\beta) > \lambda_1(\beta) \geq \lambda_2(\beta) \geq\cdots$, 
where $\lambda_0=1$ corresponds to the stationary Gibbs distribution. The corresponding invariant subspace of the quantum walk
$U_w(\beta)$ has eigenvalues
\begin{equation}
e^{\pm i\theta_k(\beta)}, \qquad
\cos\theta_k(\beta) = \lambda_k(\beta),
\label{eq:QW_phase_mapping}
\end{equation}
up to the convention used for the definition of the Szegedy walk. For the stationary state, $\lambda_0(\beta)=1$,  and therefore $\theta_0(\beta)=0$. The corresponding eigenstate is the coherent encoding of the Gibbs distribution,
\begin{equation}
\ket{\pi_\beta} = \sum_x \sqrt{\pi_\beta(x)} \ket{x}.
\end{equation}

Notice that the walk $U_w(\beta)$ leaves the coherent Gibbs state $|\psi_\beta\rangle=|\pi_\beta\rangle\otimes|0\rangle$ unchanged. However, other states may also have eigenvalue $+1$. We, therefore, apply the adiabatic theorem to the space containing all these states, with projector
\begin{equation}
P_+(\beta)
=|\psi_\beta\rangle\langle\psi_\beta|
+P_{\mathrm{extra}}(\beta), \quad P_{\mathrm{extra}}(\beta)\Pi_0=0.   
\end{equation}
As proved in \cref{app:statio_eig}, the extra states lie outside the preparation subspace, and the ideal adiabatic evolution does not mix the Gibbs states with them. Thus, starting from the initial Gibbs state, the ideal evolution reaches the final Gibbs state. The adiabatic theorem bounds the actual evolution’s deviation from this ideal evolution.

Let $\delta(\beta) = 1-\lambda_1(\beta)$
denote the spectral gap of the underlying Markov chain. The eigenphase gap separating the entire eigenvalue $+1$ eigenspace from the remaining spectrum is
\begin{equation}
\Delta_w(\beta) = \arccos\!\left[ 1-\delta(\beta)\right],
\label{eq:QW_gap_exact}
\end{equation}
which is  the spectral separation used in the adiabatic theorem.
For a small Markov-chain gap, $\delta(\beta)\ll1$ it follows that
\begin{equation} 
\Delta_w(\beta) = \Theta\!\left( \sqrt{\delta(\beta)} \right).
\label{eq:QW_spectral_amplification}
\end{equation}

This is the spectral amplification produced by the quantized Markov chain:
a classical gap $\delta$ is mapped to a quantum-walk eigenphase gap
$\Delta_w=\Theta(\sqrt{\delta})$. However, this does not by itself imply a
quadratic speedup. Since $U_w(\beta)$ is unitary, repeated application does not
mix to the Gibbs state, and $1/\Delta_w$ is not a mixing time. Likewise, for a
temperature-dependent product $ U_w(\beta_T)\cdots U_w(\beta_1)$,  the square-root gap advantage is preserved only if the walk varies slowly enough along the path. In the discrete adiabatic theorem this variation is controlled by the finite-difference coefficients $c_1$ and $c_2$, which therefore determine whether the quadratic gap improvement survives. We then, next, show under with conditions the finite-difference coefficients allow us to have the spectrum amplification to prepare the coherent Gibbs state.

\section{Gap-adapted quantum schedules}
\label{sec:gap_adap}

To determine the conditions under which the quantum walk preserves the spectral amplification, we evaluate the first- and second-order difference operators, $k=1,2$, defined in \cref{eq:dif_op} for the Szegedy quantum walk. We begin by deriving bounds on the corresponding coefficients $c_k$, $k=1,2$, assuming bounds are available for the finite differences of the state-preparation unitary $W$, where the proof is given in \cref{app:proof_lemqw_finiy}.

\begin{lemma}[Finite-difference bounds for the Szegedy walk] \label{lem:qw_finite_differences_a}
Let $U_w(s) = R W(s)^\dagger S W(s)$ where $R$ and $S$ are independent of $s$. Suppose that
\begin{equation}
\left\| D^{(1)}W(s) \right\| \leq\frac{a_1(s)}{T_d}, \quad  \left\| D^{(2)}W(s) \right\| \leq
    \frac{a_2(s)}{T_d^2}.
\end{equation}
Then
\begin{equation}
 \left\| D^{(1)}U_w(s)\right\| \leq \frac{c_1(s)}{T_d},\quad \left\|
    D^{(2)}U_w(s)
    \right\|
    \leq
    \frac{c_2(s)}{T_d^2},
\end{equation}
with
\begin{equation}
\label{eq:c1_c2}
c_1(s) \leq 2a_1(s),\quad c_2(s) \leq 2a_2(s) + 2a_1(s)
    a_1\!\left(s+\frac{1}{T_d}\right).
\end{equation}
\end{lemma}
The finite-difference bounds above are fully discrete; when the walk is generated by a smooth path $W(\beta(s))$, they can be related to derivatives of $W$ and of the schedule $\beta(s)$ through a Taylor expansion, as stated in the following lemma, where the proof is given in \cref{app:proof_finit_difc}.

\begin{lemma}[Smooth-path bounds]
\label{lem:qw_finite_differences_L}
Let $W(s)=W(\beta(s))$ and define
\begin{equation}
L_1(\beta) = \left\| \partial_\beta W(\beta) \right\|, \qquad L_2(\beta) = \left\| \partial_\beta^2 W(\beta) \right\|.
\end{equation}
For a smooth schedule $\beta(s)$,
\begin{equation}
a_1(s) = |\dot\beta(s)| L_1(\beta(s)) + O(T_d^{-1}),
\end{equation}
and
\begin{equation}
a_2(s) \leq |\ddot\beta(s)| L_1(\beta(s)) + |\dot\beta(s)|^2 L_2(\beta(s)) + O(T_d^{-1}).
\end{equation}
Consequently,
\begin{equation}
c_1(s) \leq 2|\dot\beta(s)| L_1(\beta(s)) + O(T_d^{-1}),
    \label{eq:c1_L}
\end{equation}
and
\begin{equation} 
c_2(s) \leq 2|\ddot\beta(s)|L_1(\beta(s))+2|\dot\beta(s)|^2\left[L_2(\beta(s))+L_1(\beta(s))^2\right]+O(T_d^{-1}).
\label{eq:c2_L}
\end{equation}
\end{lemma}

The previous lemmas determine how the variation of the Szegedy walk depends on the temperature schedule. We now choose the schedule so that this variation is suppressed in regions where the eigenphase gap becomes small and show that the square-root spectral amplification is then preserved in
the discrete adiabatic evolution.

For $1<p<2$, define the gap-adapted schedule by
\begin{equation}
\label{eq:quantum_gap_schedule}
\dot{\beta}(s)
= \mathcal N_p
\frac{\Delta_w(\beta(s))^p}{L_1(\beta(s))},
\end{equation}
where the normalization constant $\mathcal N_p$ is chosen such that
$\beta(0)=\beta_0$ and $\beta(1)=\beta_f$. Explicitly,
\begin{equation}
\label{eq:Np}
\mathcal N_p
= \int_{\beta_0}^{\beta_f} \frac{L_1(\beta)}
{\Delta_w(\beta)^p} \,d\beta.
\end{equation}
Substitution into the leading term of \cref{eq:c1_L} suggests the scaling below. Applying this scaling uniformly to the discrete coefficients requires additional control of the finite-difference remainders, which we impose explicitly in \cref{thm:quantum_gap_scaling}.
\begin{equation}
\label{eq:c1_gap_schedule}
c_1(s) = O\left( \mathcal N_p \Delta_w(\beta(s))^p \right).
\end{equation}

Now in order to get the second order coefficient, we first differentiate \cref{eq:quantum_gap_schedule} with respect to $s$, that gives
\begin{equation}
\ddot{\beta}(s) =
\mathcal N_p^2 p\,
\frac{ \Delta_w(\beta)^{2p-1}
\partial_\beta\Delta_w(\beta)}{ L_1(\beta)^2} -
\mathcal N_p^2 \frac{ \Delta_w(\beta)^{2p}
\partial_\beta L_1(\beta)}{L_1(\beta)^3
}.
\label{eq:beta_ddot_schedule}
\end{equation}
We, then, assume that along the interpolation the walk satisfies the regularity conditions
\begin{equation}
\label{eq:quantum_regular_conditions}
\frac{|\partial_\beta\Delta_w|}{L_1}=O(1),
\qquad \frac{|\partial_\beta L_1|}{L_1^2}=O(1), \qquad
\Delta_w\frac{L_2}{L_1^2}=O(1).
\end{equation}
Combining \cref{eq:beta_ddot_schedule} with the regularity conditions in \cref{eq:quantum_regular_conditions} gives the following scaling for the leading derivative contributions in \cref{eq:c2_L}. Uniform bounds on the full discrete coefficients are assumed separately in \cref{thm:quantum_gap_scaling}.
\begin{equation}
\label{eq:c2_gap_schedule}
c_2(s) = O\left( \mathcal N_p^2 \Delta_w(\beta(s))^{2p-1} +
\mathcal N_p^2 \Delta_w(\beta(s))^{2p} \right).
\end{equation}

\begin{theorem}[Gap scaling of the discrete quantum adiabatic evolution]\label{thm:quantum_gap_scaling}

Consider the temperature-dependent Szegedy walk $U_w(\beta(s))$ and assume
the hypotheses of Lemmas~\ref{lem:qw_finite_differences_a} and
\ref{lem:qw_finite_differences_L}. Let the schedule be chosen according to
\cref{eq:quantum_gap_schedule}, with $1<p<2$, and suppose that the regularity conditions in \cref{eq:quantum_regular_conditions} hold.
Assume furthermore that the local multistep gap entering the discrete adiabatic theorem is uniformly comparable to the instantaneous Szegedy gap,
$\check{\Delta}_2(s)=\Theta(\Delta_w(s))$, and define
$\Delta_{w,\min}=\min_s\Delta_w(s)$. Let $s_j=j/T_d$. Assume that the neighboring-step maxima of the discrete coefficients satisfy
\begin{equation}
\label{eq:c1c2_M}
\hat c_1(s_j)
\le A_1\mathcal N_p\Delta_w(s_j)^p, \quad    \hat c_2(s_j)
\le A_2\mathcal N_p^2
\left[ \Delta_w(s_j)^{2p-1} +\Delta_w(s_j)^{2p}
\right],
\end{equation}
at all indices where the corresponding coefficients enter the adiabatic bound. Assume also the uniform quadrature estimate
\begin{equation}
\frac1{T_d}
\sum_{j=0}^{T_d-1}\Delta_w(s_j)^{2p-3}
\le
C_{\mathrm q}\int_0^1\Delta_w(s)^{2p-3}\,ds.    
\end{equation}
The constants $A_1,A_2,C_{\mathrm q}$, the constants in the gap comparison, and those in the integral assumptions below are independent of $T_d$ and $\Delta_{w,\min}$ over the family under consideration. The runtime is assumed to satisfy \cref{eq:run_cond}. If the schedule integrals satisfy
\begin{equation}
\label{eq:schedule_integral_condition1}
\mathcal N_p
= O\left( \Delta_{w,\min}^{\,1-p}\right),
\end{equation}
and
\begin{equation}
\label{eq:schedule_integral_condition2}
\mathcal N_p
\int_{\beta_0}^{\beta_f}
\frac{L_1(\beta)}
{\Delta_w(\beta)^{3-p}}
\,d\beta
= O\left(
\Delta_{w,\min}^{-1}
\right),
\end{equation}
then the discrete adiabatic error satisfies
\begin{equation}
\label{eq:quantum_error_scaling}
\left\|
\mathcal U(1)-\mathcal U_A(1)
\right\|
=
O\left( \frac{1}
{T_d\Delta_{w,\min}}
\right).
\end{equation}
Consequently, for fixed target adiabatic error,
\begin{equation}
\label{eq:quantum_runtime_phasegap}
T_d
=
O\left(
\Delta_{w,\min}^{-1}
\right).
\end{equation}
Since
$\Delta_{w,\min}=\Theta(\sqrt{\delta_{\min}})$,
the corresponding dependence on the minimum Markov-chain gap is
\begin{equation}
\label{eq:quantum_runtime_markovgap}
T_d = O\left(\delta_{\min}^{-1/2}\right).
\end{equation}
\end{theorem}
The bounds on the first and second finite differences, including the term $2|\dot\beta|^2 L_1^2$ in Eq.~\eqref{eq:c2_L}, are smooth-path estimates. The normalization $\mathcal N_p$ and the uniform discrete, quadrature, and path-integral conditions remain separate assumptions: the numerical schedules starting at $\beta=0$ are not certified by this theorem without an endpoint and uniformity analysis. The smooth-path estimates motivate these discrete hypotheses but do not, by themselves, establish their uniformity as the minimum gap decreases. The theorem, where the proof is given in \cref{app:proof_gap}, therefore identifies sufficient family of asymptotic schedules to retain the square-root spectral amplification of the Szegedy walk,
\begin{equation}
\label{eq:quantum_schedule_asymptotic_final}
\dot{\beta}(s) \propto \frac{\Delta_w(\beta)^p}{L_1(\beta)}
=\Theta\!\left(\frac{\delta(\beta)^{p/2}}{L_1(\beta)} \right),
\qquad
1<p<2.
\end{equation}
Here, $L_1(\beta)=\|\partial_\beta W(\beta)\|$, measures how quickly the state-preparation unitary changes with $\beta$. The quantum schedule therefore slows down both when the eigenphase gap becomes small and when the walk itself changes fast.

For the Metropolis construction in \cref{eq:prep_W}, the transition probabilities determine the variation restricted to the preparation subspace. On the state-preparation subspace,
\begin{equation}
\left\|
\partial_\beta W(\beta)\Pi_0
\right\|
= \frac{1}{2}
\max_x \left[ \sum_y \frac{ \left( \partial_\beta\mathcal W_{yx}(\beta)
\right)^2}{\mathcal W_{yx}(\beta)}
\right]^{1/2}.
\end{equation}
Notice that the equation above determines the derivative norm restricted to the preparation subspace $\|(\partial_{\beta}W)\Pi_0\|\leq \|\partial_{\beta}W\|$. However, determining the full-unitary norm requires specifying the unitary implementation. We will do it later for the LHPST construction used in the numerical experiments.

In the limit $p\rightarrow1^+$, the quantum schedule approaches
\begin{equation}
\label{eq:quantum_schedule_p1}
\dot{\beta}(s)
\sim \frac{\sqrt{\delta(\beta)}}{L_1(\beta)},
\end{equation}
up to an arbitrarily small additional power of the gap. In the small-gap regime, this can be compared with the classical Tracking schedule,
\begin{equation}
\dot{\beta}_{\rm C}(s)
\propto \frac{\delta(\beta)}
{\sqrt{\operatorname{Var}_{\pi_\beta}(E)}}.
\end{equation}
Both schedules slow the evolution where the local tracking cost is large. For classical Gibbs dynamics, that cost depends on stationary-state fluctuations and relaxation; for the quantum walk, it also depends on the full preparation-unitary variation. The eigenphase amplification motivates the different gap dependence, but the quantum runtime conclusion requires all the uniform hypotheses of \cref{thm:quantum_gap_scaling}, not only the derivative conditions.

\section{Numerical results: scaling and schedule transfer}
\label{sec:numerical_quantum_scaling}

We first compare the classical Tracking evolution with the gap-adapted quantum evolution for SK instances with system sizes $n=6,\ldots,13$. For each size, we generate 40 independent disorder realizations, using the same Hamiltonian instances for both classical and quantum calculations. To ensure a consistent comparison across system sizes, the SK couplings are divided by $\sqrt{n}$. Specifically, the implemented energy is
\begin{equation}
 E_r(\sigma) =  \frac{2}{\sqrt{n}}\sum_{1\leq i<j\leq n} J_{ij}^{(r)}\sigma_i\sigma_j ,
\end{equation}
where the independent coupling take values $\pm1$ with equal probability. Both algorithms use the same Hamiltonian and target Gibbs distribution, with the inverse-temperature interval fixed to $\beta\in[0,0.5]$. This normalization and sign convention are specific to the paired classical--quantum scaling experiment, and should not be identified with the Hamiltonian convention used in the separate classical schedule benchmark in Sec.~\ref{sec:Num_valid_cl}. Comparisons of schedules within either experiment use a common energy convention.

For the classical dynamics, the temperature path follows the Tracking schedule derived in \cref{eq:classical_sched_Trac}. For the quantum evolution, we use the LHPST walk as given in \cite{Lemieux2020,costa2025assessing}
\begin{equation}
\widetilde{U}_w(\beta) = RV^{\dagger}B(\beta)^{\dagger}F B(\beta)V,    
\end{equation}
where $V$ prepares the proposed moves, $B(\beta)$ encodes their Metropolis acceptance probabilities, $F$ performs the move conditioned on acceptance, and $R$ reflects about the initial auxiliary state. Only $B(\beta)$ depends on temperature. Defining the full preparation unitary as $W_{\mathrm{LHPST}}(\beta)=B(\beta)V$, we construct the schedule from its discrete variation. On a preprocessing grid $\{\beta_i\}$, the local schedule weights are proportional to
\begin{equation}
\label{eq:Wvari}
\frac{\|W_{\mathrm{LHPST}}(\beta_{i+1})-W_{\mathrm{LHPST}}(\beta_i)\|}
{\Delta_{w,i}^{p}} =\frac{\|B(\beta_{i+1})-B(\beta_i)\|}
{\Delta_{w,i}^{p}},
\qquad p=1.1,
\end{equation}
where $\Delta_{w,i}$ denotes the local eigenphase-gap estimate. This discrete construction is particularly useful near $\beta=0$, where the derivative of the state-preparation operator need not remain regular.

For the normalized binary SK model, the variation of the full preparation unitary satisfies
\begin{equation}
D_{n,i} \coloneqq \|W_{\mathrm{LHPST}}(\beta_{i+1}) -W_{\mathrm{LHPST}}(\beta_i)\| =\|B(\beta_{i+1})-B(\beta_i)\|,    
\end{equation}
where the equality in the equation above is because, as we have in  \cref{eq:Wvari}, $V$ is  temperature independent and unitary. The acceptance unitary is block diagonal, so its difference norm is the maximum over the individual acceptance blocks. As derived in \cref{app:Analy_Diff}.
\begin{equation}
\label{eq:SK_discrete_L1}
D_{n,i}
= \max_{m\in\mathcal M_n}\sqrt{
\left(\sqrt{1-e^{-\beta_{i+1}\Delta E_m}}
- \sqrt{1-e^{-\beta_i\Delta E_m}}\right)^2+
\left(e^{-\beta_{i+1}\Delta E_m/2}
- e^{-\beta_i\Delta E_m/2}\right)^2
},
\end{equation}
where $\Delta E_m=4m/\sqrt n$ and $\mathcal M_n$ contains the possible positive magnitudes of the local field
$h_k(x)=\sum_{j\neq k}J_{kj}x_j$.
For this implementation and $\beta_i>0$,
\begin{equation}
L_1(\beta_i)
= \|\partial_\beta W_{\mathrm{LHPST}}(\beta_i)\| = \|\partial_\beta B(\beta_i)\|
\approx
\frac{D_{n,i}}{\beta_{i+1}-\beta_i}.
\end{equation}
where the approximation holds as the temperature increment tends to zero. We use $D_{n,i}$ directly to construct the discrete schedule. This finite difference remains well defined at $\beta_i=0$, although the derivative norm is singular there. Applying the adiabatic runtime theorem at this endpoint requires a separate regularity argument.

\paragraph{Classical--quantum scaling.}

To compare the two evolutions under the same success-probability and approximation-ratio thresholds, we define the race cost $T_\star$ as the minimum number of evolution steps required to satisfy both $q_{\rm succ}\geq0.99$ and $q_{\rm AEAR}\geq0.99$, using the criteria introduced in \cref{eq:target_acc}. At each system size, we report the median race cost over the disorder realizations. The results, together with the corresponding classical and quantum schedules, are shown in \cref{fig:classical_quantum_race}.

The finite-size dependence is well described by exponential fits, $T_\star=Ae^{bn}$. For classical Tracking, we obtain
\begin{equation}
T_C\simeq161.8\,2^{0.3686n},
\qquad
T_Q\simeq47.5\,2^{0.1828n},
\end{equation}
for the quantum evolution. The fitted quantum growth exponent is therefore approximately half the classical exponent,
$b_Q/b_C\simeq0.496$. This behavior is consistent with the square-root eigenphase amplification of the Szegedy walk, although the fits characterize only the investigated finite-size regime and do not establish an asymptotic quadratic speedup.

\begin{figure}[t]
    \centering
    \includegraphics[width=0.45\linewidth]{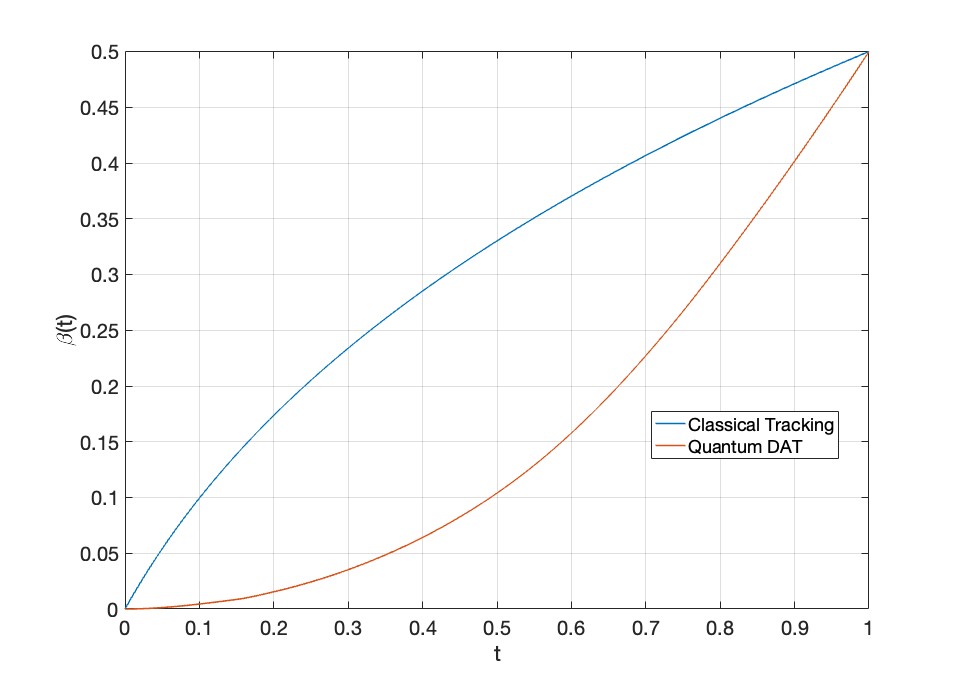}
    \includegraphics[width=0.45\linewidth]{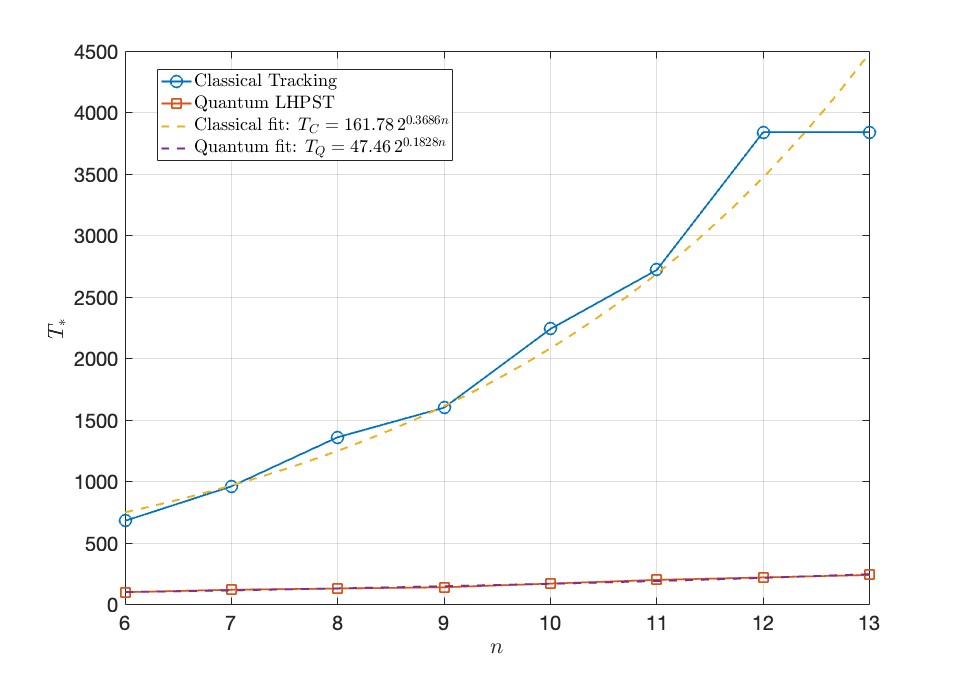}
    \caption{Classical and quantum schedules (left) and race-to-target costs for classical Tracking and gap-adapted LHPST evolution (right). Points show the median cost over disorder realizations, and dashed lines denote exponential finite-size fits. The ratio of fitted growth exponents is $b_Q/b_C\simeq0.496$.}
    \label{fig:classical_quantum_race}
\end{figure}

\paragraph{Transferability to larger SK instances.}

We next investigate whether the quantum schedule can be constructed for larger systems without computing their instantaneous Markov spectra. Following the classical extrapolation procedure, we estimate the spectral-gap dependence on system size from smaller SK instances and use the predicted gaps to construct the quantum schedules for $n=12,14,16$. The corresponding eigenphase gaps are obtained through
\begin{equation}
\widehat{\Delta}_{w,n}(\beta)
=
\arccos\left[1-2\widehat{\delta}_{L,n}(\beta)\right].
\end{equation}
The variation $D_{n,i}$ is calculated analytically using \cref{eq:SK_discrete_L1}, so no instance-specific diagonalization is required for constructing the larger-system quantum schedules.

We compare the predicted LHPST evolution with classical Linear and Tracking schedules over 20 independent disorder realizations at each system size. At $T=1600$, the mean terminal total-variation error of LHPST remains approximately $3\times10^{-4}$, compared with errors of order $10^{-3}$ for the classical methods. These results indicate that spectral information extracted from smaller systems can produce effective quantum schedules for larger SK instances. Details of the extrapolation and additional numerical results are provided in \cref{app:SK_extap}.

\begin{figure}[t]
    \centering
    \includegraphics[width=0.5\linewidth]{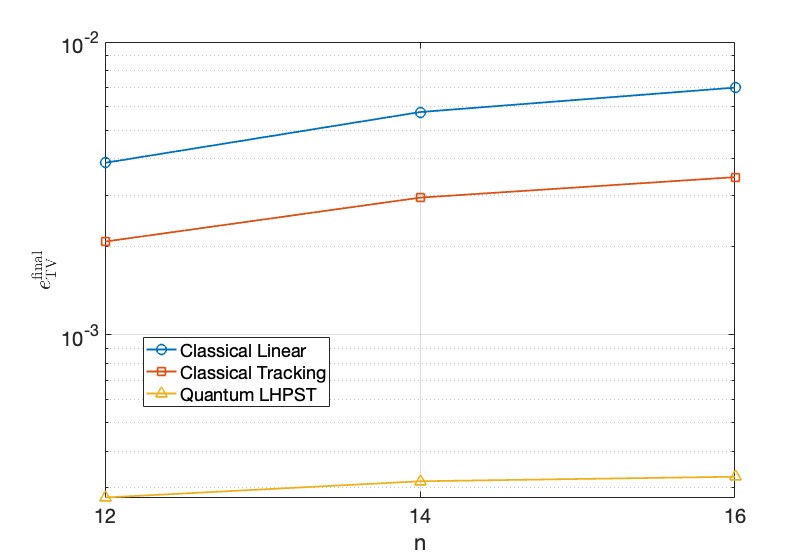}
    \caption{
    Mean terminal total-variation error,
    $e_{\mathrm{TV}}^{\mathrm{final}}$, as a function of the
    SK system size $n$ for classical Linear, classical Tracking,
    and quantum LHPST evolution.
    The spectral-gap model is constructed using smaller systems
    ($n=6,8,10$) and subsequently transferred to
    $n=12,14,16$ to determine the adiabatic schedules.
    Results are averaged over 20 independent disorder
    realizations for each system size, with $T=1600$
    evolution steps and $\beta_F=0.5$.
    The predicted quantum schedules yield substantially smaller
    preparation errors without requiring instance-specific
    spectral diagonalization.
    }
    \label{fig:SK_quantum_transfer}
\end{figure}

\section{Conclusion}
\label{sec:conc}

We have developed a discrete adiabatic framework connecting classical Markov-chain tracking, geometric annealing schedules, and coherent quantum Gibbs-state preparation. For reversible time-inhomogeneous Markov chains, the tracking theorem separates contraction of accumulated deviations from the variation of the instantaneous stationary distribution, retaining the local spectral gaps and changes in weighted norm. Along a smooth Gibbs path, equilibrium energy fluctuations determine the stationary-state variation and lead to an effective tracking metric. Its constant-speed schedule minimizes the associated metric-energy functional, not necessarily the exact finite-time tracking or terminal preparation error.

The solvable two-state example validates the product-form tracking bound, and the SK numerical tests distinguish maximum pathwise error from terminal Gibbs-state accuracy. In the investigated classical instances, schedules incorporating gap information can improve terminal preparation, although the Tracking schedule is not uniformly optimal among the tested prescriptions. For $n=18$, schedules based on gaps extrapolated from smaller SK systems and energy fluctuations sampled on the target system also improve the observed endpoint metrics relative to linear cooling. These inferred gaps are empirical estimates, not certified lower bounds.

For the quantum dynamics, we analyze direct sequences of temperature-dependent Szegedy walks without intermediate phase-estimation or quantum-Zeno projections. Spectral amplification alone does not ensure improved discrete adiabatic complexity, because first- and second-order operator variations also enter the error estimate. Under the explicit uniform discrete, spectral-comparison, quadrature, and path-integral hypotheses, gap-adapted schedules give a sufficient fixed-accuracy walk count $T_Q=O(\Delta_{w,\min}^{-1})=O(\delta_{+,\min}^{-1/2})$. This is a conditional result about inverse-gap dependence, not an unconditional total-cost quantum speedup.

The quantum SK experiments provide complementary finite-size observations. Over the sizes studied, the fitted growth exponent for the walk-step race to prescribed optimisation thresholds is approximately half that of classical Tracking; these empirical fits neither establish asymptotic scaling nor certify a Gibbs-distance bound at the race threshold. Separately, gap-informed quantum schedules transferred from smaller systems exhibit low terminal preparation errors on larger SK instances without instance-specific spectral diagonalization during schedule construction. The tested walk implementation is not yet covered uniformly by the theorem, notably because its acceptance-unitary derivative is singular at $\beta=0$.

The results help clarify one schedule-dependent mechanism relevant to earlier comparisons of classical and quantum optimization heuristics~\cite{costa2025assessing}. Important next steps are to establish uniform estimates for concrete problem families, handle the initial endpoint rigorously, and quantify the robustness and computational cost of predicting the local gaps. More broadly, the framework links finite-time Gibbs-state preparation to the balance between equilibrium fluctuations, spectral relaxation, and variation along an annealing path.

\section*{Acknowledgments}
The authors thank Stephen Jordan and Dominic W. Berry for helpful discussions. DA acknowledges funding from Innovation Program for Quantum Science and Technology via Project 2024ZD0301900, the support by The Fundamental Research Funds for the Central Universities, Peking University. KK worked on this project with funding from Molecular Quantum Solutions ApS and Sydney Quantum Academy.
%

%\bibliographystyle{apsrev4-1}	
%\bibliography{ref}

\appendix

\section{A bound for the norm-comparison constants}

We now expand with the definition of the norm-comparison constants, \cref{eq:norm_comp} 
\begin{align}
\|q\|_{j+1}&=\sqrt{\sum_x\frac{q(x)^2}{\pi_{k+1}(x)}}\nonumber\\
&=\sqrt{\sum_x\frac{q(x)^2}{\pi_{j}(x)}\frac{\pi_{j}(x)}{\pi_{j+1}(x)}}\nonumber\\
&=\sqrt{\sum_xw_j(x)\frac{\pi_{j}(x)}{\pi_{j+1}(x)}},    
\end{align}
where we have defined $w_{k}(x)\coloneqq q(x)^2/\pi_{k}(x)$. Using this, we can also write
\begin{equation}
\|q\|_{k}=\sqrt{\sum_xw_k(x)}.    
\end{equation}
Therefore
\begin{equation}
C_k \coloneqq \frac{\|q\|_{k+1}}{\|q\|_{k}}=\sup_{\substack{q\in H_0}\\q\neq 0}\sqrt{\frac{\sum_xw_{k}(x)\frac{\pi_{k}(x)}{\pi_{k+1}(x)}}{\sum_xw_k(x)}},    
\end{equation}

This is a weighted average of $\pi_{k}(x)/\pi_{k+1}(x)$ with weights $w(x)$. The weighted average of any sequence must be larger than its smallest value and smaller than its largest value, regardless of the weights. Therefore
\begin{equation}
\sqrt{\min_x\frac{\pi_{k}(x)}{\pi_{k+1}(x)}}\leq C_k\leq\sqrt{\max_x\frac{\pi_{k}(x)}{\pi_{k+1}(x)}}    
\end{equation}
These bounds rely only on the stationary distributions and are valid for any $q(x)$.

In particular, for the Gibbs distribution, we have
\begin{equation}
\frac{\pi_{\beta_j}(x)}{\pi_{\beta_{j+1}}(x)} = e^{h_j E(x)}\frac{Z(\beta_{j+1})}{Z(\beta_j)},     
\end{equation}
where $h_j=\beta_j-\beta_{j-1}$. Next, taking the maximum over $x$,
\begin{equation}
 C^2_j \leq e^{h_j E_{\text{max}}}\frac{Z(\beta_{j+1})}{Z(\beta_j)}.
\end{equation}
We can now take the log,
\begin{equation}
2\log C_j \leq h_j E_{\text{max}} + \log Z(\beta_{j+1}) - \log Z (\beta_j).   
\end{equation}
Next we can use
\begin{equation}
\frac{d}{d \beta} \log Z(\beta)  = -\langle E \rangle_\beta,   \end{equation}
so
\begin{equation}
\log Z(\beta_{j+1}) - \log Z(\beta_{j}) = - \int^{\beta_{j+1}}_{\beta_j}db\langle E \rangle_b.     
\end{equation}
Therefore,
\begin{equation}
\log C_k \leq \frac{1}{2} \int^{\beta_{j+1}}_{\beta_j}db\left( E_{\text{max}} - \langle  E\right) )_b.   
\end{equation}

\section{Proofs of the Classical Tracking Results}

\subsection{Proof of \cref{lemma:contractiongap}}

\begin{proof}
\label{proof:contractiongap}
Since $\mathcal W$ is reversible with respect to $\pi$, it is self-adjoint
with respect to the weighted inner product
\begin{equation}
\langle u,v\rangle_{\pi^{-1}} \coloneqq
    \sum_{x\in\Omega}\frac{u(x)v(x)}{\pi(x)}.
\end{equation}
Consequently, $\mathcal W$ admits an orthonormal eigenbasis $\{v_j\}_{j=1}^N$ with respect to this inner product, with real eigenvalues. We choose $v_1=\pi$, noting that
\begin{equation}
\|\pi\|^2_{\pi^{-1}} =1.    
\end{equation}

Now let $r\in\mathcal H_0$, since
\begin{equation}
\langle r,\pi\rangle_{\pi^{-1}}
= \sum_{x\in\Omega}r(x) = 0,
\end{equation}
the vector $r$ is orthogonal to $\pi$, so it has no component along the stationary eigenvector. Hence it can
be expanded only in the remaining eigenspaces of $\mathcal{W}$:
\begin{equation}
\label{eq:eigvec_W}
r=\sum_{j\geq 2}c_j v_j,    
\end{equation}
where $v_1=\pi$. Applying $\mathcal W$ gives
\begin{equation}
\mathcal W r = \sum_{j\geq 2}c_j\lambda_j v_j.
\end{equation}
Thus,
\begin{align}
\|\mathcal W r\|_{\pi^{-1}}^2 &=
    \sum_{j\geq 2}|c_j|^2|\lambda_j|^2\nonumber\\
    &\leq \lambda_*^2 \sum_{j\geq 2}|c_j|^2 \nonumber\\
    &= \lambda_*^2\|r\|_{\pi^{-1}}^2.
\end{align}
Taking square roots yields
\begin{equation}
 \|\mathcal W r\|_{\pi^{-1}}
    \leq
    \lambda_*\|r\|_{\pi^{-1}}.    
\end{equation}
\end{proof}

\subsection{Proof of \cref{thm:tracking_changing_norm}}\label{proof:track}

\begin{proof}
Define the tracking error
\begin{equation}
    e_k\coloneqq p_k-\pi_k .
\end{equation}
Since both $p_k$ and $\pi_k$ are probability distributions, we have
$e_k\in\mathcal H_0$ for every $k$. Also, since $\pi_{k-1}$ and $\pi_k$ are
probability distributions,
\begin{equation}
    \pi_{k-1}-\pi_k\in\mathcal H_0 .
\end{equation}

Using the evolution equation and stationarity of $\pi_k$ under $\mathcal W_k$,
we obtain, for $k=1,\ldots,T$,
\begin{align}
e_k &= p_k-\pi_k \nonumber\\
&= \mathcal W_k p_{k-1}-\pi_k \nonumber\\
&= \mathcal W_k p_{k-1}-\mathcal W_k\pi_k \nonumber\\
&= \mathcal W_k(p_{k-1}-\pi_k) \nonumber\\
&= \mathcal W_k(p_{k-1}-\pi_{k-1}) + \mathcal W_k(\pi_{k-1}-\pi_k) \nonumber\\
&= \mathcal W_k e_{k-1} + \mathcal W_k(\pi_{k-1}-\pi_k).
\end{align}
Taking the instantaneous norm $\|\cdot\|_k$ and using the triangle inequality,
\begin{align}
\|e_k\|_k &\leq \|\mathcal W_k e_{k-1}\|_k + \|\mathcal W_k(\pi_{k-1}-\pi_k)\|_k .
\end{align}
Both $e_{k-1}$ and $\pi_{k-1}-\pi_k$ lie in $\mathcal H_0$, so by the contraction assumption in the norm $\|\cdot\|_k$,
\begin{equation}
\|e_k\|_k \leq (1-\delta_k)\|e_{k-1}\|_k + (1-\delta_k)\|\pi_{k-1}-\pi_k\|_k .
\end{equation}
The first term is measured in the norm $\|\cdot\|_k$, whereas the previous
error is naturally measured in the norm $\|\cdot\|_{k-1}$. By the definition of
the norm-comparison constant $C_{k-1}$,
\begin{equation}
\|e_{k-1}\|_k \leq C_{k-1}\|e_{k-1}\|_{k-1}.
\end{equation}
Hence
\begin{equation}
\label{eq:error_sched}
\|e_k\|_k \leq (1-\delta_k)C_{k-1}\|e_{k-1}\|_{k-1} +
(1-\delta_k)\|\pi_k-\pi_{k-1}\|_k .
\end{equation}
Define
\begin{equation}
E_k\coloneqq\|e_k\|_k, \qquad d_k\coloneqq \|\pi_k-\pi_{k-1}\|_k, \qquad a_k\coloneqq1-\delta_k .
\end{equation}
Then the recursion becomes
\begin{equation}
E_k \leq a_k C_{k-1}E_{k-1} + a_k d_k .
\end{equation}
Since $p_0=\pi_0$, we have $E_0=0$, so
\begin{equation}
E_1 \leq a_1d_1,    
\end{equation}
now for $E_2$
\begin{align}
 E_2 &\leq a_2C_1 E_1 + a_2d_2\nonumber\\
 &\leq a_2a_1C_1d_1  + a_2d_2,
\end{align}
and for $E_3$,
\begin{align}
 E_3 &\leq a_3C_2 E_2 + a_3d_3\nonumber\\
 &\leq a_3C_2\left(a_2a_1C_1d_1  + a_2d_2\right)  + a_3d_3 \nonumber\\
 & = a_3a_2a_1C_2C_1d_1 + a_3a_2C_2d_2 + a_3d_3\nonumber\\
 & = \prod_{j=1}^3 a_j\prod_{j=1}^2 C_j d_1  + \prod_{j=2}^3 a_j\prod_{j=2}^2 C_j d_2 + a_3d_3\nonumber\\
 &= \sum_{j=1}^3 \left(\prod_{k=j}^3 a_k\prod_{k=j}^2 C_k\right)d_j,
\end{align}
from where we can conclude the following expression for the total $T$
\begin{equation}
e_T \leq  \sum_{j=1}^T \left[
\left( \prod_{r=j}^T a_r \right)
\left( \prod_{r=j}^{T-1} C_r \right)\right]d_j.
\end{equation}
Substituting back $a_r=1-\delta_r$ and
$d_j=\|\pi_j-\pi_{j-1}\|_j$, we obtain
\begin{equation}
\|p_T-\pi_T\|_T \leq \sum_{j=1}^T\left[
    \left(\prod_{r=j}^T(1-\delta_r) \right)
\left( \prod_{r=j}^{T-1}C_r\right) \right]
\|\pi_j-\pi_{j-1}\|_j .
\end{equation}
Finally, using $1-\delta_r\leq e^{-\delta_r}$
and
\begin{equation}
\prod_{r=j}^{T-1}C_r = \exp\left( \sum_{r=j}^{T-1}\log C_r \right),
\end{equation}
we obtain
\begin{equation}
\|p_T-\pi_T\|_T \leq \sum_{j=1}^T \exp\!\left( -\sum_{r=j}^T\delta_r + \sum_{r=j}^{T-1}\log C_r \right) \|\pi_j-\pi_{j-1}\|_j .
\end{equation}
This proves the theorem.
\end{proof}

\subsection{Proof of Corollaries \ref{cor:uniform_gap_changing_norm},  \ref{cor:maximal_one_step_variation} and  \ref{cor:stable_effective_damping} }\label{app:proof_coroll5_6_7}

We start with the proof of \cref{cor:uniform_gap_changing_norm}

\begin{proof}
By the uniform bounds on the contraction gaps and norm-comparison
factors,
\begin{align}
\exp\left(
-\sum_{r=j}^T\delta_r
+\sum_{r=j}^{T-1}\log C_r
\right)
&\leq
\exp\left(
-(T-j+1)\delta_{\min}
+(T-j)\log C_{\max}
\right)
\nonumber\\
&=
e^{-(T-j+1)\delta_{\min}}
C_{\max}^{T-j}.
\end{align}
Substituting this estimate into the tracking bound of
\cref{thm:tracking_changing_norm} proves
\cref{eq:uniform_gap_bound}.
\end{proof}

Now we present the proof of \cref{cor:maximal_one_step_variation}.

\begin{proof}
Applying the definition of $\omega_\pi$ to
\cref{eq:uniform_gap_bound} gives
\begin{align}
\|p_T-\pi_T\|_T
&\leq
\omega_\pi
\sum_{j=1}^T
e^{-(T-j+1)\delta_{\min}}
C_{\max}^{T-j}
\nonumber\\
&=
\omega_\pi e^{-\delta_{\min}}
\sum_{j=1}^T
\left(
C_{\max}e^{-\delta_{\min}}
\right)^{T-j}.
\end{align}
Introducing $m=T-j$ gives
\cref{eq:uniform_geometric_sum}. Evaluating the geometric series proves
\cref{eq:uniform_finite_bound} when
$C_{\max}e^{-\delta_{\min}}\neq1$. When this ratio equals one, the sum
contains $T$ identical terms, which proves
\cref{eq:uniform_critical_bound}.
\end{proof}

We now present the proof of \cref{cor:stable_effective_damping}
\begin{proof}
From \cref{eq:damp_rate},
\begin{equation}
    C_{\max}e^{-\delta_{\min}}
    =
    e^{-\delta_{\min}+\log C_{\max}}
    =
    e^{-\Gamma}.
\end{equation}
Substitution into \cref{eq:uniform_finite_bound} gives
\cref{eq:effective_damping_finite_bound}. Since $\Gamma>0$, we have $0<e^{-\Gamma}<1$, 
and therefore
\begin{equation}
    1-e^{-T\Gamma}\leq 1.
\end{equation}
This proves \cref{eq:effective_damping_uniform_bound}. Finally, as $\Gamma\to0^+$,
\begin{equation}
1-e^{-\Gamma} = \Gamma+O(\Gamma^2),
\end{equation}
so
\begin{equation}
\frac{1}{1-e^{-\Gamma}} =
O\left(\frac{1}{\Gamma}\right).
\end{equation}
Substitution into \cref{eq:effective_damping_uniform_bound} proves \cref{eq:effective_damping_scaling}.
\end{proof}

\subsection{Proof of \cref{lemma:One_step}}

\begin{proof}
\label{proof:one_step}
For each $j$, by the fundamental theorem of calculus,
\begin{equation}
\label{eq:fund_theor}
\pi_{s_j} - \pi_{s_{j-1}} = \int_{s_{j-1}}^{s_j}\dot\pi_s ds.    \end{equation}
taking the norm $\|\cdot\|_j$, and using the triangle inequality 
\begin{equation}
\|\pi_{s_j} - \pi_{s_{j-1}}\| \leq \int_{s_{j-1}}^{s_j} \|\dot\pi_s \|ds.   
\end{equation}
Next, we use the uniform bound on $\|\dot\pi_s\|_j$, then
\begin{equation}
\|\pi_{s_j} - \pi_{s_{j-1}}\| \leq \int_{s_{j-1}}^{s_j} L_\pi  ds = \frac{L_{\pi}}{T}.   
\end{equation}
Thus, taking the maximum over $j=1,\cdots,T$ gives

\begin{equation}
    \omega_{\pi} \leq \frac{L_{\pi}}{T}.
\end{equation}
\end{proof}

\subsection{Proof of \cref{cor:uniform_gap_complexity}}

\begin{proof}
\label{proof:uniform_gap_complexity}
Let $s_j=j/T$ and $\pi(s)=\pi_{\beta(s)}$, the Taylor's formula with integral remainder gives
\begin{equation}
\pi(s_{j-1}) = \pi(s_j) - \frac{1}{T}\dot{\pi}(s_j) + \int^{s_j}_{s_{j-1}}(u - s_{j-1})\ddot{\pi}(u)du,   
\end{equation}
Taking the instantaneous norm at $s_j$,
\begin{align}
\|\pi_j -\pi_{j-1}\| &\leq  \frac{1}{T}\|\dot{\pi}(s_j)\|_j + \int^{s_j}_{s_{j-1}}(u-s_{j-1})\|\ddot{\pi}(u)\|_j\,du\nonumber\\
&\leq \frac{1}{T}\|\dot{\pi}(s_j)\|_j + \frac{M_2}{2T^2},
\end{align}
where $M_2$ is given in \cref{eq:A_M}. For the Gibbs family,
 $\dot{\pi}(s)  = -\dot{\beta}(s)\big(E-\langle E\rangle_{\beta(s)}\big)\pi(s)$, so
\begin{equation}
\|\dot{\pi}(s_j)\|_j  = |\dot{\beta}(s_j)|\sqrt{\text{Var}_{{\pi}_{\beta(s_j)}} (E)} \leq A, 
\end{equation}
where $A$ is given in \cref{eq:A_M}. Thus, from \cref{lemma:One_step}
\begin{equation}
\omega_{\pi} \leq  \frac{A}{T} +\frac{M_2}{2T^2}.
\end{equation}

Applying \cref{cor:stable_effective_damping} and $\|p_T-\pi_T\|_{\mathrm{TV}}\leq\tfrac12\|p_T-\pi_T\|_T$ yields the error bound in \cref{cor:uniform_gap_complexity}. Solving the resulting quadratic inequality in $T$ gives its sufficient runtime condition.

\end{proof}

\subsection{Proof of \cref{lemma:effective_metric}}
\label{proof:metric}

\begin{proof}
Let $t_j=j/T$ and $\beta_j=\beta(t_j)$, where $\beta:[0,1]\rightarrow[\beta_0,\beta_f]$
is an arbitrary smooth monotone schedule. Define
\begin{equation}
\varepsilon_j\coloneqq p_j-\pi_j,
\qquad
e_j\coloneqq \|\varepsilon_j\|_j.
\end{equation}
Since
\begin{equation}
p_j=\mathcal W_j p_{j-1}, \qquad
\mathcal W_j\pi_j=\pi_j,
\end{equation}
we have
\begin{align}
\varepsilon_j
&= \mathcal W_jp_{j-1}-\pi_j \nonumber\\
&= \mathcal W_j \left(\varepsilon_{j-1}+\pi_{j-1}-\pi_j
    \right).
\end{align}
Using contraction on the zero-mass subspace in the instantaneous weighted norm,
\begin{equation}
    \|\mathcal W_j q\|_j
    \leq
    (1-\delta_j)\|q\|_j,
\end{equation}
we obtain
\begin{equation}
e_j \leq (1-\delta_j) \left(
\|\varepsilon_{j-1}\|_j + \|\pi_j-\pi_{j-1}\|_j \right).
\end{equation}
By the norm-comparison inequality,
\begin{equation}
\|\varepsilon_{j-1}\|_j \leq C_{j-1}\|\varepsilon_{j-1}\|_{j-1},
\end{equation}
and hence
\begin{equation}
\label{eq:local_tracking_recursion}
e_j \leq (1-\delta_j)C_{j-1}e_{j-1} + (1-\delta_j)
\|\pi_j-\pi_{j-1}\|_j.
\end{equation}
For the Gibbs path,
\begin{equation}
 \partial_\beta\pi_{\beta}(x) = -(E(x) - \langle E\rangle_{\beta})\pi_{\beta}(x),\quad \|\partial_\beta\pi_\beta \|_{\pi^{-1}_\beta } = V(\beta).   
\end{equation}
Thus, writing $h_j=\beta_j-\beta_{j-1}$, Taylor expansion gives
\begin{equation}
\label{eq:Vh_2}
\|\pi_j - \pi_{j-1}\|_j = |h_j|V_j + O(h^2_j).   
\end{equation}
We rewrite the exact recurrence \cref{eq:local_tracking_recursion} as
\begin{equation}
 e_j \leq \rho_j e_{j-1} + b_j,   
\end{equation}
where
\begin{equation}
\label{eq:rho_b}
 \rho_j = (1-\delta_j )C_{j-1} , \quad b_j = (1-\delta_j)\|\pi_j - \pi_{j-1}\|_j,   
 \end{equation}
 so by \cref{eq:Vh_2} $b_j=(1-\delta_j)|h_j|V_j + O(h^2_j)$. Thus, the recurrence $B_j=\rho_jB_{j-1} + b_j$ agrees with the envelope defined in \cref{eq:B_env}.

Now, let $t_j=j/T$ and define
 \begin{equation}
 a(t) = |\dot{\beta}(t)|V(\beta(t)), \quad D(t) = \frac{1 - \delta(\beta(t))}{\delta(\beta(t))}a(t).    
 \end{equation}
 Thus, from the smoothness of the schedule function and from $C_{j-1}= 1 + O(T^{-1})$, we can rewrite \cref{eq:rho_b} as
 \begin{equation}
  \rho_j =1 - \delta_j + O(T^{-1}),\quad b_j = \frac{(1-\delta_j)a(t_j)}{T} + O(T^{-2}). 
 \end{equation}

 We now set $z_j=D(t_j)/T$, since by construction $D(t)$ is continuously differentiable on the compact interval $[0,1]$, $D'$ is bounded, so $D$ is Lipschitz, so we can write
 \begin{equation}
  z_{j-1} = z_j + O(T^{-2}),   
 \end{equation}
 so
 \begin{equation}
     \rho_j z_{j-1} + b_j -z_j = \frac{-\delta_jD(t_j) + (1-\delta_j)a(t_j)}{T} + O(T^{-2}) =O(T^{-2}),
 \end{equation}
 where $\delta_jD(t_j) = (1-\delta_j)a(t_j)$. 
 
 From the previous definition and bounds, we can now estimate the following bound, $|B_j-z_{j-1}|$
\begin{align}
|B_j - z_j| & = |\rho_j B_{j-1} + b_j - z_j |   \nonumber\\
&= |\rho_j (B_{j-1} - z_{j-1}) + (\rho_jz_{j-1} + b_j- z_j) |\nonumber\\
&\leq |\rho_j (B_{j-1} - z_{j-1})| + | \rho_jz_{j-1} + b_j- z_j|\nonumber \\
&\leq  q| (B_{j-1} - z_{j-1})| + \frac{L}{T^2},
\end{align}
where $L$ is independent of $j$ and $T$. Notice that the last inequality comes from the fact that for a sufficiently large $T$ there exists a fixed $q \in (0,1)$ such that $\rho_j \leq q$. Thus, iterating the inequality above, using $z_j =D(t_j)/T$ and \cref{eq:B_env}, we find
\begin{equation}
|B_j - z_j| \leq \frac{q^j|D(0)|}{T} + \frac{L}{T^2 (1-q)}.    
\end{equation}
Therefore, uniformly  for $t_j \in [t_0,1]$, with any fixed $t_0>0$,
\begin{equation}B_j=\frac1T
\frac{[1-\delta(\beta_j)]|\dot\beta(t_j)|V(\beta_j)}
{\delta(\beta_j)} + O(T^{-2}), 
\end{equation}
where $e_j \leq B_j$. 

The preceding bound estimates
\begin{equation}
D(\beta,\dot{\beta}) =  \frac{(1-\delta(\beta))V(\beta)}{\delta(\beta)}|\dot{\beta}|,
\end{equation}
as the leading coefficient of the local tracking bound, so that $B_j =D(\beta_j,\dot{\beta}(t_j) )/T +O(T^{-2})$. To express this cost geometrically, we define a metric whose speed along the Gibbs path equals $D$. In one dimension, the squared metric speed is $g_{\beta \beta}(\beta)\dot{\beta}^2$, so 
\begin{equation}
 g_{\beta \beta}^{\text{track}}(\beta)\dot{\beta}^2 = D(\beta_j,\dot{\beta}(t_j) )^2,  
\end{equation}
which is the squared speed of the curve $\beta(t)$ under this metric. Therefore, we conclude \cref{eq:main_metricOd}.
\end{proof}

\subsection{Proof of \cref{lemma:ched}}

\begin{proof}
\label{proof:ched}

For a one-dimensional metric
\begin{equation}
ds^2 = g_{\beta\beta}^{\mathrm{track}}(\beta)\,d\beta^2,
\end{equation}
the affinely parametrized geodesic equation is
\begin{equation}
\ddot{\beta} +
\frac{1}{2}
\partial_\beta
\log g_{\beta\beta}^{\mathrm{track}}
\,\dot{\beta}^{\,2} = 0.
\end{equation}
For a monotone schedule, $\dot{\beta}\neq0$, and using
\begin{equation}
\ddot{\beta}
= \dot{\beta} \frac{d\dot{\beta}}{d\beta},
\end{equation}
we obtain
\begin{equation}
\frac{d\dot{\beta}}{d\beta}
+ \frac{1}{2} \partial_\beta
\log g_{\beta\beta}^{\mathrm{track}}
\,\dot{\beta} = 0.
\end{equation}
Dividing by $\dot{\beta}$ gives
\begin{equation}
\partial_\beta\log\dot{\beta}
= - \frac{1}{2}\partial_\beta \log g_{\beta\beta}^{\mathrm{track}}.
\end{equation}
Integrating with respect to $\beta$ yields
\begin{equation}
\log\dot{\beta} = - \frac{1}{2}
\log g_{\beta\beta}^{\mathrm{track}}
+ \log\alpha,
\end{equation}
where $\alpha>0$ is an integration constant. Hence
\begin{equation}
\dot{\beta} = \frac{\alpha}
{\sqrt{g_{\beta\beta}^{\mathrm{track}}}}.
\end{equation}
Substituting the tracking metric gives
\begin{equation}
\dot{\beta} = \alpha
\frac{\delta(\beta)}
{(1-\delta(\beta))V(\beta)}.
\end{equation}

Equivalently,
\begin{equation}
dt= \frac{1}{\alpha}
\frac{(1-\delta(\beta))V(\beta)}
{\delta(\beta)}
\,d\beta.
\end{equation}
Since $t\in[0,1]$, integration between the boundary values
$\beta(0)=\beta_0$ and $\beta(1)=\beta_f$ gives
\begin{equation}
1 = \frac{1}{\alpha}
\int_{\beta_0}^{\beta_f} \frac{ (1-\delta(\beta))V(\beta)}{
\delta(\beta)}
\,d\beta,
\end{equation}
and therefore
\begin{equation}
\alpha=
\int_{\beta_0}^{\beta_f}
\frac{ (1-\delta(\beta))V(\beta)
}{\delta(\beta)}
\,d\beta.
\end{equation}
Finally,
\begin{equation}
\sqrt{g_{\beta\beta}^{\mathrm{track}}}\,\dot{\beta}=\alpha,
\end{equation}
so the resulting parametrization has constant speed with respect to the
tracking metric.
\end{proof}

\section{Proofs of the Quantum discrete adiabatic theorem for the Coherent Gibbs state}

\subsection{Stationary eigenspace and adiabatic transport}\label{app:statio_eig}

Let's write the quantum walk operator as follows, $U_w(\beta)=RA(\beta)$, where
\begin{equation}
\label{eq:Abeta}
A(\beta) \coloneqq W(\beta)^{\dagger}SW(\beta).    
\end{equation}
Now we characterize the full eigenspace of $U_w(\beta)$ with eigenvalue $+1$. Suppose $U_w(\beta)\ket{\chi}=\ket{\chi}$, since $R^{\dagger}=R$, we find
\begin{equation}
A(\beta)\ket{\chi} = R \ket{\chi},\quad A(\beta)R\ket{\chi}=\ket{\chi}.    
\end{equation}
Next, consider the following  orthogonal decomposition
\begin{equation}
|u\rangle=\Pi_0|\chi\rangle
=\frac{I+R}{2}|\chi\rangle,
\qquad |v\rangle=(I-\Pi_0)|\chi\rangle
=\frac{I-R}{2}|\chi\rangle,    
\end{equation}
we, then, see that the operator $A(\beta)$ remain the vectors orthogonal to each other, one part associated with eigenvalues $+1$ and the other $-1$, i.e., 
\begin{equation}
A(\beta)|u\rangle=|u\rangle,
\qquad A(\beta)|v\rangle=-|v\rangle,  
\end{equation}
thus 
\begin{equation}
\ket{u} \in \operatorname{ran}\Pi_0\cap\ker(A-I),
\qquad
\ket{v}\in\ker\Pi_0\cap\ker(A+I),   
\end{equation}
where $\operatorname{ran}\Pi_0=\{\Pi_0|\phi\rangle:|\phi\rangle\in\mathcal H\}$. Consequently, any sum of vectors satisfying these conditions is fixed by $U_w(\beta)$, hence
\begin{equation}
\ker(U_w-I)
= \bigl[\operatorname{ran}\Pi_0\cap\ker(A-I)\bigr] \oplus
\bigl[\ker\Pi_0\cap\ker(A+I)\bigr].   
\end{equation}
Now from the definition of the discriminant operator, \cref{eq:discriminant_operator}, we have
\begin{equation}
X(\beta)\ket{u} = \Pi_0 A(\beta)\Pi_0\ket{u} =\ket{u},    
\end{equation}
thus every vector in the intersection $\operatorname{ran}\Pi_0\cap\ker(A-I)$ is an eigenvector of $X(\beta)$ with eigenvalue $+1$. On the other hand, we have $X(\beta)$ similar to the Markov kernel, \cref{eq:X_mat}, $W(\beta)$. Since $W(\beta)$ is irreducible,  its stationary eigenvalue is $+1$ is simple. Therefore, the eigenvalue-1 space of $X(\beta)$ is one-dimensional, spanned by  the quantum Gibbs state $\ket{\psi_\beta}=\ket{\pi_\beta} \otimes \ket{0}$. Therefore, we can conclude that 
\begin{equation}
\operatorname{ran}\Pi_0\cap\ker(A(\beta)-I) = \text{span}\{\ket{\psi_\beta}\},    
\end{equation}
which also can be verified from $\Pi_0\ket{\psi_\beta}=\ket{\psi_\beta}$ together with $(A(\beta)-I)\ket{\psi_\beta}=0$. Thus, the Gibbs state is the only independent $+1$ eigenvector in this first summand.

We next show that the canonical discrete adiabatic transport maps consecutive coherent Gibbs states into one another. Write
\begin{equation}
P_j=P_+(\beta_j),\qquad Q_j=I-P_j,    
\end{equation}
and define
\begin{equation}
M_j=P_{j+1}P_j+Q_{j+1}Q_j.    
\end{equation}
Under the step-size assumptions of the discrete adiabatic theorem, $M_j$ is invertible and its unitary polar factor is
\begin{equation}
V_j = M_j(M_j^{\dagger}M_j)^{-1/2},    
\end{equation}
which is an operator that also commutes with $\Pi_0$, since both projectors $P_j$ and $P_{j+1}$ commutes with $\Pi_0$.

Let 
\begin{equation}
c_j \coloneqq \langle\psi_{\beta_{j+1}}|\psi_{\beta_j}\rangle
=\sum_x\sqrt{\pi_{\beta_{j+1}}(x)\pi_{\beta_j}(x)}>0,   
\end{equation}
the projector decomposition gives
\begin{equation}
 M_j \ket{\psi_{\beta_j}} = c_j \ket{\psi_{\beta_{j+1}}}, \quad  M^{\dagger}_j \ket{\psi_{\beta_{j+1}}} = c_j \ket{\psi_{\beta_{j}}},   
\end{equation}
thus $M^{\dagger}_jM_j\ket{\psi_{\beta_j}}=c^2_j\ket{\psi_{\beta_j}}$. Consequently, the adiabatic step $V_jU_w(\beta_j)\ket{\psi_{\beta_j}}=\ket{\psi_{\beta_{j+1}}}$ maps consecutive Gibbs states exactly. Iterating gives $U_A(1)|\psi_{\beta_0}\rangle=|\psi_{\beta_f}\rangle$. Therefore, the operator-norm adiabatic estimate implies the state-preparation bound
\begin{equation}
\bigl\|U(1)|\psi_{\beta_0}\rangle-|\psi_{\beta_f}\rangle\bigr\|
\le\|U(1)-U_A(1)\|.    
\end{equation}
This establishes the connection between adiabatic tracking of the full $+1$ eigenspace and preparation of the coherent Gibbs state.

\subsection{Proof of \cref{lem:qw_finite_differences_a}}\label{app:proof_lemqw_finiy}

\begin{proof}
Using $ U_w(s)=R W(s)^\dagger S W(s)$, the first finite difference, $D^{(1)}U_w(s)=U_w(s+h)-U_w(s)$, where here $h=1/T_d$, can be written as
\begin{align}
D^{(1)}U_w(s)
&= R W^{\dagger}(s+h) S W(s+h) -
R W^{\dagger}(s) S W(s) \nonumber\\
&= R W^{\dagger}(s+h) S W(s+h) - R W^{\dagger}(s) S W(s+h)\nonumber\\
&\quad + R W^{\dagger}(s) S W(s+h) - R W^{\dagger}(s) S W(s) \nonumber\\
&= R\left[ W^{\dagger}(s+h)-W^{\dagger}(s) \right]
S W(s+h)\nonumber\\
&\quad +
R W^{\dagger}(s) S \left[ W(s+h)-W(s) \right] \nonumber\\
&= R\left[D^{(1)}W^{\dagger}(s) \right] S W(s+h) + R W^{\dagger}(s) S D^{(1)}W(s),
\label{eq:D1Uw_expansion}
\end{align}
Since $R$, $S$, $W(s)$, and $W(s+h)$ are unitary,  we obtain
\begin{equation}
\left\| D^{(1)}U_w(s) \right\| \leq 2 \left\| D^{(1)}W(s) \right\|.
\end{equation}
which gives $c_1(s)\leq 2a_1(s)$.

For the second finite difference, where $D^{(2)}U_{w}(s)= D^{(1)}U_w(s+h) - D^{(1)}(s)$, define
\begin{equation}
A(s) = W^\dagger(s)S W(s),
\end{equation}
so that $U_w(s)=R A(s)$. Since $R$ is independent of $s$,
\begin{equation}
D^{(2)}U_w(s) = R D^{(2)}A(s).
\end{equation}
We therefore first calculate $D^{(2)}A(s)$. By definition,
\begin{align} D^{(2)}A(s)  &= D^{(1)}A(s+h) -D^{(1)}A(s)
\nonumber\\ &= A(s+2h)-2A(s+h)+A(s).
\end{align}

From the first-order calculation,
\begin{equation}
    D^{(1)}A(s)
    =
    D^{(1)}W^\dagger(s)
    S W(s+h)
    +
    W^\dagger(s)
    S D^{(1)}W(s).
    \label{eq:D1A}
\end{equation}
Evaluating the same expression at $s+h$ gives
\begin{equation}
D^{(1)}A(s+h) = D^{(1)}W^\dagger(s+h) S W(s+2h) + W^\dagger(s+h)
    S D^{(1)}W(s+h).
    \label{eq:D1A_shift}
\end{equation}
Subtracting \cref{eq:D1A} from \cref{eq:D1A_shift},
\begin{align}
D^{(2)}A(s) &=
D^{(1)}W^\dagger(s+h)
S W(s+2h) +
W^\dagger(s+h)
S D^{(1)}W(s+h) \nonumber\\
&\quad-
D^{(1)}W^\dagger(s)
S W(s+h) -
W^\dagger(s)
S D^{(1)}W(s).
\label{eq:D2A_step1}
\end{align}

We now regroup the first and third terms by adding and subtracting $D^{(1)}W^\dagger(s)S W(s+2h)$, which gives
\begin{align}
D^{(1)}W^\dagger(s+h)S W(s+2h) - D^{(1)}W^\dagger(s)S W(s+h)&=
\left[D^{(1)}W^\dagger(s+h) - D^{(1)}W^\dagger(s) \right]\nonumber\\
&\quad\times S W(s+2h) + D^{(1)}W^\dagger(s) S \left[W(s+2h)-W(s+h)\right].
\end{align}
Using
\begin{equation}
D^{(2)}W^\dagger(s) = D^{(1)}W^\dagger(s+h) - D^{(1)}W^\dagger(s),
\end{equation}
and
\begin{equation}
W(s+2h)-W(s+h) = D^{(1)}W(s+h),
\end{equation}
this becomes
\begin{align}
& D^{(1)}W^\dagger(s+h)S W(s+2h)
-
D^{(1)}W^\dagger(s)S W(s+h)
\nonumber\\
={}&
D^{(2)}W^\dagger(s)
S W(s+2h)
+
D^{(1)}W^\dagger(s)
S D^{(1)}W(s+h).
\label{eq:D2A_pair1}
\end{align}

Similarly, regrouping the second and fourth terms in
\cref{eq:D2A_step1}, we add and subtract
\begin{equation}
    W^\dagger(s)S D^{(1)}W(s+h),
\end{equation}
to obtain
\begin{align}
W^\dagger(s+h)S D^{(1)}W(s+h) -
W^\dagger(s)S D^{(1)}W(s)
&=\left[
W^\dagger(s+h)-W^\dagger(s)
\right]
S D^{(1)}W(s+h)
\nonumber\\
&\quad+
W^\dagger(s)S
\left[
D^{(1)}W(s+h)-D^{(1)}W(s)
\right].
\end{align}
Using
\begin{equation}
    W^\dagger(s+h)-W^\dagger(s)
    =
    D^{(1)}W^\dagger(s),
\end{equation}
and
\begin{equation}
    D^{(1)}W(s+h)-D^{(1)}W(s)
    =
    D^{(2)}W(s),
\end{equation}
we obtain
\begin{align}
&
W^\dagger(s+h)S D^{(1)}W(s+h)
-
W^\dagger(s)S D^{(1)}W(s)
\nonumber\\
={}&
D^{(1)}W^\dagger(s)
S D^{(1)}W(s+h)
+
W^\dagger(s)S D^{(2)}W(s).
\label{eq:D2A_pair2}
\end{align}

Combining \cref{eq:D2A_pair1,eq:D2A_pair2} gives
\begin{align}
    D^{(2)}
    \left(
    W^\dagger S W
    \right)(s)
    ={}&
    D^{(2)}W^\dagger(s)
    S W(s+2h)
    \nonumber\\
    &+
    2D^{(1)}W^\dagger(s)
    S D^{(1)}W(s+h)
    \nonumber\\
    &+
    W^\dagger(s)
    S D^{(2)}W(s).
    \label{eq:D2_product_rule}
\end{align}

Since $U_w(s)=RA(s)$ and $R$ is unitary,
\begin{equation}
    \left\|
    D^{(2)}U_w(s)
    \right\|
    =
    \left\|
    D^{(2)}A(s)
    \right\|.
\end{equation}
Using the triangle inequality and unitary invariance of the operator norm,
\begin{align}
\left\|
D^{(2)}U_w(s)
\right\|
&\leq
\left\|
D^{(2)}W^\dagger(s)
\right\|
\nonumber\\
&\quad+
2
\left\|
D^{(1)}W^\dagger(s)
\right\|
\left\|
D^{(1)}W(s+h)
\right\| +
\left\|
D^{(2)}W(s)
\right\|
\end{align}
so we finally obtain
\begin{equation}
\left\|D^{(2)}U_w(s)\right\| \leq 2 \left\| D^{(2)}W(s) \right\| + 2 \left\| D^{(1)}W(s)\right\| \left\| D^{(1)}W(s+h) \right\|.
\end{equation}
\end{proof}

\subsection{Proof of \cref{lem:qw_finite_differences_L}}\label{app:proof_finit_difc}

\begin{proof}
Let $h=1/T_d$, since $W(s)=W(\beta(s))$ is assumed to be sufficiently smooth, we may expand
$W(s+h)$ around $s$ as
\begin{equation}
W(s+h) = W(s) + h\frac{dW}{ds}(s) + \frac{h^2}{2} \frac{d^2W}{ds^2}(s) + O(h^3).
    \label{eq:W_Taylor_1}
\end{equation}
Therefore, the first finite difference is
\begin{align}
D^{(1)}W(s) &=  W(s+h)-W(s) \nonumber\\
&= h\frac{dW}{ds}(s) + O(h^2).
\label{eq:D1W_Taylor}
\end{align}
Using the chain rule,
\begin{equation}
\frac{dW}{ds} = \dot\beta(s) \partial_\beta W(\beta(s)),
\end{equation}
so that
\begin{equation}
D^{(1)}W(s) = h\,\dot\beta(s) \partial_\beta W(\beta(s)) +
O(h^2).
\end{equation}
Taking the operator norm gives
\begin{align}
\left\| D^{(1)}W(s) \right\| &\leq h|\dot\beta(s)| \left\|
\partial_\beta W(\beta(s)) \right\| + O(h^2)
\nonumber\\
&= \frac{1}{T_d} |\dot\beta(s)| L_1(\beta(s)) + O(T_d^{-2}).
\end{align}
Comparing with
\begin{equation}
\left\| D^{(1)}W(s) \right\| \leq \frac{a_1(s)}{T_d},
\end{equation}
we may therefore choose
\begin{equation}
a_1(s) = |\dot\beta(s)| L_1(\beta(s)) + O(T_d^{-1}).
\label{eq:a1_smooth_proof}
\end{equation}

We next consider the second finite difference,
\begin{equation}
D^{(2)}W(s) =  W(s+2h)-2W(s+h)+W(s).
\end{equation}
Expanding $W(s+2h)$ around $s$ gives
\begin{equation}
 W(s+2h) = W(s) + 2h\frac{dW}{ds}(s) + 2h^2 \frac{d^2W}{ds^2}(s) + O(h^3),
\label{eq:W_Taylor_2}
\end{equation}
whereas \cref{eq:W_Taylor_1} gives
\begin{equation}
2W(s+h) = 2W(s) + 2h\frac{dW}{ds}(s) + h^2\frac{d^2W}{ds^2}(s)
+ O(h^3).
\end{equation}
Substituting these expressions into the definition of $D^{(2)}W(s)$, the
zeroth- and first-order terms cancel, yielding
\begin{align}
D^{(2)}W(s) &= h^2 \frac{d^2W}{ds^2}(s) + O(h^3)\nonumber\\
&= \frac{1}{T_d^2} \frac{d^2W}{ds^2}(s) + O(T_d^{-3}).
    \label{eq:D2W_Taylor}
\end{align}

To evaluate the second derivative, differentiate
\begin{equation}
\frac{dW}{ds} = \dot\beta(s)
\partial_\beta W(\beta(s)).
\end{equation}
Using the product and chain rules,
\begin{align}
\frac{d^2W}{ds^2} &= \frac{d}{ds} \left[ \dot\beta(s) \partial_\beta W(\beta(s)) \right] \nonumber\\
&= \ddot\beta(s) \partial_\beta W(\beta(s)) + \dot\beta(s) \frac{d}{ds}\partial_\beta W(\beta(s)) \nonumber\\
&= \ddot\beta(s) \partial_\beta W(\beta(s)) + \dot\beta(s)^2 \partial_\beta^2W(\beta(s)).
    \label{eq:d2W_chain}
\end{align}
Therefore,
\begin{align}
\left\| \frac{d^2W}{ds^2} \right\| &\leq |\ddot\beta(s)|
\left\| \partial_\beta W(\beta(s)) \right\| + |\dot\beta(s)|^2 \left\| \partial_\beta^2W(\beta(s))  \right\| \nonumber\\
 &= |\ddot\beta(s)|
    L_1(\beta(s))
    +
    |\dot\beta(s)|^2
    L_2(\beta(s)).
\end{align}
Combining this with \cref{eq:D2W_Taylor} gives
\begin{align}
    \left\|
    D^{(2)}W(s)
    \right\|
    \leq{}&
    \frac{1}{T_d^2}
    \left[
    |\ddot\beta(s)|
    L_1(\beta(s))
    +
    |\dot\beta(s)|^2
    L_2(\beta(s))
    \right] +
    O(T_d^{-3}).
\end{align}
Hence,
\begin{equation}
    a_2(s)
    \leq
    |\ddot\beta(s)|
    L_1(\beta(s))
    +
    |\dot\beta(s)|^2
    L_2(\beta(s))
    +
    O(T_d^{-1}).
    \label{eq:a2_smooth_proof}
\end{equation}

We now substitute \cref{eq:a1_smooth_proof} into
Lemma~\ref{lem:qw_finite_differences_a}. From
\begin{equation}
    c_1(s)\leq 2a_1(s),
\end{equation}
we immediately obtain
\begin{equation}
    c_1(s)
    \leq
    2|\dot\beta(s)|
    L_1(\beta(s))
    +
    O(T_d^{-1}),
\end{equation}
which proves \cref{eq:c1_L}. For $c_2$, Lemma~\ref{lem:qw_finite_differences_a} gives
\begin{equation}
c_2(s)\leq 2a_2(s) + 2a_1(s)a_1(s+h).
    \label{eq:c2_substitution}
\end{equation}
Since the schedule and $L_1(\beta)$ are smooth,
\begin{equation}
a_1(s+h) = a_1(s) + O(h),
\end{equation}
and therefore
\begin{align}
a_1(s)a_1(s+h) &= a_1(s)^2 + O(h) \nonumber\\
&= |\dot\beta(s)|^2 L_1(\beta(s))^2 + O(T_d^{-1}).
    \label{eq:a1_product}
\end{align}
Using \cref{eq:a2_smooth_proof,eq:a1_product} in
\cref{eq:c2_substitution}, we find
\begin{align}
c_2(s) &\leq 2|\ddot\beta(s)| L_1(\beta(s)) + 2|\dot\beta(s)|^2
L_2(\beta(s)) + 2|\dot\beta(s)|^2 L_1(\beta(s))^2 + O(T_d^{-1})
\nonumber\\
&= 2|\ddot\beta(s)| L_1(\beta(s)) + 2|\dot\beta(s)|^2 \left[ L_2(\beta(s))
+ L_1(\beta(s))^2 \right] + O(T_d^{-1}),
\end{align}
which proves \cref{eq:c2_L}.
\end{proof}

\subsection{Proof of \cref{thm:quantum_gap_scaling}}
\label{app:proof_gap}

\begin{proof}
Write $s_j=j/T_d$ and
$\Delta_{w,\min}=\min_{s\in[0,1]}\Delta_w(s)$.
By the uniform discrete hypotheses of
\cref{thm:quantum_gap_scaling},
\begin{equation}
\hat c_1(s_j)
\le A_1\mathcal N_p\Delta_w(s_j)^p,
\end{equation}
and
\begin{equation}
\hat c_2(s_j)
\le A_2\mathcal N_p^2
\left[
\Delta_w(s_j)^{2p-1}
+\Delta_w(s_j)^{2p}
\right],
\end{equation}
at all relevant indices. The assumed gap comparison gives
\begin{equation}
\check\Delta_2(s_j)\ge b\Delta_w(s_j),
\end{equation}
where $A_1,A_2,b>0$ are independent of $T_d$ and
$\Delta_{w,\min}$.

We first show that the runtime prerequisite in
\cref{eq:run_cond} is compatible with the claimed
inverse-gap scaling. Since
$0<\Delta_w(s)\le\pi$, we have
\begin{equation}
\max_j
\frac{4\hat c_1(s_j)}{\check\Delta_2(s_j)}
\le
\frac{4A_1}{b}\mathcal N_p
\max_j\Delta_w(s_j)^{p-1}
\le
\frac{4A_1}{b}\mathcal N_p\pi^{p-1}.
\end{equation}
By the assumed normalization bound, there exists a
constant $C_N>0$, independent of the minimum gap, such that
\begin{equation}
\mathcal N_p\le C_N\Delta_{w,\min}^{1-p}.
\end{equation}
Consequently, \cref{eq:run_cond} is satisfied whenever
\begin{equation}
T_d\ge K\Delta_{w,\min}^{1-p},
\qquad
K:=\frac{4A_1C_N\pi^{p-1}}{b}.
\end{equation}
In the small-gap regime $0<\Delta_{w,\min}\le1$,
the assumption $1<p<2$ implies
\begin{equation}
\Delta_{w,\min}^{1-p}
\le\Delta_{w,\min}^{-1}.
\end{equation}
Thus, choosing $T_d\ge K/\Delta_{w,\min}$ also satisfies
the runtime prerequisite, without increasing the claimed
inverse-gap scaling.

We now bound the contributions in
\cref{eq:discrete_adiabatic_bound}, for runtimes satisfying
\cref{eq:run_cond}. For each boundary term with a squared
gap denominator,
\begin{align}
\frac{\hat c_1(s_j)}
{T_d\check\Delta_2(s_j)^2}
&\le
\frac{A_1\mathcal N_p}
{b^2T_d\Delta_w(s_j)^{2-p}}
\nonumber\\
&\le
\frac{A_1C_N}{b^2T_d}
\frac{\Delta_{w,\min}^{1-p}}
{\Delta_{w,\min}^{2-p}}\nonumber\\
&=
\frac{A_1C_N}
{b^2T_d\Delta_{w,\min}}.
\end{align}
For the boundary term with a single gap denominator,
\begin{align}
\frac{\hat c_1(s_j)}
{T_d\check\Delta_2(s_j)}
&\le
\frac{A_1\mathcal N_p}{bT_d}
\Delta_w(s_j)^{p-1}
\nonumber\\
&\le \frac{A_1C_N\pi^{p-1}}{bT_d}
\Delta_{w,\min}^{1-p}\nonumber\\
&= O\left(\frac{1}{T_d\Delta_{w,\min}}\right).
\end{align}

For the accumulated first-order contribution with a
cubed gap denominator, the coefficient bound gives
\begin{align}
\frac{1}{T_d^2}
\sum_{j=1}^{T_d-1}
\frac{\widehat c_1(s_j)^2}
{\check\Delta_2(s_j)^3}
&\le
\frac{A_1^2\mathcal N_p^2}{b^3T_d^2}
\sum_{j=1}^{T_d-1}\Delta_w(s_j)^{2p-3}
\nonumber\\
&\le
\frac{A_1^2C_q\mathcal N_p^2}{b^3T_d}
\int_0^1\Delta_w(s)^{2p-3}\,ds,
\end{align}
where the second inequality follows from the assumed
uniform quadrature estimate. The sum can be extended to
$j=0,\ldots,T_d-1$ because all its terms are nonnegative.

Using the schedule relation
\begin{equation}
ds=
\frac{L_1(\beta)}
{\mathcal N_p\Delta_w(\beta)^p}\,d\beta,
\end{equation}
we obtain
\begin{equation}
\mathcal N_p^2
\int_0^1\Delta_w(s)^{2p-3}\,ds
=
\mathcal N_p
\int_{\beta_0}^{\beta_f}
\frac{L_1(\beta)}
{\Delta_w(\beta)^{3-p}}\,d\beta.
\end{equation}
By \cref{eq:schedule_integral_condition2}, this quantity
is $O(\Delta_{w,\min}^{-1})$. Therefore,
\begin{equation}
\frac{1}{T_d^2}
\sum_{j=1}^{T_d-1}
\frac{\hat c_1(s_j)^2}
{\check\Delta_2(s_j)^3}
=
O\left(\frac{1}{T_d\Delta_{w,\min}}\right).
\end{equation}

For the accumulated first-order contribution with a
squared gap denominator, we similarly have
\begin{align}
\frac{1}{T_d^2}
\sum_{j=0}^{T_d-1}
\frac{\hat c_1(s_j)^2}
{\check\Delta_2(s_j)^2}
&\le
\frac{A_1^2\mathcal N_p^2}{b^2T_d^2}
\sum_{j=0}^{T_d-1}\Delta_w(s_j)^{2p-2}
\nonumber\\
&\le
\frac{\pi A_1^2\mathcal N_p^2}{b^2T_d^2}
\sum_{j=0}^{T_d-1}\Delta_w(s_j)^{2p-3},
\end{align}
where we used $\Delta_w(s_j)\le\pi$.
Applying the same quadrature estimate and schedule
integral bound gives
\begin{equation}
\frac{1}{T_d^2}
\sum_{j=0}^{T_d-1}
\frac{\hat c_1(s_j)^2}
{\check\Delta_2(s_j)^2}
=
O\left(\frac{1}{T_d\Delta_{w,\min}}\right).
\end{equation}

Finally, the uniform second-difference bound implies
\begin{align}
\frac{1}{T_d^2}
\sum_{j=1}^{T_d-1}
\frac{\hat c_2(s_j)}
{\check\Delta_2(s_j)^2}
&\le
\frac{A_2\mathcal N_p^2}{b^2T_d^2}
\sum_{j=1}^{T_d-1}
\left[
\Delta_w(s_j)^{2p-3}
+\Delta_w(s_j)^{2p-2}
\right]
\nonumber\\
&\le
\frac{A_2(1+\pi)\mathcal N_p^2}{b^2T_d^2}
\sum_{j=0}^{T_d-1}\Delta_w(s_j)^{2p-3}.
\end{align}
The same quadrature estimate and change of variables
therefore give
\begin{equation}
\frac{1}{T_d^2}
\sum_{j=1}^{T_d-1}
\frac{\hat c_2(s_j)}
{\check\Delta_2(s_j)^2}
=
O\left(\frac{1}{T_d\Delta_{w,\min}}\right).
\end{equation}

All contributions in \cref{eq:discrete_adiabatic_bound}
are thus bounded by
$O((T_d\Delta_{w,\min})^{-1})$, with uniform constants
under the stated hypotheses. Consequently,
\begin{equation}
\|U(1)-U_A(1)\|
=
O\left(\frac{1}{T_d\Delta_{w,\min}}\right).
\end{equation}
For fixed target error, a sufficiently large constant
multiple of $\Delta_{w,\min}^{-1}$ therefore suffices
both to satisfy the runtime prerequisite and to achieve
the target accuracy. Using
$\Delta_{w,\min}=\Theta(\sqrt{\delta_{\min}})$ gives the
claimed sufficient scaling
$T_d=O(\delta_{\min}^{-1/2})$.
\end{proof}

\section{Spectral identification and operator variation of the LHPST walk}\label{app:Analy_Diff}

We describe the LHPST realization used in the numerical experiments and derive the operator-variation formula entering the discrete quantum schedule.

\subsection{Walk construction and discriminant}

Let $x\in\{-1,+1\}^n$ denote a spin configuration and $x^{(k)}$ the configuration obtained by flipping spin $k$. The Metropolis acceptance probability is
\begin{equation}
p_{xk}(\beta) =\min\!\left\{1,e^{-\beta[E(x^{(k)})-E(x)]}\right\}.    
\end{equation}
The corresponding transition kernel with uniform single-spin proposals is
\begin{equation}
P_\beta(x^{(k)}|x)=\frac{p_{xk}(\beta)}n,
\qquad
P_\beta(x|x)=1-\frac1n\sum_{k=1}^n p_{xk}(\beta).    
\end{equation}
Let $\Pi_0$ project onto the initial auxiliary state. The temperature-independent unitary $V$ prepares
\begin{equation}
V|x,0\rangle
=\frac1{\sqrt n}\sum_{k=1}^n|x,k,0\rangle,    \end{equation}
where the last register is the acceptance qubit. Conditioned on $x,k$, the acceptance unitary is $B_{xk}(\beta)\ket{1} = \sqrt{p}\ket{0} - \sqrt{1-p}\ket{1}$, where $p=p_{xk}(\beta)$. Thus, in the ordered basis $\{\ket{0},\ket{1}\}$, it gives 
\begin{equation}
\label{eq: Block_B}
B_{xk}(\beta)=
\begin{pmatrix}
\sqrt{1-p_{xk}(\beta)}&\sqrt{p_{xk}(\beta)}\\
\sqrt{p_{xk}(\beta)}&-\sqrt{1-p_{xk}(\beta)}
\end{pmatrix}.  
\end{equation}
These blocks are Hermitian and unitary. Any unused auxiliary blocks are chosen to be temperature independent.
The controlled move operator satisfies
\begin{equation}
F|x,k,0\rangle=|x,k,0\rangle,
\qquad F|x,k,1\rangle=|x^{(k)},k,1\rangle.    
\end{equation}
Thus $F^\dagger=F$ and $F^2=I$. Defining
\begin{equation}
W_{\mathrm{LHPST}}(\beta)=B(\beta)V,
\qquad R=2\Pi_0-I,    
\end{equation}
the walk is
\begin{equation}
\widetilde U_w(\beta)
=R\,W_{\mathrm{LHPST}}(\beta)^\dagger
F W_{\mathrm{LHPST}}(\beta)
=R\,V^\dagger B(\beta)^\dagger F B(\beta)V.    
\end{equation}
Set
\begin{equation}
A(\beta)=V^\dagger B(\beta)^\dagger F B(\beta)V,
\qquad
X_\beta=\Pi_0A(\beta)\Pi_0.    
\end{equation}
The operator $A_\beta$ is a Hermitian unitary. On the preparation subspace, direct evaluation gives
\begin{equation}
\langle y,0|X_\beta|x,0\rangle
=
\begin{cases}
\displaystyle 1-\frac1n\sum_k p_{xk}(\beta),
&y=x,\\[6pt]
\displaystyle \frac1n
\sqrt{p_{xk}(\beta)p_{yk}(\beta)},
&y=x^{(k)},\\[6pt]
0,&\text{otherwise}.
\end{cases}    
\end{equation}
Consequently,
\begin{equation}
\langle y,0|X_\beta|x,0\rangle
=\sqrt{P_\beta(y|x)P_\beta(x|y)}.   
\end{equation}
By detailed balance, this is the discriminant of $P_\beta$. With the column-stochastic convention, $X_\beta\big|_{\operatorname{ran}\Pi_0}
\cong D_{\pi_\beta}^{-1/2}P_\beta D_{\pi_\beta}^{1/2}.$

\subsection{Gibbs eigenspace and eigenphase gap}
For finite $\beta$, the single-spin Metropolis kernel is irreducible, and its stationary eigenvalue $1$ is simple. Its discriminant therefore has the unique normalized eigenvalue-$1$ vector $|\psi_\beta\rangle
=\sum_x\sqrt{\pi_\beta(x)}\,|x,0\rangle.$ Since $A(\beta)$ is a Hermitian unitary, $\|(A(\beta) - I)\ket{\psi_\beta}\|=0$, then 
\begin{equation}
A_\beta|\psi_\beta\rangle=|\psi_\beta\rangle, \qquad \widetilde U_w(\beta)|\psi_\beta\rangle=|\psi_\beta\rangle.    
\end{equation}

The full eigenvalue- $+1$ eigenspace may contain additional states. Its projector has the decomposition
\begin{equation}
P_+(\beta) =|\psi_\beta\rangle\langle\psi_\beta|
+P_{\mathrm{extra}}(\beta),
\qquad P_{\mathrm{extra}}(\beta)\Pi_0=0.   \end{equation}
The projector and transport argument in \cref{app:statio_eig} therefore applies to this realization: the canonical adiabatic transport preserves the fixed preparation sector and maps its Gibbs line between successive temperatures.

Now, if we write the eigenvalues of $P_\beta$ as $1=\lambda_0(\beta)>\lambda_1(\beta)\ge\cdots.$ The nontrivial reflection-product blocks have eigenphases $\pm\arccos \lambda_k(\beta)$. Thus the angular gap separating the entire eigenvalue-$+1$ eigenspace from the remaining spectrum is
\begin{equation}
\Delta_w(\beta) =\arccos\lambda_1(\beta)
=\arccos\!\left(1-\delta_+(\beta)\right),
\qquad \delta_+(\beta)=1-\lambda_1(\beta).   
\end{equation}
Here $\delta_+$ is the right spectral gap of the non-lazy Metropolis kernel. If the classical gap is instead computed from
\begin{equation}
P_\beta^L=\frac12(I+P_\beta), 
\end{equation}
then 
\begin{equation}
\delta_L(\beta)=\frac{\delta_+(\beta)}2,
\qquad
\Delta_w(\beta)=\arccos\!\left(1-2\delta_L(\beta)\right).   
\end{equation}

The implemented evolution uses the prepared move basis and the opposite sign for the reflection. These conventions correspond to a temperature-independent change of basis and an overall phase per step, respectively, and leave the relevant angular gap and output probabilities unchanged.

\subsection{Exact finite-difference norm}
Because $V$ is temperature independent and unitary $\|W_{\mathrm{LHPST}}(\beta')-W_{\mathrm{LHPST}}(\beta)\|
=\|B(\beta')-B(\beta)\|.$ For $\beta,\beta'\ge0$, downhill and zero-energy moves have acceptance probability $1$, so their blocks are constant.

For an uphill energy increment $d>0$, write
\begin{equation}
a_d(\beta)=\sqrt{1-e^{-\beta d}},
\qquad b_d(\beta)=e^{-\beta d/2}.    
\end{equation}
Given that the acceptance unitary has the block-form given by \cref{eq: Block_B}, the block difference is 
\begin{equation}
B_d(\beta')-B_d(\beta)
= \begin{pmatrix}
\Delta a_d&\Delta b_d\\
\Delta b_d&-\Delta a_d
\end{pmatrix},    
\end{equation}
where $\Delta a_d=a_d(\beta')-a_d(\beta)$, and similarly for $\Delta b_d$. Its square equals
\begin{equation}
\bigl[B_d(\beta')-B_d(\beta)\bigr]^2
=\left[(\Delta a_d)^2+(\Delta b_d)^2\right]I.    
\end{equation}
Therefore,
\begin{equation}
\|B_d(\beta')-B_d(\beta)\|
=\sqrt{(\Delta a_d)^2+(\Delta b_d)^2}.    
\end{equation}
Taking the maximum over the uphill increments occurring in the instance gives the full operator norm.

For the classical Hamiltonian used here, $E(x) = \frac2{\sqrt n}\sum_{i<j}J_{ij}x_ix_j$ with $J_{ij}\in\{-1,+1\}$, a spin flip produces 
\begin{equation}
E(x^{(k)})-E(x)
=-\frac4{\sqrt n}x_k h_k(x),
\qquad h_k(x)=\sum_{j\ne k}J_{kj}x_j.    
\end{equation}
Thus the positive increments are $\Delta E_m =4m/\sqrt{n}$, where $m$ ranges over the positive local-field magnitudes. For this fully connected binary model,
\begin{equation}
\mathcal{M}_n=\left\{|n-1-2r|:r=0,\ldots,n-1\right\}\setminus\{0\}.
\end{equation}
Consequently, 
\begin{align}
D_{n,i} &=\|W_{\mathrm{LHPST}}(\beta_{i+1})
-W_{\mathrm{LHPST}}(\beta_i)\|\nonumber\\
&=\max_{m\in\mathcal M_n}
\left[ \left(
\sqrt{1-e^{-\beta_{i+1}\Delta E_m}}
-\sqrt{1-e^{-\beta_i\Delta E_m}}
\right)^2
+ \left( e^{-\beta_{i+1}\Delta E_m/2}
-e^{-\beta_i\Delta E_m/2}
\right)^2 \right]^{1/2}.
\end{align}
This establishes \cref{eq:SK_discrete_L1}. The analytical maximum contains $O(n)$ distinct increments. The supplied implementation instead obtains the same maximum by enumerating configurations and proposed flips.

\subsection{Derivative norms and the initial endpoint}
For $\beta>0$, let's define $\theta_d(\beta)=\arccos (e^{-\beta d/2})$, then the block form of the acceptance unitary is
\begin{equation}
B_d(\beta)=
\begin{pmatrix}
\sin\theta_d&\cos\theta_d\\
\cos\theta_d&-\sin\theta_d
\end{pmatrix}.  
\end{equation}
Differentiating the unitary above in respect to $\beta$, then, gives $B'_d=\theta'_dC_d$, where 
\begin{equation}
C_d(\beta)=
\begin{pmatrix}
\cos\theta_d&-\sin\theta_d\\
-\sin\theta_d&-\cos\theta_d
\end{pmatrix}.  
\end{equation}
Thus, $\|B_d'(\beta)\|=|\theta_d'(\beta)|$, and also $\|B_d''(\beta)\| =\sqrt{\theta_d''(\beta)^2+\theta_d'(\beta)^4}$, with 
\begin{equation}
\theta_d'(\beta)
=\frac{d e^{-\beta d/2}}{2\sqrt{1-e^{-\beta d}}},
\qquad \theta_d''(\beta) =-\frac{d^2e^{-\beta d/2}}
{4(1-e^{-\beta d})^{3/2}}.    
\end{equation}
The full preparation-unitary norms are therefore
\begin{align}
\label{eq:L1_l2}
L_1(\beta) &=\|\partial_\beta W_{\mathrm{LHPST}}(\beta)\|
=\max_d|\theta_d'(\beta)|\nonumber\\
L_2(\beta)&= \|\partial_\beta^2W_{\mathrm{LHPST}}(\beta)\|
=\max_d\sqrt{\theta_d''(\beta)^2+\theta_d'(\beta)^4},
\end{align}
where the maxima run over the positive energy increments.

Thus, at the initial endpoint \cref{eq:L1_l2}, $\sqrt{1-e^{-\beta d}} = \sqrt{d\beta}+O(\beta^{3/2})$, so $B(\beta)$ is continuous but not differentiable in $\beta$ at zero. Nevertheless, its finite differences remain well defined, and these are used directly to construct the numerical schedule.

We, then, see that the numerical schedule starting at $\beta=0$ is not automatically covered by \cref{thm:quantum_gap_scaling}. Such coverage requires verifying the theorem’s uniform discrete regularity and gap assumptions along the sampled sequence. We do not establish those endpoint estimates here; the numerical results for this schedule are therefore presented as empirical evidence rather than as a direct application of the asymptotic runtime bound.

\section{Construction of the schedules used in the numerical experiments}\label{app:sched_construc}

In this appendix we give the explicit construction of the schedules used
in the numerical experiments.  It is useful to distinguish between \textit{fixed schedules}, which are identical for all disorder realizations,
and \textit{instance-dependent schedules}, which are constructed from the dynamical geometry of each realization.

\subsection{Fixed schedules}

We now list the fixed schedules explored in this report.

\paragraph{Linear schedule.}
As a reference, we first consider a uniform interpolation in inverse
temperature,
\begin{equation}
\label{eq:schedule_linear_two_state}
\beta_{\mathrm{lin}}(t) = \beta_0 + (\beta_f-\beta_0)t.
\end{equation}
This schedule does not make use of either the statistical properties of the Gibbs distribution or the relaxation properties of the Markov chain.

\paragraph{Power-law schedules.}
We also consider the family
\begin{equation}
\label{eq:schedule_power_two_state}
\beta_p(t) = \beta_0 + (\beta_f-\beta_0)t^p,
\end{equation}
with $p=1/2$, 2 and 3. The choice $p<1$ produces a front-loaded schedule, for which the inverse temperature changes faster at early times, whereas $p>1$ produces a back-loaded schedule.

\paragraph{Exponential schedules.}
A second family of heuristic schedules is given by
\begin{equation}
\label{eq:schedule_exp_two_state}
\beta_{\exp}(t) = \beta_0 + (\beta_f-\beta_0) \frac{e^{at}-1}{e^a-1},
\end{equation}
where we consider different values for $a$ according to the problem studied. For $a<0$ the schedule is front-loaded and slows down toward the final inverse temperature, while $a>0$ gives the opposite behavior.

\subsection{Construction of the schedules for the two-state model}

Here we derive the inverse temperature schedules based on different proposals, including the schedule results from the main theorem \cref{thm:tracking_changing_norm} for the solvable two-state model, \cref{sec:two-state_m}.
\paragraph{Variance-only schedule.}
For the Gibbs family,
\begin{equation}
    V(\beta)
    =
    \sqrt{\operatorname{Var}_{\pi_\beta}(E)}
    =
    \frac{1}{2\cosh(\beta/2)}.
\end{equation}
The corresponding statistical metric is $g_{\beta\beta}^{\mathrm{stat}}(\beta) = V(\beta)^2$. A constant-speed parametrization with respect to this metric satisfies $V(\beta(t))\,\dot\beta(t) = \alpha_V$ and therefore
\begin{equation}
\label{eq:variance_schedule_velocity}
    \dot\beta_{\mathrm{var}}(t)
    =
    \frac{\alpha_V}{V(\beta(t))}.
\end{equation}
Using
\begin{equation}
\int \frac{d\beta}{2\cosh(\beta/2)}  =
    \arctan\!\left[\sinh\!\left(\frac{\beta}{2}\right)\right],
\end{equation}
we obtain the explicit schedule
\begin{equation}
\label{eq:variance_schedule_explicit}
\beta_{\mathrm{var}}(t) = 2\,\operatorname{arsinh}
\left[ \tan\left((1-t)F_V(\beta_0)+ tF_V(\beta_f) \right)\right],
\end{equation}
where
\begin{equation}
 F_V(\beta) = \arctan\!\left[ \sinh\!\left(\frac{\beta}{2}\right)\right].
\end{equation}

\paragraph{Gap-only schedule.}
The discrete tracking theorem exhibits a leading inverse-linear dependence
on the spectral gap. To isolate this contribution, we consider a schedule
satisfying
\begin{equation}
\frac{\dot\beta(t)}{\Delta(\beta(t))} = \alpha_\Delta.
\end{equation}
Since in the present model $\Delta(\beta)=e^{-\beta}$ the corresponding schedule obeys
\begin{equation}
\label{eq:gap_schedule_velocity}
\dot\beta_{\mathrm{gap}}(t) = \alpha_\delta e^{-\beta(t)}.
\end{equation}
Integration gives
\begin{equation}
\label{eq:gap_schedule_explicit}
    \beta_{\mathrm{gap}}(t)
    =
    \log\left[
    (1-t)e^{\beta_0}
    +
    te^{\beta_f}
    \right].
    \end{equation}
This schedule therefore directly probes the inverse-linear gap dependence
predicted by the tracking theorem.

\paragraph{Friction schedule.}
We also compare with the constant-excess-power schedule obtained from the
thermodynamic friction tensor. For a one-dimensional control parameter,
the friction coefficient takes the form
\begin{equation}
\label{eq:friction_tensor_two_state}
    \zeta_{\beta\beta}(\beta)
    \sim
    \operatorname{Var}_{\pi_\beta}(E)\,
    \tau_E(\beta),
\end{equation}
where $\tau_E(\beta)$ is the equilibrium correlation time of the energy.
For the present two-state transition operator, the centered energy
observable lies entirely in the unique nonstationary mode, whose eigenvalue
is $1-\delta(\beta)$. Using the one-sided discrete correlation-time
convention,
\begin{equation}
\tau_E(\beta) = \sum_{k=0}^{\infty}
\frac{C_E(k)}{C_E(0)} = \frac{1}{\delta(\beta)}.
\end{equation}
Hence
\begin{equation}
 \zeta_{\beta\beta}(\beta) \sim \frac{V(\beta)^2}{\delta(\beta)}.
\end{equation}
The constant-excess-power condition then gives
\begin{equation}
\label{eq:friction_schedule_velocity}
    \dot\beta_{\mathrm{fr}}(t)
    \propto
    \frac{\sqrt{\delta(\beta(t))}}
    {V(\beta(t))}.
\end{equation}
For
$\delta(\beta)=e^{-\beta}$ and
$V(\beta)=1/[2\cosh(\beta/2)]$, this simplifies to
\begin{equation}
    \dot\beta_{\mathrm{fr}}(t)
    \propto
    1+e^{-\beta(t)}.
\end{equation}
Integrating, we obtain
\begin{equation}
\label{eq:friction_schedule_explicit}
\beta_{\mathrm{fr}}(t) = \log\left[ \exp\left( (1-t)\log(1+e^{\beta_0}) + t\log(1+e^{\beta_f}) \right) -1 \right].
\end{equation}

\paragraph{Local tracking schedule.}
Finally, we consider the schedule obtained from the local tracking metric \cref{lemma:effective_metric}. A constant-speed parametrization satisfies
\begin{equation}
\label{eq:tracking_schedule_velocity}
\dot\beta_{\mathrm{track}}(t) = \alpha_{\mathrm{track}} \frac{\delta(\beta(t))}{(1-\delta(\beta(t)))V(\beta(t))}.
\end{equation}
For the present two-state model, we have
\begin{equation}
\sqrt{g_{\beta\beta}^{\mathrm{track}}(\beta)}
= e^{\beta/2}\tanh\!\left(\frac{\beta}{2}\right),
\end{equation}
so introducing
\begin{equation}
F_{\mathrm{track}}(\beta) = 2e^{\beta/2} - 4\arctan\!\left(e^{\beta/2}\right),
\end{equation}
the constant-speed condition becomes
\begin{equation}
F_{\mathrm{track}}(\beta(t)) = (1-t)F_{\mathrm{track}}(\beta_0) + tF_{\mathrm{track}}(\beta_f).
\end{equation}
Thus,
\begin{equation}
\label{eq:tracking_schedule_explicit}
\beta_{\mathrm{track}}(t) = F_{\mathrm{track}}^{-1} \left[ (1-t)F_{\mathrm{track}}(\beta_0) + tF_{\mathrm{track}}(\beta_f)\right],
\end{equation}
where the inverse is evaluated numerically. Notice that at $\beta_0=0$, the tracking metric vanishes for the two-state model. However, the schedule \cref{eq:tracking_schedule_explicit} remains continuous, but behaves as $\beta(t)\sim 2\sqrt{\alpha_{\text{track}}t}$ near $t=0$, and hence is not differentiable there. Thus, \cref{lemma:effective_metric} does not directly apply to this schedule on the full interval, nevertheless, the discrete adiabatic theorem applies to the sampled sequence, and the numerical results remain valid.

These schedules allow us to compare different notions of optimality. The variance-only schedule accounts exclusively for the statistical variation of the Gibbs state, the gap-only schedule isolates the inverse-linear dependence on the relaxation gap, and the friction schedule optimizes the constant-excess-power criterion. In contrast, the local
tracking schedule is constructed to equalize the leading local adiabatic
tracking difficulty,
\begin{equation}
    \frac{(1-\delta(\beta))V(\beta)}
    {\delta(\beta)}
    |\dot\beta|,
\end{equation}
and is therefore directly connected to the discrete tracking theorem.

\subsection{Markov dynamics and schedules constructions}\label{app:schedu_sk}

Before we show how we extract the schedule in the experiments for the SK instance, we explain how the Markov dynamics is constructed in our numerical testing.

\subsubsection{Construction of the SK Markov dynamics}

For each SK realization, the state space contains all $2^n$ spin configurations, $\Omega=\{-1,+1\}^n$. It is convenient numerically to label each configuration by an $n$-bit string
\begin{equation}
    x=(b_1,\ldots,b_n),
    \qquad
    b_i\in\{0,1\},
\end{equation}
with the correspondence $\sigma_i=1-2b_i$.  Thus, every integer $x\in\{0,\ldots,2^n-1\}$ labels one spin configuration.

The Markov dynamics is constructed from single-spin-flip Metropolis updates. From a configuration $x$, a proposal consists of choosing one of the $n$ spins uniformly and flipping it. In the binary representation this corresponds to flipping a single bit,
\begin{equation}
    y=x\oplus 2^{i-1},
\end{equation}
where $\oplus$ denotes the bitwise exclusive OR operation. Consequently, two configurations are directly connected by the Markov chain if and only if their Hamming distance is one, $d_{\mathrm H}(x,y)=1$. Each configuration therefore has exactly $n$ possible single-spin-flip neighbours.

For the SK energy used in the classical schedule benchmark,
\begin{equation}
E_r(\sigma)=-2\sum_{1\leq i<j\leq n}J_{ij}^{(r)}\sigma_i\sigma_j,
\end{equation}
the energy difference associated with a proposed move $x\rightarrow y$ is
\begin{equation}
    \Delta E_{yx}
    =
    E_r(y)-E_r(x).
\end{equation}
The Metropolis acceptance probability at inverse temperature $\beta$ is
\begin{equation}
A_\beta(y,x) = \min\left\{1,
    e^{-\beta\Delta E_{yx}}
    \right\}.
\end{equation}
Thus, energy-lowering moves are always accepted, whereas energy-increasing moves are accepted with probability
$e^{-\beta\Delta E_{yx}}$. With the convention that probability distributions are column vectors, the non-lazy Metropolis transition matrix has off-diagonal elements
\begin{equation}
    \mathcal W_{yx}(\beta)
    =
    \frac{1}{n}
    A_\beta(y,x),
    \qquad
    d_{\mathrm H}(x,y)=1,
\end{equation}
and
\begin{equation}
    \mathcal W_{yx}(\beta)=0,
    \qquad
    d_{\mathrm H}(x,y)>1.
\end{equation}
The probability of remaining in the same configuration is determined by
probability conservation,
\begin{equation}
    \mathcal W_{xx}(\beta)
    =
    1-
    \sum_{y\neq x}
    \mathcal W_{yx}(\beta).
\end{equation}
The diagonal element therefore contains the probability associated with
rejected spin-flip proposals.

In the numerical results reported below we use the lazy version of the chain,
\begin{equation}
\mathcal W^{\mathrm L}_\beta = \frac{1}{2}I+
\frac{1}{2}\mathcal W_\beta .
\label{eq:sk_lazy_kernel}
\end{equation}
Equivalently, at each Markov step there is probability $1/2$ of remaining at the current configuration without attempting a move, while with probability
$1/2$ a standard single-spin Metropolis update is performed. Lazification removes the period-two behaviour that can occur for the single-spin-flip walk on the hypercube, particularly at $\beta=0$, while preserving the Gibbs stationary distribution.

For every value of $\beta$, the Gibbs distribution satisfies detailed balance,
\begin{equation}
    \pi_\beta(x)
    \mathcal W_{yx}(\beta)
    =
    \pi_\beta(y)
    \mathcal W_{xy}(\beta),
\end{equation}
and hence $\mathcal W_\beta\pi_\beta = \pi_\beta$. The same stationary distribution is preserved by the lazy kernel.

\subsubsection{Estimation of the spectral gap}

The annealing schedules that use spectral information require the instantaneous gap of the lazy Markov kernel $W^L_\beta$, defined in \cref{eq:SK_lazy_transition}. Using reversibility, we introduce the symmetric similarity transform
\begin{equation}
S_\beta = \operatorname{diag}(\pi_\beta^{-1/2})
\mathcal W^L_\beta \operatorname{diag}(\pi_\beta^{1/2}).
    \label{eq:sk_symmetrized_walk}
\end{equation}
Detailed balance implies $S_\beta=S_\beta^{\mathsf T}$. Since $S_\beta$ and $W^L_\beta$ are similar they have the same eigenvalues. 

Write the eigenvalues of $W^L_\beta$ as $1=\lambda_0(\beta), \lambda_1(\beta), \ldots, \lambda_{2^n-1}(\beta)$. The spectrum is nonnegative because $W^L_\beta=1/2(I+W_\beta)$. Consequently, its absolute spectral gap is 
\begin{equation}
\delta(\beta) = 1- \max_{k\geq1} |\lambda_k(\beta)|.
\label{eq:sk_absolute_gap}
\end{equation}
All spectral gaps used to construct the classical SK schedules refer to this lazy kernel. Lazification maps each eigenvalue $\mu_k=[-1,1]$ of $W_\beta$ to $\lambda_k=(1+\mu_k)/2 \in [0,1]$. This, the absolute gap of the lazy kernel equals $1-\lambda_1(\beta)$.

The extraction of the gap therefore follows the sequence
\begin{equation}
A_\beta(y,x) \;\longrightarrow\; \mathcal W_\beta
    \;\longrightarrow\; \mathcal{W}_\beta^L \longrightarrow
    S_\beta
    \;\longrightarrow\;
    \{\lambda_k(\beta)\}
    \;\longrightarrow\;
    \delta(\beta).
\end{equation}
That is, the Metropolis acceptance probabilities determine the entries of the full Markov matrix, and the spectral gap is then obtained from the spectrum of that matrix.

The matrix $\mathcal W_\beta$ is sparse: each column contains at most $n$ off-diagonal single-spin-flip transitions, corresponding to the $n$ configurations at Hamming distance one, together with one diagonal entry accounting for rejected moves and, for the lazy chain, the additional probability of remaining in the same state. We therefore construct $\mathcal W_\beta$ directly in sparse form.

For the smaller state spaces, the complete spectrum of the symmetric matrix $S_\beta$ is obtained by direct diagonalization. For the larger systems, only the extremal eigenvalues required to evaluate \cref{eq:sk_absolute_gap} are computed using a sparse symmetric eigensolver. The stationary eigenvalue closest to one is identified and removed, after which $ \max_{k\geq1}|\lambda_k(\beta)|$ is determined from the remaining spectrum. The gap is evaluated on a grid of inverse temperatures along $\beta\in[\beta_0,\beta_f]$. Together with the exactly computed Gibbs energy standard deviation, $V(\beta)$, the resulting behaviors of $\Delta(\beta)$ and $V(\beta)$ provide the local information required to construct the annealing schedules described next.

\subsubsection{Construction of the annealing schedules}

We construct the annealing schedules from quantities evaluated along the inverse-temperature interval $\beta\in[\beta_0,\beta_f]$. For each SK disorder realization, the Gibbs distribution $\pi_\beta(x)$ is evaluated exactly on a grid of inverse temperatures. From this distribution we compute the energy standard deviation $ V(\beta)$, together with the instantaneous absolute spectral gap of the lazy Metropolis kernel, $\delta(\beta)$

Each annealing prescription specifies a local rate
\begin{equation}
\frac{d\beta}{ds}, \qquad s\in[0,1],
\end{equation}
which determines how fast the inverse temperature is varied along the normalized evolution parameter $s$. Rather than integrating this differential equation directly, we define the associated schedule density as the inverse local speed,
\begin{equation}
    \rho(\beta)
    \propto
    \left(
    \frac{d\beta}{ds}
    \right)^{-1}.
    \label{eq:schedule_density_def}
\end{equation}
The overall proportionality constant is irrelevant, since it is fixed by the boundary conditions $\beta_0$ and $\beta_f$.

The quantity $\rho(\beta)$ therefore measures the amount of evolution parameter allocated per unit change in inverse temperature. A large value of $\rho(\beta)$ corresponds to a small $d\beta/ds$ and hence to a larger number of Markov steps being concentrated in that region of the Gibbs path.

Given a positive schedule density $\rho(\beta)$, we define the cumulative schedule coordinate
\begin{equation}
    A(\beta)
    =
    \int_{\beta_0}^{\beta}
    \rho(u)\,du,
\end{equation}
and normalize it according to
\begin{equation}
    s(\beta)
    =
    \frac{
    A(\beta)
    }{
    A(\beta_f)
    }
    =
    \frac{
    \displaystyle\int_{\beta_0}^{\beta}\rho(u)\,du
    }{
    \displaystyle\int_{\beta_0}^{\beta_f}\rho(u)\,du
    }.
    \label{eq:sk_schedule_coordinate}
\end{equation}
Differentiating this expression gives
\begin{equation}
    \frac{ds}{d\beta}
    =
    \frac{
    \rho(\beta)
    }{
    \displaystyle\int_{\beta_0}^{\beta_f}\rho(u)\,du
    },
\end{equation}
which is equivalent to the original local annealing-rate prescription.

The continuous annealing schedule is obtained by numerically inverting
\cref{eq:sk_schedule_coordinate}, $\beta=\beta(s)$. For an evolution containing $T$ applications of the Markov kernel, the normalized evolution coordinate is sampled uniformly, $s_j=j/T$ for $j=0,\ldots,T$, and the corresponding inverse temperatures are $\beta_j=\beta(j/T)$. The resulting time-inhomogeneous evolution is then $p_{j+1} = \mathcal W_{\beta_{j+1}}p_j$.

Numerically, $\rho(\beta)$ is evaluated on the same inverse-temperature grid used to compute $V(\beta)$ and $\Delta(\beta)$. The cumulative integral $A(\beta)$ is evaluated using the trapezoidal rule, the normalized values $s(\beta)$ are constructed, and this monotone relation is interpolated and inverted to obtain the discrete sequence $\{\beta_j\}_{j=0}^{T}$. Thus, all schedules are implemented through the same numerical procedure; they differ only in the choice of the density $\rho(\beta)$.

\subsection{Variance-only schedule}

The variance-only schedule uses only the local statistical variation of the Gibbs distribution. Since
\begin{equation}
    \|\partial_\beta\pi_\beta\|_{\pi_\beta^{-1}}
    =
    V(\beta),
\end{equation}
we choose
\begin{equation}
    \rho_{\mathrm{var}}(\beta)
    =
    V(\beta).
    \label{eq:variance_schedule_density}
\end{equation}
The resulting schedule satisfies
\begin{equation}
    \dot\beta
    \propto
    \frac{1}{V(\beta)}.
\end{equation}
It therefore slows down in regions where neighboring Gibbs distributions are changing fast, but it does not include information about the relaxation properties of the Markov chain.

\subsection{Gap-only schedule}

To isolate the role of the relaxation timescale, we also consider a schedule constructed only from the spectral gap. The local relaxation time is proportional to the inverse gap, motivating
\begin{equation}
    \rho_{\mathrm{gap}}(\beta)
    =
    \frac{1}{\delta(\beta)}.
    \label{eq:gap_schedule_density}
\end{equation}
Thus,
\begin{equation}
    \dot\beta
    \propto
    \delta(\beta).
\end{equation}
The schedule allocates more steps to small-gap regions where relaxation is slow, but does not account for how fast the stationary Gibbs distribution itself changes.

\subsection{Tracking schedule}

The Tracking schedule is constructed from the effective local tracking metric derived in \cref{lemma:effective_metric}. Its constant-metric-speed parametrization satisfies $\sqrt{g^{\text{track}}_{\beta\beta}\dot{\beta}}=\alpha$, where $\alpha>0$  is fixed by the endpoints and the duration of the normalized schedule. Thus, for an increasing schedule,
\begin{equation}
\dot\beta
=\alpha\,\frac{\delta(\beta)}
{[1-\delta(\beta)]V(\beta)}.    
\end{equation}
Following the schedule-density construction described above, we use
\begin{equation}
\rho_{\mathrm{track}}(\beta)
=\frac{1-\delta(\beta)}{\delta(\beta)}V(\beta).    
\end{equation}
In the small-gap regime, $\delta(\beta)\ll 1$, the factor $1-\delta(\beta)$ is approximately unity, giving
\begin{equation}
  \rho_{\mathrm{track}}(\beta)\simeq\frac{V(\beta)}{\delta(\beta)},
\qquad
\dot\beta
\propto
\frac{\delta(\beta)}
    {V(\beta)}.
\end{equation}
We retain the factor $1-\delta(\beta)$ in the numerical implementation. This factor arises from the contraction of the newly generated tracking error by the Markov step; it remains present in the leading local bound as $T\to\infty$ at fixed gap.

This schedule therefore combines the two local quantities entering the tracking theorem: the energy fluctuations determine how fast the target Gibbs state changes, while the spectral gap determines how fast the Markov dynamics can respond to that change.

\subsection{Friction schedule}

The friction schedule is motivated by thermodynamic geometry, which combines equilibrium energy fluctuations with their relaxation time. For the inverse-temperature Gibbs family, the thermodynamic friction scales as
\begin{equation}
\zeta_{\beta\beta}(\beta) \propto
\operatorname{Var}_{\pi_\beta}(E)
\tau_E(\beta),
\end{equation}
where $\tau_E(\beta)$ is the integrated energy autocorrelation time. For the SK numerical comparasion, we use the inverse spctral gap as a proxy for this correlation time and define
\begin{equation}
\hat{\tau}_E(\beta) \coloneqq
    \frac{1}{\delta(\beta)},    
\end{equation}
\begin{equation}
    \dot\beta
    \propto
    \frac{1}
    {\sqrt{\zeta_{\beta\beta}(\beta)}}.
\end{equation}

In the numerical comparison, the relaxation time is estimated from the spectral gap,
\begin{equation}
    \tau_E(\beta)
    \sim
    \frac{1}{\delta(\beta)}, \quad \hat{\zeta}_{\beta \beta}(\beta) \coloneqq \frac{V(\beta)^2}{\delta(\beta)}.
\end{equation}
The constant-spped condition for this proxy gives $\sqrt{ \hat{\zeta}_{\beta \beta}}\dot{\beta}=\text{constant}$, and therefore
\begin{equation}
\dot\beta
\propto\frac{\sqrt{\delta(\beta)}}{V(\beta)}.
\end{equation}
Equivalently, the schedule density is
\begin{equation}
\rho_{\mathrm{fric}}(\beta)
= \frac{V(\beta)} {\sqrt{\delta(\beta)}}.
\label{eq:friction_schedule_density}
\end{equation}

For fixed endpoints and duration, this schedule minimizes the metric-energy functional constructed from $\hat{\zeta}_{\beta \beta}$. Since the actual energy autocorrelation time need not to be equal to $1/\delta(\beta)$, it should be interpreted as a gap-based friction proxy rather than an exact minimum-dissipation protocol for SK dynamics.

\subsection{Numerical construction}

The quantities $V(\beta)$ and $\Delta(\beta)$ are first evaluated on a grid
\begin{equation}
\beta_0=\beta^{(0)} < \beta^{(1)} < \cdots < \beta^{(M)}=\beta_f.
\end{equation}
For each schedule, the corresponding density $\rho(\beta)$ is evaluated on this grid and integrated numerically using the trapezoidal rule,
\begin{equation}
A_k \simeq \sum_{\ell=0}^{k-1}
\frac{\rho(\beta^{(\ell)}) +  \rho(\beta^{(\ell+1)})}{2}
\left( \beta^{(\ell+1)}  - \beta^{(\ell)} \right).
\end{equation}
The cumulative values are normalized by $A_M$ and the resulting monotone function $s(\beta)$ is numerically inverted to obtain $\beta(s)$. Finally, $\beta_j = \beta\!\left(\frac{j}{T}\right)$ is used in the time-inhomogeneous Markov evolution.

A small numerical floor is applied only when constructing schedules containing inverse powers of the gap, in order to avoid numerical singularities due to finite-precision eigenvalue estimates. The measured spectral gap itself is not modified when reporting spectral quantities or analysing the dynamics.

\section{Spectral-gap extrapolation for classical and quantum schedules}
\label{app:SK_extap}

To investigate the transferability of the classical and quantum schedules to larger SK systems, we estimate the finite-size dependence of the Markov-chain spectral gap using smaller instances. At each inverse temperature $\beta$, we average the gap over $N_{\rm inst}$ independent disorder realizations,
\begin{equation}
\overline{\delta}_n(\beta)
=
\frac{1}{N_{\rm inst}}
\sum_{r=1}^{N_{\rm inst}}
\delta_{n,r}(\beta),
\end{equation}
and fit its dependence on system size according to
\begin{equation}
\log\overline{\delta}_n(\beta)
=
a(\beta)+b(\beta)n.
\end{equation}
For the classical simulations, the fit uses $n=6,8,10,12$ and is extrapolated to $n=18$, whereas the quantum simulations use $n=6,8,10$ to predict the gaps for $n=12,14,16$. The predicted gaps are therefore empirical estimates, rather than rigorous spectral-gap lower bounds.

For the classical $n=18$ simulations, the energy-fluctuation scale $V(\beta)$ is estimated directly from samples generated by independent lazy Metropolis chains. At each inverse temperature, the chains undergo initial relaxation updates, followed by energy measurements. The resulting estimate is
\begin{equation}
\widehat V(\beta_j)
=
\left[
\frac{1}{N_s-1}
\sum_{\ell=1}^{N_s}
\left(
E_j^{(\ell)}
-
\overline E_j
\right)^2
\right]^{1/2},
\qquad
\overline E_j
=
\frac{1}{N_s}
\sum_{\ell=1}^{N_s}E_j^{(\ell)}.
\end{equation}
The final configurations are reused at the next inverse temperature to provide warm starts. The predicted classical Gap-only and Tracking schedules are then constructed from
\begin{equation}
\rho_{\rm gap}^{(18)}(\beta)
=
\frac{1}{\widehat{\delta}_{18}(\beta)},
\qquad
\rho_{\rm track}^{(18)}(\beta)
=
\frac{1-\widehat{\delta}_{18}(\beta)}
{\widehat{\delta}_{18}(\beta)}
\widehat V(\beta).
\label{eq:n18_tracking_density}
\end{equation}

For the quantum simulations, write $\widehat\delta_{L,n}$ for the predicted gap of the lazy kernel. Using $\delta_+=2\delta_L$, the quantum-walk eigenphase-gap estimate is
\begin{equation}
\widehat{\Delta}_{w,n}(\beta)
=
\arccos\left[1-2\widehat{\delta}_{L,n}(\beta)\right].
\end{equation}
The corresponding discrete schedule is constructed using the finite-difference variation $D_{n,i}$, evaluated analytically from \cref{eq:SK_discrete_L1}, together with the predicted eigenphase gaps. Thus, constructing the quantum schedules for $n=12,14,16$ does not require instance-specific spectral diagonalization.

In both cases, the cumulative schedule densities are normalized and inverted following \cref{app:schedu_sk} to obtain the discrete inverse-temperature sequence used in the evolution.

\section{Numerical tests -- supporting material}
\label{sec:supp_mate}

\begin{table*}[t]
\centering
\caption{ Comparison of inverse-temperature schedules for the analytically solvable two-state model using $T=200$ Markov steps, with $\beta_0=0$ and $\beta_f=4$. Here,
$\mathrm{TV}_T=\|p_T-\pi_T\|_{\mathrm{TV}}$ is the final total-variation error,
$e_T=\|p_T-\pi_T\|_{\pi_T^{-1}}$ is the final weighted tracking error, and
$e_{\max}=\max_j\|p_j-\pi_j\|_{\pi_j^{-1}}$
is the maximum instantaneous tracking error.
For this two-state model and monotone schedules, the product-form weighted-norm tracking bound is saturated, on the other hand, the exponential relaxation bound and conversion to total variation need not be saturated. The column $e_T/2$ gives the generic bound
$\|p_T-\pi_T\|_{\mathrm{TV}}\leq e_T/2$.
The Tracking schedule is derived from the local tracking metric and depends
on both the spectral gap and the Gibbs-state energy variance, whereas the Friction schedule is the constant-excess-power schedule obtained from $\zeta_{\beta\beta}\sim
\operatorname{Var}_{\pi_\beta}(E)\tau_E(\beta)$ with
$\tau_E(\beta)=1/\delta(\beta)$.}
\label{tab:two_state_T200}
{\begin{tabular}{lccccc}
\hline\hline
Schedule & $\mathrm{TV}_T$ & $e_T$ & $e_{\max}$ & $e_T/2$
& $\max_j |B_j-e_j|$
\\
\hline
Gap-only
& $6.1655\times10^{-3}$
& $4.6392\times10^{-2}$
& $8.9069\times10^{-2}$
& $2.3196\times10^{-2}$
& $1.11\times10^{-16}$
\\
Exp.\ $a=-3$
& $6.4186\times10^{-3}$
& $4.8296\times10^{-2}$
& $8.7623\times10^{-2}$
& $2.4148\times10^{-2}$
& $1.94\times10^{-16}$
\\
Tracking
& $8.0471\times10^{-3}$
& $6.0550\times10^{-2}$
& $\mathbf{7.3121\times10^{-2}}$
& $3.0275\times10^{-2}$
& $1.25\times10^{-16}$
\\
Power $p=0.5$
& $1.0341\times10^{-2}$
& $7.7808\times10^{-2}$
& $7.7808\times10^{-2}$
& $3.8904\times10^{-2}$
& $1.67\times10^{-16}$
\\
Friction
& $1.5637\times10^{-2}$
& $1.1766\times10^{-1}$
& $1.1766\times10^{-1}$
& $5.8828\times10^{-2}$
& $1.53\times10^{-16}$
\\
Linear
& $1.8068\times10^{-2}$
& $1.3595\times10^{-1}$
& $1.3595\times10^{-1}$
& $6.7976\times10^{-2}$
& $1.39\times10^{-16}$
\\
Variance-only
& $2.9362\times10^{-2}$
& $2.2093\times10^{-1}$
& $2.2093\times10^{-1}$
& $1.1046\times10^{-1}$
& $1.40\times10^{-16}$
\\
Power $p=2$
& $3.0813\times10^{-2}$
& $2.3185\times10^{-1}$
& $2.3185\times10^{-1}$
& $1.1592\times10^{-1}$
& $1.80\times10^{-16}$
\\
Exp.\ $a=3$
& $4.0497\times10^{-2}$
& $3.0471\times10^{-1}$
& $3.0471\times10^{-1}$
& $1.5236\times10^{-1}$
& $1.67\times10^{-16}$
\\
Power $p=3$
& $4.1486\times10^{-2}$
& $3.1215\times10^{-1}$
& $3.1215\times10^{-1}$
& $1.5608\times10^{-1}$
& $2.78\times10^{-16}$
\\
\hline\hline
\end{tabular}
}
\end{table*}

\begin{table*}[t]
\centering
\caption{
Performance of the inverse-temperature schedules for the lazy SK chain using
$T=6400$ Markov steps. Each entry is averaged over the same $56$ disorder realizations for each system size $n$.
Smaller values of $e_{\max}$ and $e^{\mathrm{final}}_{TV}$ indicate better pathwise tracking and final Gibbs-state preparation, respectively, whereas values of $q_{\rm succ}$ and $q_{\rm AEAR}$ closer to unity indicate better
terminal preparation. The Tracking schedule is obtained from the local tracking metric, the
Friction schedule from the scalar-gap friction proxy, and the Gap-only schedule uses only the instantaneous spectral gap.
Bold values indicate the best result within each system size and column. }
\label{tab:sk_equal56_performance}
{ \begin{tabular}{clcccc}
\hline\hline
$n$ & Schedule & $e_{\max}$ & $e^{\mathrm{final}}_{TV}$
& $q_{\rm succ}$ & $q_{\rm AEAR}$
\\
\hline

6
& Linear
& $\mathbf{5.6168\times10^{-3}}$
& $1.1676\times10^{-3}$
& $0.998297$
& $0.999488$
\\
& Variance-only
& $7.0659\times10^{-3}$
& $2.3682\times10^{-3}$
& $0.996632$
& $0.998969$
\\
& Friction
& $6.3417\times10^{-3}$
& $7.3292\times10^{-4}$
& $0.998915$
& $0.999678$
\\
& Gap-only
& $9.9235\times10^{-2}$
& $2.7043\times10^{-4}$
& $0.999594$
& $0.999881$
\\
& Tracking
& $4.0771\times10^{-2}$
& $3.5396\times10^{-4}$
& $0.999471$
& $0.999844$
\\
& Exp.\ $a=-3$
& $9.5686\times10^{-3}$
& $\mathbf{1.8612\times10^{-4}}$
& $\mathbf{0.999728}$
& $\mathbf{0.999919}$
\\
\hline

8
& Linear
& $\mathbf{1.3675\times10^{-2}}$
& $2.5367\times10^{-3}$
& $0.996227$
& $0.999308$
\\
& Variance-only
& $1.9533\times10^{-2}$
& $5.7780\times10^{-3}$
& $0.991612$
& $0.998415$
\\
& Friction
& $1.8817\times10^{-2}$
& $1.5167\times10^{-3}$
& $0.997707$
& $0.999590$
\\
& Gap-only
& $5.5326\times10^{-1}$
& $5.2804\times10^{-4}$
& $0.999192$
& $0.999858$
\\
& Tracking
& $2.0561\times10^{-1}$
& $7.0234\times10^{-4}$
& $0.998928$
& $0.999811$
\\
& Exp.\ $a=-3$
& $2.4006\times10^{-2}$
& $\mathbf{4.1731\times10^{-4}}$
& $\mathbf{0.999379}$
& $\mathbf{0.999887}$
\\
\hline

10
& Linear
& $\mathbf{3.2756\times10^{-2}}$
& $6.8202\times10^{-3}$
& $0.990646$
& $0.998732$
\\
& Variance-only
& $6.0243\times10^{-2}$
& $1.6327\times10^{-2}$
& $0.977672$
& $0.996904$
\\
& Friction
& $4.4745\times10^{-2}$
& $3.6065\times10^{-3}$
& $0.995066$
& $0.999348$
\\
& Gap-only
& $1.5183\times10^{0}$
& $\mathbf{1.2532\times10^{-3}}$
& $\mathbf{0.998283}$
& $\mathbf{0.999779}$
\\
& Tracking
& $5.1878\times10^{-1}$
& $1.6414\times10^{-3}$
& $0.997760$
& $0.999710$
\\
& Exp.\ $a=-3$
& $5.0346\times10^{-2}$
& $1.2979\times10^{-3}$
& $0.998242$
& $0.999768$
\\
\hline

12
& Linear
& $\mathbf{6.4267\times10^{-2}}$
& $1.1995\times10^{-2}$
& $0.984378$
& $0.998417$
\\
& Variance-only
& $1.2214\times10^{-1}$
& $2.7732\times10^{-2}$
& $0.963393$
& $0.996034$
\\
& Friction
& $8.2633\times10^{-2}$
& $6.7686\times10^{-3}$
& $0.991203$
& $0.999146$
\\
& Gap-only
& $3.5707\times10^{0}$
& $\mathbf{2.6769\times10^{-3}}$
& $\mathbf{0.996541}$
& $\mathbf{0.999669}$
\\
& Tracking
& $9.7862\times10^{-1}$
& $3.4498\times10^{-3}$
& $0.995539$
& $0.999571$
\\
& Exp.\ $a=-3$
& $9.3918\times10^{-2}$
& $3.0682\times10^{-3}$
& $0.996149$
& $0.999626$
\\

\hline\hline
\end{tabular}
}
\end{table*}

Table~\ref{tab:two_state_T200} first summarizes the numerical results for the analytically solvable two-state model discussed in the main text. For this two-state model, monotone schedules and equilibrium initialization ensure that the error and stationary-distribution increment are aligned. Consequently, the product-form weighted-norm tracking bound is saturated at every step up to numerical precision. The table compares several inverse-temperature schedules at fixed $T=200$ using the final total-variation error $\mathrm{TV}_T$, the final weighted tracking error $e_T$, and the maximum instantaneous tracking error $e_{\max}$. The results already show that pathwise and terminal performance are distinct objectives: the Tracking schedule gives the smallest $e_{\max}$ among the schedules derived from the local geometric quantities, while Gap-only and the decreasing exponential schedule yield smaller final errors.

Table~\ref{tab:sk_equal56_performance} then extends the same comparison to the many-state SK model. It reports the average performance at fixed $T=6400$ over the same $56$ disorder realizations for each system size. The quantity $e_{\max}$ measures the largest departure from the instantaneous Gibbs distribution along the entire annealing path, whereas $e^{\mathrm{final}}_{TV}$, $q_{\rm succ}$, and $q_{\rm AEAR}$ quantify the quality of the final prepared state.

The SK results reinforce the distinction already visible in the two-state model. Among the schedules retained in the comparison, Linear gives the smallest $e_{\max}$ for all system sizes shown, indicating the strongest control of the full path. This does not, however, lead to the best terminal preparation. Schedules that place greater weight on the spectral gap, particularly Gap-only and Tracking, produce substantially smaller final total-variation errors and values of $q_{\rm succ}$ and $q_{\rm AEAR}$ closer to the exact Gibbs targets.

For $n=6$ and $n=8$, the fixed exponential schedule with $a=-3$ gives the best terminal performance among the schedules considered. For $n=10$ and $n=12$, Gap-only gives the smallest final total-variation error and the largest values of both terminal performance ratios. Tracking follows the same qualitative trend while providing a less aggressive concentration of steps in the small-gap region. Friction yields a more moderate improvement over Linear, whereas Variance-only performs consistently worse, indicating that equilibrium energy fluctuations alone do not capture the dominant dynamical bottlenecks of the SK annealing process.

Taken together, Tables~\ref{tab:two_state_T200} and~\ref{tab:sk_equal56_performance} show the same general feature in increasingly complex settings: minimizing the largest transient deviation from the instantaneous Gibbs state and minimizing the final preparation error are different finite-time objectives. Schedules that devote more of the evolution to slow-relaxation regions can improve the final Gibbs-state preparation even when they tolerate larger intermediate deviations.

\end{document}